\documentclass[12pt,draftcls,peerreview,onecolumn]{IEEEtran}

\usepackage{amsmath,amssymb,epstopdf,cite,subfigure,epsf}
\usepackage{psfrag}
\usepackage{amssymb}
\usepackage{amsmath}
\usepackage{bbm}
\usepackage{dsfont}
\usepackage{pifont}
\usepackage{cite}
\usepackage{graphics}
\usepackage{epsfig}
\usepackage{subfigure}
\usepackage{amscd}
\usepackage{accents}
\usepackage{multirow}
\usepackage{epstopdf}
\newtheorem{theorem}{Theorem}%[section]

\newtheorem{definition}[theorem]{Definition}

\newenvironment{proof}[1][Proof]{\begin{trivlist}
\item[\hskip \labelsep {\bfseries #1}]}{\end{trivlist}}
\newcommand{\matc}[2][ccccccccccccccccccc]{\left[
\begin{array}{#1}
#2\\
\end{array}
\right]}

\newcommand{\RN}[1]{%
  \textup{\uppercase\expandafter{\romannumeral#1}}%
}

\newlength{\dhatheight}

\title{Compressed Sensing for Wireless Communications : Useful Tips and Tricks}

\author{\IEEEauthorblockN{Jun Won Choi$^{*}$, Byonghyo Shim$^{\flat}$, Yacong Ding$^{\natural}$, Bhaskar Rao$^{\natural}$, Dong In Kim$^{\dagger}$ } \\
\IEEEauthorblockA{$^{*}$Hanyang University, Seoul, Korea \\ $^{\flat}$Seoul National University, Seoul, Korea \\ $^{\natural}$University of California at San Diego, CA, USA \\ $^{\dagger}$Sungkyunkwan University, Suwon, Korea }

\thanks{This research was funded by the research grant from the National Research Foundation of Korea (NRF) grant funded by the Korea government (MSIP) (No. 2014R1A5A1011478). }
}

\begin{document}

% make the title area
\maketitle

\begin{abstract}
As a paradigm to recover the sparse signal from a small set of linear measurements, compressed sensing (CS) has stimulated a great deal of interest in recent years. In order to apply the CS techniques to wireless communication systems, there are a number of things to know and also several issues to be considered.
However, it is not easy to come up with simple and easy answers to the issues raised while carrying out research on CS. The main purpose of this paper is to provide essential knowledge and useful tips that wireless communication researchers need to know when designing CS-based wireless systems.
First, we present an overview of the CS technique, including basic setup, sparse recovery algorithm, and performance guarantee. Then, we describe three distinct subproblems of CS, viz., {\it sparse estimation}, {\it support identification}, and {\it sparse detection}, with various wireless communication applications. We also address main issues encountered in the design of CS-based wireless communication systems.
% we go over main issues in a way of answering to seven fundamental questions. In each issue, we provide useful tips, essential knowledge, and benefits and limitations so that readers can catch the gist and thus take advantage of CS techniques.
%
These include potentials and limitations of CS techniques, useful tips that one should be aware of, subtle points that one should pay attention to, and some prior knowledge to achieve better performance.
Our hope is that this article will be a useful guide for wireless communication researchers and even non-experts to grasp the gist of CS techniques.
\end{abstract}


%\IEEEpeerreviewmaketitle

%================================
\section{Introduction}
\label{sec:intro}
%================================
Compressed sensing (CS) is an attractive paradigm to acquire, process, and recover the sparse signals \cite{donoho}. This new paradigm is very competitive alternative to conventional information processing operations including sampling, sensing, compression, estimation, and detection.
Traditional way to acquire and reconstruct analog signals from sampled signal is based on the celebrated Nyquist-Shannon's sampling theorem \cite{shannon_94Rauhut} which states that the sampling rate should be at least twice the bandwidth of
an analog signal to restore it from the discrete samples accurately.
In case of a discrete signal regime, the fundamental theorem of linear algebra states that the number of observations in a linear system should be at least equal to the length of the desired signal to ensure the accurate recovery of the desired signal.
While these fundamental principles always hold true, it might be too stringent in a situation where signals of interest are sparse, meaning that the signals can be represented using a relatively small number of nonzero coefficients.

The CS paradigm provides a new perspective on the way we process the information.
While approaches exploiting the sparsity of signal vector have been used for a long time in image processing and transform coding, this topic has attracted wide attention ever since the works of Donoho \cite{donoho} and Candes, Romberg, and Tao \cite{stable}. At the heart of the CS lies the fact that a sparse signal can be recovered from the underdetermined linear system in a computationally efficient way. In other words, a small number of linear measurements (projections) of the signal contain enough information for its reconstruction.
Main wisdom behind the CS is that essential knowledge in the large dimensional signals is just handful, and thus measurements with the size being proportional to the sparsity level of the input signal are enough to reconstruct the original signal. In fact, in many real-world applications, signals of interest are sparse or can be approximated as a sparse vector in a properly chosen basis. Sparsity of underlying signals simplifies the acquisition process, reduces memory requirement and computational complexity, and further enables to solve the problem which has been believed to be unsolvable.

In the last decade, CS techniques have spread rapidly in many disciplines such as medical imaging, machine learning, computer science, statistics, and many others.
Also, various wireless communication applications exploiting the sparsity of a target signal have been proposed in recent years. Notable examples, among many others, include channel estimation, interference cancellation, direction estimation, spectrum sensing, and symbol detection (see Fig. \ref{fig:bigpicture}).
These days, many tutorials, textbooks, and papers are available\cite{eldar_book,cs_magazine,zhu_han}, but it might not be easy to grasp the essentials and useful tips tailored for wireless communication engineers.
One reason is because the seminal works are highly theoretic and based on harmonic analysis, group theory, random matrix theory, and convex optimization so that it is not easy to catch the gist from these works.
Another reason is that the CS-based wireless communication works are mostly case studies so that one cannot grasp a general idea from these studies. Due to these reasons, CS remains somewhat esoteric and vague field for many wireless communication researchers who want to catch the gist of CS and use it in their applications. Notwithstanding the foregoing, much of the fundamental principle and basic knowledge is simple, intuitive, and easy to understand.

%%%%%%%%%%%%%%%%%%%%%%%%%%%%%%%%%%%%%%%%%%%%%%%%%%%
\begin{figure}[t]
\centerline{\psfig{file=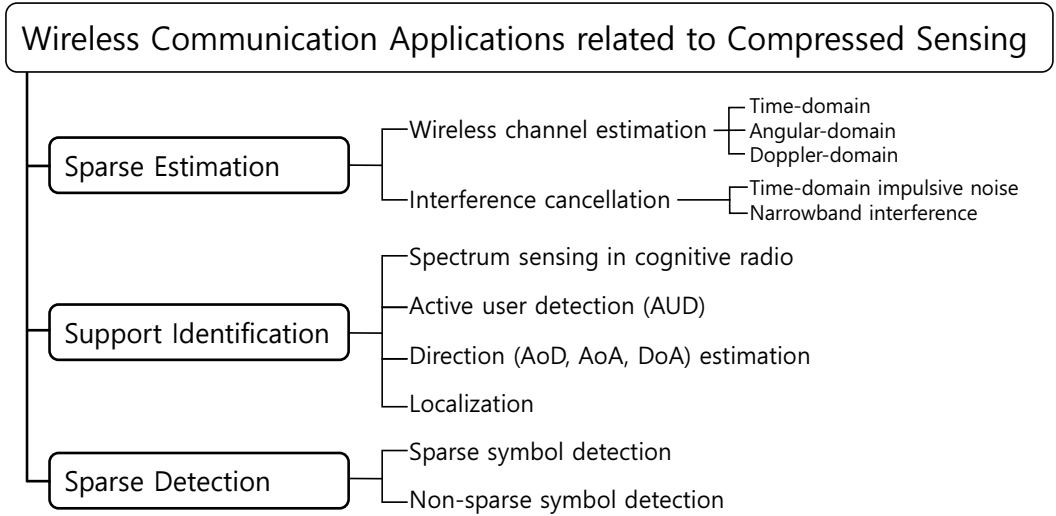,height=70mm,width=140mm}}
\caption{Outline of subproblems of CS related to wireless communications and their examples.}
\label{fig:bigpicture}
\end{figure}
%%%%%%%%%%%%%%%%%%%%%%%%%%%%%%%%%%%%%%%%%%%%%%%%%%%
%
The purpose of this paper is neither to describe the complicated mathematical expressions required for the characterization of the CS, nor to describe the details of state-of-the-art CS techniques and sparse recovery algorithms, but to bridge the gap between the wireless communications and CS principle by providing essentials and useful tips that communication engineers and researchers need to be aware of.
With this purpose in mind, we organized this article as follows.
In Section II, we provide an overview of the CS techniques. We review how to solve the systems with linear equations for both overdetermined and underdetermined systems and then address the scenario where the input vector is sparsely represented in the underdetermined setting. We also discuss the sparse signal recovery algorithm and performance guarantee of the CS technique under which accurate or stable recovery of the sparse signal is ensured.
%can be reconstructed accurately or stable recovered accurately or stablrecovery algorithm
%
%
In Section III, we describe the basic wireless communication system model and then discuss each subproblems of CS related to wireless communications. Depending on the sparse structure of the desired signal vector, CS problems can be divided into three subproblems: {\it sparse estimation}, {\it support identification}, and {\it sparse detection}. We explain the subproblem with the specific wireless communication applications.
Developing successful CS technique for the specific wireless application requires good understanding on key issues.
These include properties of system matrix and input vector, algorithm selection/modification/design, system setup and performance requirements. In Section IV, we go over main issues in a way of answering to seven fundamental questions. In each issue, we provide useful tips, essential knowledge, and benefits and limitations so that readers can catch the gist and thus take advantage of CS techniques.
We conclude the paper in Section V by summarizing the contributions and discussing open issues. Our hope is that this paper would provide better view, insight, and understanding of the potentials and limitations of CS techniques to wireless communication researchers.

% Compressive Sensing for Wireless Networks Zhu Han, Husheng Li, Wotao Yin, cambridge univ. press 2013
%
% R. G. Baraniuk. Compressive Sensing. IEEE SIGNAL Proc. Magazine July 2007
%
In the course of this writing, we observe a large body of researches on CS, among which we briefly summarize some notable tutorial and survey results here. Short summary of CS is presented by Baraniuk \cite{baraniuk_magazine}. Extended summary can be found in Candes and Wakin \cite{cs_magazine}. Forucart and Rauhut provided a tutorial of CS with an emphasis on mathematical properties for performance guarantee \cite{rauhut} and similar approach can be found in \cite{fornasier}. Comprehensive treatment on various issues, such as sparse recovery algorithms, performance guarantee, and CS applications, can be found in the book of Eldar and Kutyniok \cite{eldar_book}. Book of Han, Li, and Yin summarized the CS techniques for wireless network applications\cite{zhu_han} and Hayashi, Nagahara, and Tanaka discussed the applications of CS to the wireless communication systems \cite{hayashi}.

Our work is distinct from previous studies in the following aspects. First, we divide the wireless communication problems into three distinct subproblems (sparse estimation, sparse detection, support identification) and then explain the details of each problem, which could provide more systematic view and also help researchers easily classify the problem under investigation.
Second, we focus more on physical layer issues with various state-of-the-art examples. While examples of previous studies are unclassified, we introduce plentiful examples including channel estimation and impulsive noise cancellation (Section III.B), spectrum sensing, active user detection, and direction estimation (Section III.C), and non-sparse detection (Section III.D) and then classify each into one of subproblems.
From these examples, researchers and practitioners can catch the similarity and disparity between the  examples and their own problem and therefore figure out how the CS technique can be applied to their wireless applications.
We further present useful tips and tricks in a way of answering to seven important issues raised in the wireless system design for which there is no counterpart in previous studies. These include properties of system matrix and target vector, system and algorithm design when multiple measurements are available, selection of sparse recovery algorithm, and how to deal with the constraints and additional hints (Section IV).

%We further provide detailed tips in a way of answering to seven important questions (system matrix, target vector, number of measurements, recovery algorithm, existence of additional constraints and hints) raised in the wireless system design that lead to better system design (Section 4).

%================================
\section{Basics of Compressed Sensing}
\label{sec:basics}
%================================

% =============================================
\subsection{Solutions of Linear Systems}
\label{sec:overdetermined}
% =============================================
We begin with a linear system having $m$ equations and $n$ unknowns given by
\begin{eqnarray}
\mathbf{y} = \mathbf{Hs}  \label{eq:system}
\end{eqnarray}
where $\mathbf{y}$ is the measurement vector, $\mathbf{s}$ is the desired signal vector to be reconstructed, and $\mathbf{H} \in \mathcal{R}^{m\times n}$ is the system matrix. In this case, the measurement vector $\mathbf{y}$ can be expressed as a linear combination of the columns of $\mathbf{H}$, that is, $\mathbf{y} = \sum_i s_i \mathbf{h}_i$ ($s_i$ and $\mathbf{h}_i$ are the $i$-th entry of $\mathbf{s}$ and $i$-th column of $\mathbf{H}$, respectively) so that $\mathbf{y}$ lies in the subspace spanned by the columns of $\mathbf{H}$.

We first consider the scenario where the number of measurements is larger than or equal to the size of unknown vector ($m \geq n$).
In this case, often referred to as overdetermined scenario, one can recover the desired vector $\mathbf{s}$ using a simple algorithm (e.g., Gaussian elimination) as long as the system matrix is a full rank (i.e., $rank(\mathbf{H}) = \min\{ m, n\}$).
Even if this is not the case, one can find an approximate solution minimizing the error vector $\mathbf{e} = \mathbf{y} - \mathbf{Hs}$.
%This procedure is standard and well-known but we briefly provide key steps for the sake of comparison.
%First, we define the and then find the vector
%\textbf{
The conventional approach to recover $\mathbf{s}$ from the measurement vector $\mathbf{y}$ is to find the vector $\mathbf{s}^*$ minimizing the $\ell_2$-norm of the error vector, i.e.,
%------------------------------------------------
\begin{eqnarray}
    \mathbf{s}^* = \arg \min_{\mathbf{s}} \| \mathbf{e} \|_2.
\end{eqnarray}
%------------------------------------------------
%
Since $\| \mathbf{e} \|_2^2 = \mathbf{s}^T \mathbf{H}^T \mathbf{H s} - 2 \mathbf{y}^T \mathbf{H s} + \mathbf{y}^T \mathbf{y}$, by setting the derivative of $\| \mathbf{e} \|_2^2$ with respect to $\mathbf{s}$ to zero, we have
$\frac{ \partial }{ \partial \mathbf{s}} \| \mathbf{e} \|_2^2 = 2 \mathbf{H}^T \mathbf{H s} - 2 \mathbf{H}^T \mathbf{y} = 0$,
and
%------------------------------------------------
\begin{eqnarray}
\mathbf{s}^* = (\mathbf{H}^T \mathbf{H})^{-1} \mathbf{H}^T \mathbf{ y}. \label{eq:lss}
\end{eqnarray}
%------------------------------------------------
The obtained solution $\mathbf{s}^*$ is called least squares (LS) solution and the operator $(\mathbf{H}^T \mathbf{H})^{-1} \mathbf{H}^T$ is called the pseudo inverse and denoted as $\mathbf{H}^{\dagger}$. Note that $\mathbf{H} \mathbf{s}^*$ is closest to the measurement vector $\mathbf{y}$ among all possible points in the range space of $\mathbf{H}$.
% so that $\mathbf{s}^*$ is often referred to as a minimum norm solution.
%
%
%First, if $\mathbf{H}$ is a square matrix (i.e., $m=n$), it is clear that $\mathbf{x} = \mathbf{H}^{-1} \mathbf{y}$.
%Second, if $m < n$, then we solve the unknown vector $\mathbf{x} = \mathbf{H}^{\dagger} \mathbf{y}$,

While finding the solution in an overdetermined scenario is straightforward and fairly accurate in general, the task to recover the input vector in an underdetermined scenario where the measurement size is smaller than the size of unknown vector ($m < n$) is challenging and problematic, since one cannot find out the unique solution in general.
As a simple example, consider the example where $\mathbf{H} = [1 ~ 1]$ and the original vector is $\mathbf{s} = [s_1 ~ s_2]^T = [1 ~ 1]^T$ (and hence $y = 2$).
Since the system equation is $2 = s_1 + s_2$, one can easily observe that there are infinitely many possible solutions satisfying this. This is because for any vector $\mathbf{v} = [v_1 ~ v_2 ]^T$ satisfying $0 = v_1 + v_2$ (e.g., $v_1 = -1$ and $v_2 = 1$), $\mathbf{s}' = \mathbf{s} + \mathbf{v}$ also satisfies $\mathbf{y} = \mathbf{H}\mathbf{s}'$.
Indeed, there are infinitely many vectors in the null space $N (\mathbf{H}) = \{ \mathbf{v} ~|~ \mathbf{H v} = 0\}$ for the underdetermined scenario so that one cannot find out the unique solution satisfying \eqref{eq:system}.
In this scenario, because $\mathbf{H}^T \mathbf{H}$ is not full rank and hence non-invertible, one cannot compute the LS solution in \eqref{eq:lss}. Alternative approach is to find a solution minimizing the $\ell_2$-norm of $\mathbf{s}$ while satisfying $\mathbf{y} = \mathbf{H s}$:
%------------------------------------------------------------
\begin{eqnarray}
\mathbf{s}^* = \arg \min \| \mathbf{s} \|_2 ~\mbox{s.t.}~ \mathbf{y} = \mathbf{H s}. \label{eq:L2}
\end{eqnarray}
%------------------------------------------------------------
%
Using the Lagrangian multiplier method, one can obtain\footnote{By setting derivative of the Lagrangian $L(\mathbf{s}, \lambda) = \| \mathbf{s} \|_2^2 + \lambda^T (\mathbf{y} - \mathbf{H s})$ with respective to $\mathbf{s}$ to zero, we obtain $\mathbf{s}^* = -\frac{1}{2} \mathbf{H}^T \lambda$. Using this together with the constraint $\mathbf{y} = \mathbf{H s}$, we get $\lambda = -2 (\mathbf{H} \mathbf{H}^T)^{-1} \mathbf{y}$ and $\mathbf{s}^* = \mathbf{H}^T (\mathbf{H} \mathbf{H}^T)^{-1} \mathbf{y}$.}
%------------------------------------------------------------
\begin{eqnarray}
\mathbf{s}^* = \mathbf{H}^T ( \mathbf{H} \mathbf{H}^T)^{-1} \mathbf{y}. \label{eq:under}
\end{eqnarray}
%------------------------------------------------------------
%
Since the solution $\mathbf{s}^*$ is a vector satisfying the constraint ($\mathbf{y} = \mathbf{H s}$) with the minimum energy, it is often referred to as the minimum norm solution. Since the system has more unknowns than measurements, the minimum norm solution in \eqref{eq:under} cannot guarantee to recover the original input vector.
This is well-known bad news. However, as we will discuss in the next subsection, sparsity of the input vector provides an important clue to recover the original input vector.

% =============================================
\subsection{Solutions of Underdetermined Systems for Sparse Input Vector}
\label{sec:overdetermined}
% =============================================
%
As mentioned, one cannot find out the unique solution of the underdetermined system since there exist infinitely many solutions.
If one wish to narrow down the choice to convert ill-posed problem into well-posed one, additional hint (side information) is needed.
In fact, the CS technique exploits the fact that the desired signal vector is sparse in finding the solution.
A vector is called {\it sparse} if the number of nonzero entries is sufficiently smaller than the dimension of the vector.
As a metric to check the sparsity, we use $\ell_0$-norm $\| \mathbf{s} \|_0$ of a vector $\mathbf{s}$, which is defined as\footnote{One can alternatively define as $\| \mathbf{s} \|_0 = \lim_{p \rightarrow 0} \| \mathbf{s} \|_p^p =  \lim_{p \rightarrow 0} \sum_i | s_i |^p$}
$$\| \mathbf{s} \|_0 = \# \{ i : s_i \neq 0 \}.$$
%
%Before we define this, recall that a $\ell_p$-norm, formally defined as $\| \mathbf{s} \|_p = \left( \sum_i | s_i |^p \right)^{\frac{1}{p}}$ expresses a total size or length of a vector. However, when it comes to $\ell_0$-norm, this definition is due to the presence of $0$th-power and $0$th-root in it. Thus,
%
%
For example, if $\mathbf{s} = [3 ~ 0 ~ 0 ~ 0 ~ 1 ~ 0]$, then $\| \mathbf{s} \|_0 = 2$. In the simple example we discussed ($2 = s_1 + s_2$), if $\mathbf{s} = [s_1 ~ s_2]$ is sparse, then at least $s_1$ or $s_2$ needs to be zero (i.e., $s_1 = 0$ or $s_2 = 0$). Interestingly, by  invoking the sparsity constraint, the number of possible solutions is dramatically reduced from infinity to two (i.e., $(s_1, s_2) = (2,0)$ or $(0, 2)$).

%%%%%%%%%%%%%%%%%%%%%%%%%%%%%%%%%%%%%%%%%%%%%%%%%%%
\begin{figure}[t]
\centerline{\psfig{file=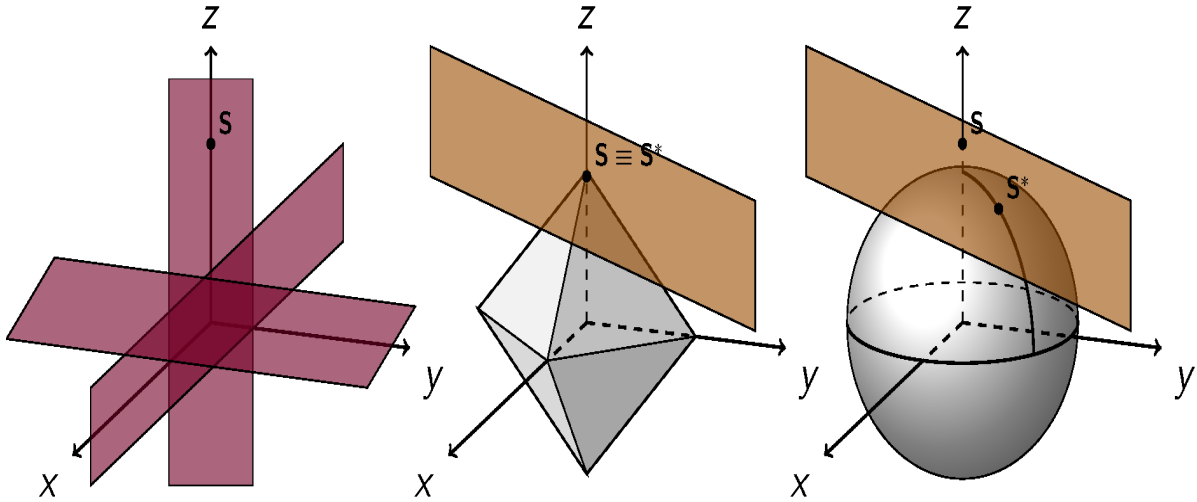,height=70mm,width=160mm}}
\caption{Illustration of $\ell_0$,$\ell_1$, and $\ell_2$-norm minimization approach. If the sparsity of the original vector $\mathbf{s}$ is one, then $\mathbf{s}$ is in the coordinate axes.}
\label{fig:opt_illus}
\end{figure}
%%%%%%%%%%%%%%%%%%%%%%%%%%%%%%%%%%%%%%%%%%%%%%%%%%%

Since the $\ell_0$-norm counts the number of nonzero entries in a vector which is a sparsity promoting function, the problem to find the sparest input vector from the measurement vector is readily expressed as
%
%------------------------------------------------------------
\begin{eqnarray}
\mathbf{s}^* = \arg \min \| \mathbf{s} \|_0 ~\mbox{s.t.}~ \mathbf{y} = \mathbf{H s}. \label{eq:L0}
\end{eqnarray}
%------------------------------------------------------------
%
Since the $\ell_0$-norm counts the number of nonzero elements in $\mathbf{s}$, one should rely on the combinatoric search to get the solution in \eqref{eq:L0}. In other words, all possible subsystems $\mathbf{y} = \mathbf{H}_{\Lambda} \mathbf{s}_{\Lambda}$ should be investigated, where $\mathbf{H}_{\Lambda}$ is the submatrix of $\mathbf{H}$ that contains columns indexed by $\Lambda$.\footnote{For example, if $\Lambda = \{ 1, 3\}$, then $\mathbf{H}_{\Lambda} = [\mathbf{h}_1 ~ \mathbf{h}_3]$.}
Initially, we investigate the solution with the sparsity one by checking $\mathbf{y} = \mathbf{h}_{i} s_{i}$ for each $i$. If the solution is found (i.e., a scalar value $s_i$ satisfying $\mathbf{y} = \mathbf{h}_{i} s_{i}$ is found), then the solution $\mathbf{s}^* = [0 ~\cdots ~0 ~ s_i ~0 ~\cdots ~0]$ is returned and the algorithm is finished. Otherwise, we investigate the solution with the sparsity two by checking if the measurement vector is constructed by a linear combination of two columns of $\mathbf{H}$.
This step is repeated until the solution satisfying  $\mathbf{y} = \mathbf{H}_{\Lambda} \mathbf{s}_{\Lambda}$ is found. Since the complexity of this exhaustive search increases exponentially in $n$, $\ell_0$-norm minimization approach is infeasible for most real-world applications.

Alternative approach suggested by Donoho \cite{donoho} and Candes and Tao \cite{cs_magazine} is $\ell_1$-norm minimization approach given by
%------------------------------------------------------------
\begin{eqnarray}
\mathbf{s}^* = \arg \min \| \mathbf{s} \|_1 ~\mbox{s.t.}~ \mathbf{y} = \mathbf{H s}. \label{eq:L1}
\end{eqnarray}
%------------------------------------------------------------
%
%
While the $\ell_1$-norm minimization problem in \eqref{eq:L1} lies in the middle of \eqref{eq:L0} and \eqref{eq:L2}, it can be cast into the convex optimization problem so that the solution of \eqref{eq:L1} can be obtained by the standard linear programming (LP) \cite{donoho}.
In Fig. \ref{fig:opt_illus}, we illustrate $\ell_0$, $\ell_1$ and $\ell_2$-norm minimization techniques.
If the original vector is sparse (say the sparsity is one), then the desired solution can be found by the $\ell_0$-norm minimization since the points being searched are those in the coordinate axes (sparsity one).
%The set of all $K$-sparse vectors $s$ in $R^N$ is a highly nonlinear space consisting of all $K$-dimensional hyperplanes that are aligned with the coordinate axes as shown in Figure 2(a).
%
%
Since $\ell_1$-norm has a diamond shape (it is in general referred to as cross-polytope), one can observe from Fig. \ref{fig:opt_illus}(b) that the solution of this approach corresponds to the vertex, not the face of the cross-polytope in most cases.
Since the vertex of the diamond lies on the coordinate axes, it is highly likely that the $\ell_1$-norm minimization technique returns the desired sparse solution. In fact, it has been shown that under the mild condition the solution of $\ell_1$-norm minimization problem becomes equivalent to the original vector \cite{cs_magazine}.
Whereas, the solution of the $\ell_2$-norm minimization corresponds to the point closest to the origin among all points $\mathbf{s}$ satisfying $\mathbf{y} = \mathbf{Hs}$ so that the solution has no special reason to be placed at the coordinate axes.
Thus, as depicted in Fig. \ref{fig:opt_illus2}, the $\ell_1$-norm minimization solution often equals the original sparse signal while the $\ell_2$-norm minimization solution does not guarantee this.
%
%
%%%%%%%%%%%%%%%%%%%%%%%%%%%%%%%%%%%%%%%%%%%%%%%%%%%
\begin{figure}[t]
\centerline{\psfig{file=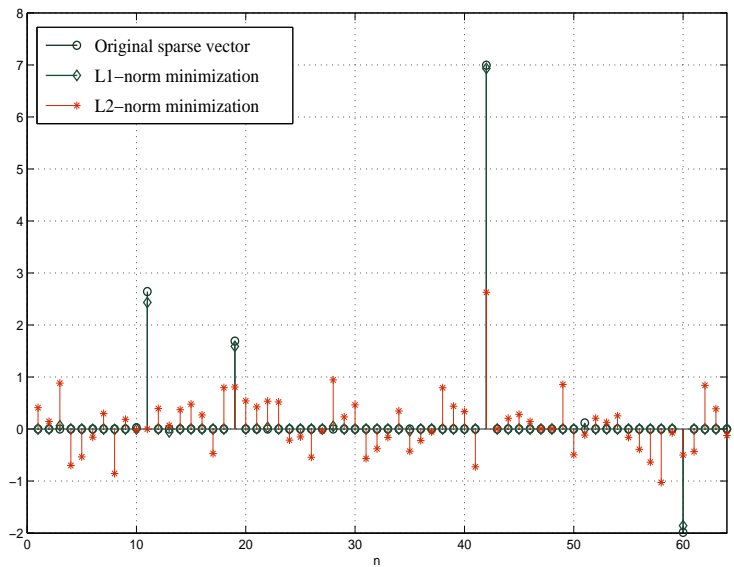,height=85mm,width=120mm}}
\caption{Illustration of performance of $\ell_1$ and $\ell_2$-norm minimization techniques. Entries of the system matrix $\mathbf{H} \in \mathds{R}^{16 \times 64}$ are chosen from standard Gaussian and the sparsity of $\mathbf{s}$ is set to $5$.}
\label{fig:opt_illus2}
\end{figure}
%%%%%%%%%%%%%%%%%%%%%%%%%%%%%%%%%%%%%%%%%%%%%%%%%%%

We also note that when the measurement vector $\mathbf{y}$ is corrupted by the noise, one can modify the equality constraint as
%------------------------------------------------------------
\begin{eqnarray}
\mathbf{s}^* = \arg \min \| \mathbf{s} \|_1 ~\mbox{s.t.}~ \| \mathbf{y} - \mathbf{H s} \|_2 \leq \epsilon \label{eq:L1modify}
\end{eqnarray}
%------------------------------------------------------------
where $\epsilon$ is a (pre-determined) noise level of the channel. This type of problem is often called basis pursuit de-noising (BPDN) \cite{bpdn,uncertainty}.
%Depending on scenarios, the $\ell_2$-norm of the penalty term can be replaced by $\ell_1$-norm or others.
%
This problem has been well-studied subject in convex optimization and there are a number of approaches to solve the problem (e.g., interior-point method \cite{boyd}). Nowadays, there are many optimization packages (e.g., CVX \cite{cvx} or L1-magic \cite{l1magic}) so that one can save the programming effort by using these software tools.

% =============================================
\subsection{Greedy algorithm}
\label{sec:overdetermined}
% =============================================
While the LP technique to solve $\ell_1$-norm minimization problem is effective in reconstructing the sparse vector, it requires computational cost, in particular for large-scale applications. For example, a solver based on the interior point method has an associated computational complexity order of $O(m^2 n^3)$ \cite{donoho}. For many real-time applications including wireless communication applications, therefore, computational cost and time complexity of $\ell_1$-norm minimization solver might be prohibitive.

Over the years, numerous algorithms to recover the sparse signals have been proposed. Notable one among them is a greedy algorithm. By the greedy algorithm, we mean an algorithm to make a local optimal selection at each time with a hope to find the global optimum solution in the end.
Perhaps the most popular greedy algorithm is the orthogonal matching pursuit (OMP) \cite{omp}. In the OMP algorithm, a column of the matrix $\mathbf{H}$ is chosen one at a time using a greedy strategy. Specifically, in each iteration, a column maximally correlated with the (modified) observation is chosen. Obviously, this is not necessarily optimal since the choice does not guarantee to pick the column associated with the nonzero element of $\mathbf{s}$.
Let $\mathbf{h}_{\pi_i}$ be the column chosen in the $i$-th iteration, then the (partial) estimate of $\mathbf{s}$ is  $\hat{\mathbf{s}}_i = \mathbf{H}_i^{\dagger} \mathbf{y}$ and the estimate of $\mathbf{y}$ is $\hat{\mathbf{y}}_i = \mathbf{H}_i \hat{\mathbf{s}}_i = \mathbf{H}_i \mathbf{H}_i^{\dagger} \mathbf{y}$ where $\mathbf{H}_i = [\mathbf{h}_{\pi_1} ~ \mathbf{h}_{\pi_2} ~ \cdots ~ \mathbf{h}_{\pi_i}]$.
% = \mathbf{H}_i \mathbf{H}_i^{\dagger} \mathbf{y}$.
%
%Note that $\hat{\mathbf{y}}_i$ is the projection of the observation $\mathbf{y}$ onto the space spanned by columns selected until $i$-th iteration ($\mathbf{H}_i = [\mathbf{h}_{\pi_1} ~ \mathbf{h}_{\pi_2} ~ \cdots ~ \mathbf{h}_{\pi_i}]$).
%
%Modified observation, denoted by $\mathbf{r}_i$ and called the residual, used in the next iteration is obtained
%
By subtracting $\hat{\mathbf{y}}_i$ from $\mathbf{y}$, we obtain the modified observation $\mathbf{r}_i = \mathbf{y} - \hat{\mathbf{y}}_i$, called the residual, used for the next iteration.
%When $\hat{\mathbf{s}}_i$ is accurate,
By removing the contribution of $\hat{\mathbf{s}}_i$ from the observation vector $\mathbf{y}$ so that we focus on the identification of the rest nonzero elemenets in the next iteration.

One can observe that when the column selection is right, the OMP algorithm can reconstruct the original sparse vector accurately. This is because columns corresponding to the zero element in $\mathbf{s}$ that do not contribute to the observation vector $\mathbf{y}$ can be removed from the system model and as a result the underdetermined system can be converted into overdetermined system (see Fig. \ref{fig:under4}). As mentioned, LS solution for the overdetermined system generates an accurate estimate of the original sparse vector.
%
%
%%%%%%%%%%%%%%%%%%%%%%%%%%%%%%%%%%%%%%%%%%%%%%%%%%%
\begin{figure}[t]
\centerline{\psfig{file=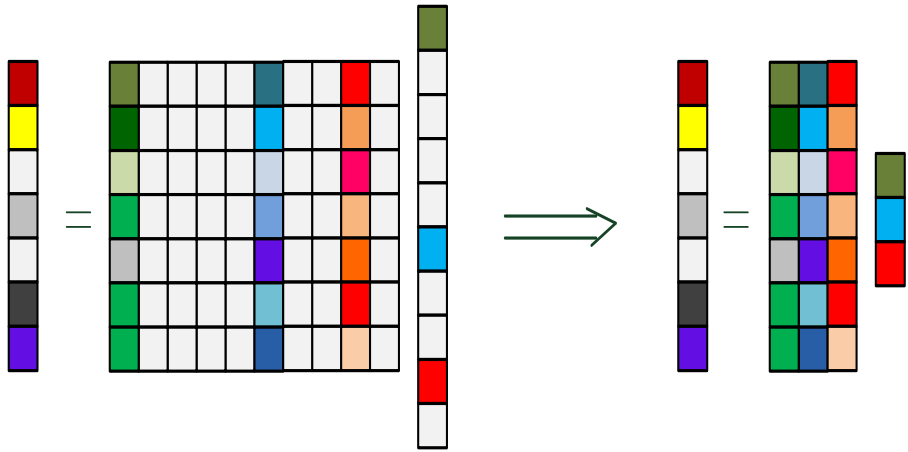,height=60mm,width=120mm}}
\caption{Principle of the greedy algorithm. If the right columns are chosen, then we can convert an underdetermined system into an overdetermined system.}
\label{fig:under4}
\end{figure}
%%%%%%%%%%%%%%%%%%%%%%%%%%%%%%%%%%%%%%%%%%%%%%%%%%%
%
Since the computational complexity is typically much smaller than that of the LP techniques to solve \eqref{eq:L1} or \eqref{eq:L1modify}, the greedy algorithm has received much attention in recent years. Interestingly, even for the simple greedy algorithm like OMP algorithm, recent results show that the recovery performance is comparable to the LP technique while requiring much lower computational overhead \cite{jian16}. We will discuss more on the sparse recovery algorithm in Section \ref{sec:recoveryalg}.
%
%In developing the sparse signal recovery algorithm, there are many issues such as performance, speed, memory requirements, and implementation cost.

% =============================================
\subsection{Performance Guarantee}
\label{sec:overdetermined}
% =============================================
In order to analyze the performance guarantee of the sparse recovery algorithm, many analysis tools have been suggested. For the sake of completeness, we briefly go over some of these tools here.
First, a simple yet intuitive property is the spark of the matrix $\mathbf{H}$. Spark of a matrix $\mathbf{H}$ is defined as the smallest number of columns of $\mathbf{H}$ that are linearly dependent.
%Recall that the rank of a matrix is the largest number of columns in H that are linearly dependent. Spark of a matrix is very difficult to obtain, compared to the rank, as it requires combinatoric search over all possible subsets of columns of $\mathbf{H}$.
From this definition, we see that a vector $\mathbf{v}$ in a null space $N (\mathbf{H}) = \{ \mathbf{v} ~|~ \mathbf{H v} = 0\}$ should satisfy $\| \mathbf{v} \|_0 \geq spark ( \mathbf{H} )$ since a vector $\mathbf{v}$ in the null space linearly combines columns in $\mathbf{H}$ to make the zero vector, and at least $spark ( \mathbf{H} )$ columns are needed to do so.
Following results provide the minimum level of spark over which uniqueness of the $k$-sparse solution is ensured.
%
%%%%%%%%%%%%%%%%%%%%%%%%%%%%%%%%%%%%%%%%%%%%%%%%%%%%%%%%%%%%%%%%%%%%%%%%%%%
\begin{theorem}[Corollary 1 \cite{cs_magazine}]\label{thm:spark}
%If a system of linear equations $y = Hs$ has a solution $s$ obeying $\| \mathbf{s} \|_0 < \frac{ spark ( \mathbf{H} ) }{2}$, then the solution is necessarily the sparsest possible
%
There is at most one $k$-sparse solution for a system of linear equations $\mathbf{y} = \mathbf{Hs}$ if and only if $spark ( \mathbf{H} ) > 2k$.
\end{theorem}
%%%%%%%%%%%%%%%%%%%%%%%%%%%%%%%%%%%%%%%%%%%%%%%%%%%%%%%%%%%%%%%%%%%%%%%%%%%
\begin{proof}
See Appendix A
\end{proof}

From the definition, it is clear that $1 \leq spark ( \mathbf{H} ) \leq n+1$. If entries of $\mathbf{H}$ are i.i.d. random, then no $m$ columns in $\mathbf{H}$ would be linearly dependent with high probability so that $spark ( \mathbf{H} ) = m+1$. Using this, together with Theorem 1, one can conclude that the uniqueness is guaranteed for every solution satisfying $k \leq \frac{m}{2}$.

It is worth mentioning that it is not easy to compute the spark of a matrix since it requires a combinatoric search over all possible subsets of columns in $\mathbf{H}$. Thus, it is preferred to use a property that is easily computable. A tool that meets this purpose is the mutual coherence. The mutual coherence $\mu (\mathbf{H})$ is defined as the largest magnitude of normalized inner product between two distinct columns of $\mathbf{H}$:
\begin{eqnarray}
\mu (\mathbf{H}) = \max_{ i \neq j} \frac{| < \mathbf{h}_i , \mathbf{h}_j > |}{ \| \mathbf{h}_i  \|_2 \| \mathbf{h}_j \|_2  } . \label{eq:mu3}
\end{eqnarray}
In \cite{Strohmer}, it has been shown that for a full rank matrix, $\mu (\mathbf{H})$ satisfies $$1 \geq \mu (\mathbf{H}) \geq \sqrt{ \frac{n-m}{m (n-1)} }.$$ In particular, if $n \gg m$, we obtain an approximate lower bound as $\mu (\mathbf{H}) \geq \frac{1}{\sqrt{m}}$.
It has been shown that $\mu (\mathbf{H})$ is related to $spark ( \mathbf{H} )$ via $spark( \mathbf{H} ) \geq 1 + \frac{1}{\mu (\mathbf{H})}$ \cite{DonohoElad}. Using this together with Theorem 1, we get the following uniqueness condition.
%%%%%%%%%%%%%%%%%%%%%%%%%%%%%%%%%%%%%%%%%%%%%%%%%%%%%%%%%%%%%%%%%%%%%%%%%%%
\begin{theorem}[Corollary 1 \cite{cs_magazine}]\label{thm:spark}
If $k < \frac{1}{2} (1 + \frac{1}{ \mu (\mathbf{H}) })$, then for each measurement vector, there exists at most one $k$-sparse signal $\mathbf{s}$ satisfying $\mathbf{y} = \mathbf{H s}$.
\end{theorem}
%%%%%%%%%%%%%%%%%%%%%%%%%%%%%%%%%%%%%%%%%%%%%%%%%%%%%%%%%%%%%%%%%%%%%%%%%%%

While the mutual coherence is relatively easy to compute, the bound obtained from this is  too strict in general. These days, restricted isometry property (RIP), introduced by Candes and Tao \cite{stable}, has been used popularly to establish the performance guarantee.
%The RIP is defined as follows:
%
\begin{definition}
A system matrix $\mathbf{H}$ is said to satisfy RIP if for all $K$-sparse vector $\mathbf{s}$, the following condition holds
\begin{align}
(1-\delta) \|\mathbf{s}\|_2^2 \leq \|\mathbf{H}\mathbf{s}\|_2^2 \leq (1+\delta) \|\mathbf{s}\|_2^2 .
\end{align}
\end{definition}
In particular, the smallest $\delta$, denoted as $\delta_k$ is referred to as a RIP constant.
In essence, $\delta_k$ indicates how well the system matrix preserves the energy of the original signal.
On one hand, if $\delta_k \approx 0$, the system matrix is close to orthonormal so that the reconstruction of $\mathbf{s}$ would be guaranteed almost surely with a simple matching filtering operation (e.g., $\hat{\mathbf{s}} = \mathbf{H}^H \mathbf{y}$).
On the other hand, if $\delta_k \approx 1$, it might be possible that $\| \mathbf{H}\mathbf{s} \|_2^2 \approx \mathbf{0}$ (i.e., $\mathbf{s}$ is in the nullspace of $\mathbf{H}$) so that the measurements $\mathbf{y} = \mathbf{H}\mathbf{s}$ might not preserve any information on $\mathbf{s}$. In this case, obviously, the recovery of $\mathbf{s}$ would be nearly impossible.
%
%
%For example, well-known result of the $\ell_1$-norm minimization approach is that the stable recovery (exact recovery for noiseless scenario) of $k$-sparse signals is ensured when the system matrix $\mathbf{H}$ satisfies $\delta_{2k} (\mathbf{H}) < \sqrt{2}-1$ \cite{cs_magazine}.

Note that RIP is useful to analyze performance when the measurements are contaminated by the noise \cite{cs_magazine,gomp,sp,tzhang}. Additionally, by the help of random matrix theory, one can perform probabilistic analysis when the entries of the system matrix are i.i.d. random. Specifically, it has been shown that many random matrices (e.g., random Gaussian, Bernoulli, and partial Fourier matrices) satisfy the RIP with exponentially high probability, when the number of measurements scales linearly in the sparsity level \cite{eldar_book}. As a well-known example, if $\delta_{2k} < \sqrt{2} - 1$, then the solution in \eqref{eq:L1} obeys
\begin{eqnarray}
\| \mathbf{s}^* - \mathbf{s} \|_2 &\le& C_0  \| \mathbf{s} - \mathbf{s}_k\|_1/\sqrt{k} \\
\| \mathbf{s}^* - \mathbf{s} \|_1 &\le& C_0  \| \mathbf{s} - \mathbf{s}_k\|_1
\end{eqnarray}
for some constant $C_0$, where $\mathbf{s}_{k}$ is the vector $\mathbf{s}$ with all but the largest $k$ components set to 0.
One can easily see that if $\mathbf{s}$ is $k$-sparse, then $\mathbf{s} = \mathbf{s}_{k}$, and thus the recovery is exact.
If $\mathbf{s}$ is not $k$-sparse, then quality of recovery is limited by the difference of the true signal $\mathbf{s}$ and
its best $k$ approximation $\mathbf{s}_{k}$. For a signal which is not exact sparse but can be well approximated by a $k$-sparse signal (i.e., $\| \mathbf{s} - \mathbf{s}_{k} \|_1$ is small), we can still achieve fairly good recovery performance.

While the performance guarantees obtained by RIP or other tools provide a simple characterization of system parameters (number of measurements, system matrix, algorithm) of the recovery algorithm, these results need to be taken with a grain of salt, in particular when designing the practical wireless systems. This is because the performance guarantee, expressed as a sufficient condition, might be too stringent and working in asymptotic sense in many cases. Also, some of them are, from the wireless communications perspective, based on impractical assumptions (e.g., Gaussianity of the system matrix, strict sparsity of input vector).
Furthermore, it is very difficult to check whether the system setup satisfies the recovery condition or not.\footnote{For example, one need to check ${n \choose 2k}$ submatrices of $\mathbf{H}$ to identify the RIP constant $\delta_{2k}$.}
%
%
%Due to these, analytic results often do not shed light on the performance estimation of real applications and designer needs to check the performance via proof-of-concept simulations.

% =============================================
%
%
%
\section{Compressed Sensing for Wireless Communications}
\label{sec:cswc}
%
%
%
% =============================================

% =============================================
\subsection{System Model}
\label{sec:system_model}
% =============================================
In this section, we describe three distinct CS subproblems related to the wireless communications: sparse estimation, support identification, and sparse detection.
We begin with the basic system model where the signal of interest is transmitted over linear channels with additive white Gaussian noise (AWGN).
The input-output relationship in this model is
%
%========================================
\begin{eqnarray}
\mathbf{y} = \mathbf{Hs} + \mathbf{v},  \label{eq:basic_model}
\end{eqnarray}
%========================================
%
where $\mathbf{y}$ is the vector of received signals, $\mathbf{H} \in \mathcal{C}^{m\times n}$ is the system matrix,\footnote{In the compressed sensing literatures, $\mathbf{y}$ and $\mathbf{H}$ are referred to as measurement vector and sensing matrix (or measurement matrix), respectively.} $\mathbf{s}$ is the desired signal vector we want to recover, and $\mathbf{v}$ is the noise vector ($\mathbf{v} \sim \mathcal{CN}(\mathbf{0}, \sigma^2\mathbf{I})$). In this article, we are primarily interested in the scenario where the desired vector $\mathbf{s}$ is sparse, meaning that the portion of nonzero entries in $\mathbf{s}$ is far smaller than its dimension.
It is worth mentioning that even when the desired vector is non-sparse, one can either approximate it to a sparse vector or convert it to the sparse vector using a proper transform. For example, when the magnitude of nonzero elements is small, we can obtain an approximately sparse vector by ignoring negligible nonzero elements. For example, if $\mathbf{s} = [2 ~0 ~0 ~0 ~0 ~3 ~0.1 ~0.05 ~0.01 ~0]^T$, then we can approximate it to $2$-sparse vector $\mathbf{s}' = [2 ~0 ~0 ~0 ~0 ~3 ~0 ~0 ~0 ~0]^T$. In this case, the effective system model would be $\mathbf{y} = \mathbf{H} \mathbf{s}' + \mathbf{v}'$ where $\mathbf{v}' = \mathbf{H}_{\nu} \mathbf{s}_{\nu} + \mathbf{v}$ ($\mathbf{H}_{\nu} = [\mathbf{h}_7 ~\mathbf{h}_8 ~\mathbf{h}_9]$ and $\mathbf{s}_{\nu} = [0.1 ~0.05 ~0.01]^T$).
Also, even in the case where the desired vector is not sparse, one might choose a basis $\{ \psi_i \}$ to express the signal as a linear combination of basis. In the image/video processing society, for example, discrete Fourier transform (DFT), discrete cosine transform (DCT), and wavelet transform have long been used. Using a properly chosen basis matrix $\mathbf{\Psi} = [\psi_1 ~ \cdots ~ \psi_n]$, the input vector can be expressed as $\mathbf{s} = \sum_{i=1}^{n} x_i \psi_i = \mathbf{\Psi x}$ and thus
%
%========================================
\begin{eqnarray}
\mathbf{y} = \mathbf{Hs} + \mathbf{v} = \mathbf{H \Psi x} + \mathbf{v}, \label{eq:new_model}
\end{eqnarray}
%========================================
%
where $\mathbf{x}$ is a representation of $\mathbf{s}$ in $\mathbf{\Psi}$ domain. By the proper choice of the basis, one can convert the original non-sparse vector $\mathbf{s}$ into the sparse vector $\mathbf{x}$. Since this new representation does not change the system model, in the sequel, we will use a standard model in \eqref{eq:basic_model}.
Depending on the way the desired vector is constructed, the CS-related problem can be classified into three distinctive subproblems.
In the following subsections, we will discuss one by one.

%%%%%%%%%%%%%%%%%%%%%%%%%%%%%%%%%%%%%%%%%%%%%%%%%%%
\begin{figure}[t]
\centerline{\psfig{file=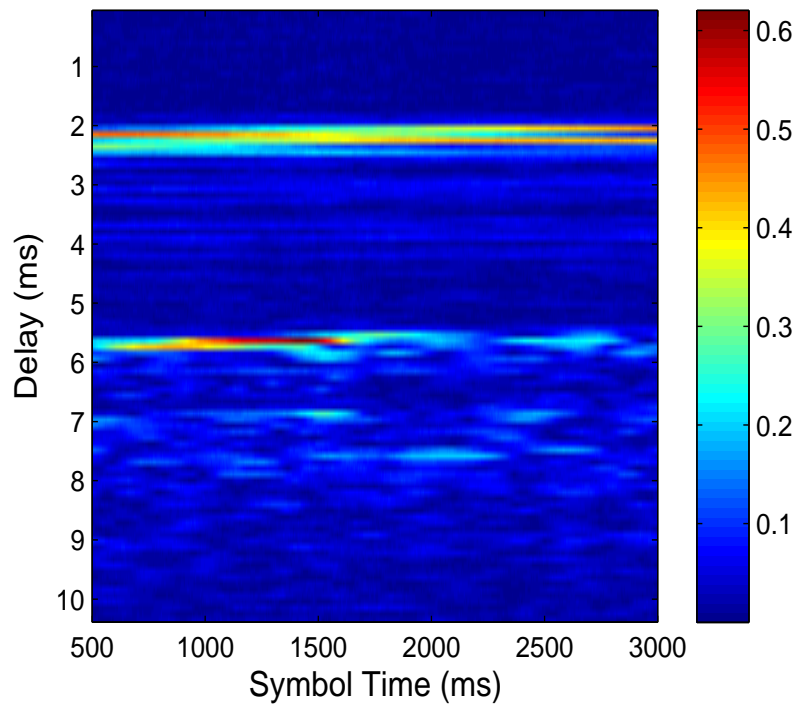,height=100mm,width=120mm}}
\caption{A record of the channel impulse response (CIR) of underwater acoustic channels measured at the coast of Martha's vinyard, MA, USA. At the given time, CIR can be well approximated as a sparse vector.}
\label{fig:sparse_ch}
\end{figure}
%%%%%%%%%%%%%%%%%%%%%%%%%%%%%%%%%%%%%%%%%%%%%%%%%%%

% =============================================
\subsection{Sparse Estimation}
\label{sec:sparse_est}
% =============================================

    In wireless communications, traditional way to estimate the information vector $\mathbf{s}$ from the model $\mathbf{y}=\mathbf{H}\mathbf{s}+\mathbf{v}$ is the linear minimum mean square (LMMSE) estimation, i.e.,
\begin{align}
\hat{\mathbf{s}} &= \mathbf{R}_s \mathbf{H}^{H} \left(\mathbf{H}\mathbf{R}_s \mathbf{H}^{H} + \mathbf{R}_v \mathbf{I} \right)^{-1}  \mathbf{y}, \\
& = \left(\mathbf{R}_s^{-1} + \mathbf{H}^{H} \mathbf{R}_v \mathbf{H}   \right)^{-1}\mathbf{H}^{H}  \mathbf{R}_v^{-1} \mathbf{y}
\end{align}
where $\mathbf{R}_s$ and $\mathbf{R}_v$ are the covariance matrix of  $\mathbf{s}$ and $\mathbf{v}$, respectively.
%
%where we assume that the vector $\mathbf{s}$ has zero mean.
However, when the size of measurement vector $\mathbf{y}$ is smaller than the size of the vector $\mathbf{s}$, the system matrix $\mathbf{H}$ is rank-deficient and the quality of the LMMSE solution degrades significantly.
     When the signal vector is sparse, the problem to recover $\mathbf{s}$ from $\mathbf{y}$ is classified into {\it sparse estimation} problem and the CS technique becomes effective means to recover the target vector $\mathbf{s}$.
    %Sparse estimation problem is popular and often regarded as a synonym of sparse signal recovery problem. %
% 주변과 harmonize가 잘 안되는듯?
% x는 무엇인가?
%

    %
    %----------------------------------------
     \subsubsection{Channel Estimation}
    %----------------------------------------
    %Channel estimation is a typical example of the sparse estimation.
    %In many wireless channels, such as ultra-wideband (UWB), underwater acoustic (UWA), or millimeter wave (mmWave) channels, delay spread is larger than the number of significant paths and hence the channel vector can be well approximated as a sparse signal \cite{cs_uwa,berger,bajwa2010compressed,cotter2002sparse,prasad2014joint} (e.g., see Fig. \ref{fig:sparse_ch}). Even for the cellular environment (e.g., extended vehicular-A (EVA) or extended typical urban (ETU) channel model in LTE systems\cite{LTE_standard}), time-domain channel impulse responses (CIR) are well modeled as a sparse vector since only a few channel paths are dominant.
    %
    %  \textbf{======================YC==========================================}

    Channel estimation is a typical example of the sparse estimation. It has been shown from various measurement campaigns that the scattering effect of the environment is limited in many wireless channels. In other words, propagation paths tend to be clustered within a small spread and there are only a few scattering clusters.
    Thus, the communication channel can be sparsely represented in the delay-Doppler domain or the angular domain.
    In the delay-Doppler domain channel model, if the maximal delay and Doppler spreads are large and there are only few dominant paths, the channel can be readily modeled as a sparse vector \cite{berger_mag}. Typical examples include the ultra-wideband  channel \cite{paredes2007ultra,michelusi2012uwb1,michelusi2012uwb2}, underwater acoustic channel \cite{cs_uwa,berger,bajwa2010compressed}, and even some cellular channels (e.g., extended vehicular-A (EVA) or extended typical urban (ETU) channel model in long term evolution (LTE) systems\cite{LTE_standard}).
    %, see Fig. \ref{fig:sparse_ch} for an example
    %
    %(Is this plot comes from our own work? Otherwise we need to ask for permission of using the plots from publisher and original author, and cite the original work).
    %
    In the massive multi-input multi-output (MIMO) \cite{dist, prior} or millimeter wave \cite{el2014spatially, mmwave} systems where the transmitter and/or receiver are equipped with large antenna array,  due to the limited number of scattering clusters and increased spatial resolvability, the channel can be sparsely represented in the angular domain  \cite{sayeed2002deconstructing, Tse_book, Van_book}.
   Also, recent experiments performed on millimeter wave channels \cite{akdeniz2014millimeter, samimi20163} reveal the limited scattering clusters in angular domain.
    %In Fig.\ref{fig:mmWave} \cite{akdeniz2014millimeter}, (This plots is from \cite{akdeniz2014millimeter}, I think we need to ask for permission form IEEE and the original author. Otherwise please delete it. )it shows the power angular profile measured at 28 GHz, where the circles represented the detected path cluster centers. The plots shows that there are only exist a few path clusters that contribute large received power, see \cite{akdeniz2014millimeter} for more detailed explanations.

%    %%%%%%%%%%%%%%%%%%%%%%%%%%%%%%%%%%%%%%%%%%%%%%%%%%%
%    \begin{figure}[t]
%    \centerline{\psfig{file=mmWave.eps,height=100mm,width=120mm}}
%    \caption{Received power angular profile measured at 28 GHz. Colors represent average received power in dBm. White areas indicate angular offsets that were either not measured or had too low power. The circles represent the detected path cluster centers.}
%    \label{fig:mmWave}
%    \end{figure}
%%%%%%%%%%%%%%%%%%%%%%%%%%%%%%%%%%%%%%%%%%%%%%%%%%%%

     When the channel can be sparsely represented,  CS-based channel estimation can offer significant performance gain over the conventional approaches.
     %System model in (\ref{eq:basic_model}) is adequate when the channel is sparsely represented in delay-doppler domain and that in (\ref{eq:new_model}) is relevant when the channel is sparse in the angular domain.
     %
     For example, we consider single carrier systems where the channel can be modeled as a \(m\)-tapped channel impulse response $\mathbf{h} = [h_0 ~ h_1 ~ \cdots ~ h_{n-1}]^T$. The received signal $\mathbf{y} = [y_0 ~ y_1 ~ \cdots ~ y_{m-1}]^T$   is expressed as a linear convolution of the channel impulse response $\mathbf{h}$  and the pilot sequence \(\mathbf{p}=[p_0 ~ p_1 ~ \cdots ~ p_{m-n-2}]^{T} \), i.e. $y_t = \sum_{i=0}^{n-1} {h}_i * {p}_{t-i}$ for $0\leq t \leq m-1$. By constructing the system matrix using the known pilot sequence, one can express the relationship between $\mathbf{y}$ and $\mathbf{h}$ using a matrix-vector form
     \begin{align}
     \mathbf{y} = \mathbf{P h} + \mathbf{v}
     \end{align}
     where
    % $\mathbf{y} = [y_0 ~ y_1 ~ \cdots ~ y_{n-1}]^T$ is the   vectorized received signal during \(n\) training duration and $\mathbf{P}$ is stacked by pilot sequences \(\mathbf{p}_n \) generating the Toeplitz matrix structure.
    $\mathbf{P}$ is the Toeplitz matrix constructed from the pilot symbols \(\mathbf{p}=[p_0 ~ p_1 ~ \cdots ~ p_{m-n-2}]^{T} \).
%    \begin{align}
%    \mathbf{P} = \begin{bmatrix} p_0 & 0 & \cdots & \cdots & \cdots & \cdots & 0 \\
%    p_1 & p_0 & 0 & \cdots & \cdots & \cdots & 0 \\
%    \vdots & \ddots & \ddots & \ddots & \cdots & \cdots & 0 \\
% p_{n-m-2} & \cdots & p_1 & p_0 & 0 & \cdots & 0 \\
% 0 & \ddots & \ddots& \ddots& \ddots & \cdots & \cdots \\
% \vdots & \ddots & \ddots& \ddots& \ddots & \cdots & 0 \\
% 0 & \cdots & 0 & p_{n-m-2} & \ddots & p_1 & p_0
%    \end{bmatrix}
%    \end{align}
    Whereas, in the orthogonal frequency division multiplexing (OFDM) systems, we allocate the pilot symbols \(\mathbf{p}_f=[p_0 ~ p_1 ~ \cdots ~ p_{m}]^{T} \) in the frequency domain so that the linear convolution operation is converted to  element-by-element product
        between $\mathbf{p}_f$ and the frequency-domain channel response
  \cite{meng,ding,hlawatsch,chan_choi}
      \begin{align}
      \mathbf{y} &= {\rm diag}(\mathbf{p}_f) \mathbf{\Phi} \mathbf{D}\mathbf{h} + \mathbf{v} \\
      &= \mathbf{P}\mathbf{h} + \mathbf{v}
      \end{align}
      where $\mathbf{D}$ is the $n \times n$ DFT matrix which plays a role to convert the time-domain channel response $\mathbf{h}$ to frequency domain response, and $\mathbf{\Phi}$ is the $m \times n$ row selection matrix determining the location of pilots in frequency domain.
          For both cases, the number of nonzero taps in $\mathbf{h}$ is small due to the limited number of scattering clusters, and their positions are unknown. Note that this problem has the same form as (\ref{eq:basic_model}) and CS techniques are effective in recovering $\mathbf{h}$ from the measurements $\mathbf{y}$ \cite{cotter2002sparse,prasad}.

    As an example to exploit the sparsity in the angular domain, we consider  a downlink massive MIMO system where the base station is equipped with  \(n\) antennas and the user has only single antenna. The user  estimates the channel vector from the $n$ base-station antennas to itself. For this purpose, the base-station transmits the training pilots from $n$ antennas over $m$ time slots. Under the narrowband block fading channel model, the pilot measurement acquired by the user over $m$ time slots can be expressed as \(\mathbf{y}=\mathbf{A}\mathbf{h}+\mathbf{v}\), where \(\mathbf{h}\in\mathbb{C}^{n\times 1}\) is the downlink channel vector, \(\mathbf{A}\in\mathbb{C}^{n\times m}\) is the matrix containing the training pilots, and  \(\mathbf{y}\) and \(\mathbf{v}\) are the received signal and noise vectors, respectively.
     When the number of antennas \(n\) is large and there are only limited number of scattering clusters in environment, we can use the angular domain channel representation \cite{sayeed2002deconstructing, Tse_book, bajwa2010compressed}, i.e.,  \(\mathbf{h}=\mathbf{D}\mathbf{s}\), where \(\mathbf{D}\) is the $n \times n$ DFT matrix (whose columns correspond to the steering vectors corresponding to $n$  equally spaced directions) and \(\mathbf{s}\) is the angular domain coefficients with only a few number of nonzeros. Then, the received downlink training signal can be written as \(\mathbf{y}=\mathbf{A} \mathbf{h}+\mathbf{n}=\mathbf{A}\mathbf{Ds}+\mathbf{n}\) which has the same form as  (\ref{eq:new_model}).
     After obtaining the estimate $\hat{\mathbf{s}}$ using the CS technique, we can reconstruct the channel from \(\hat{\mathbf{h}}=\mathbf{D}\hat{\mathbf{s}}\).

One clear benefit of the CS-based channel estimation is to achieve the reduction in the training overhead. In the above examples, the size $m$ of the measurement vector  indicates the amount of time and frequency resource needed to transmit pilot symbols. With conventional channel estimation schemes such as the LMMSE estimator, $m$ should exceed $n$, leading to substantial training overhead.
            %    for \(\mathbf{y=Ph+v}\) the dimension of \(\mathbf{P}\) is \(T\times N\), where \(N\) represents the maximal delay spread and \(T\) is training duration. For angular domain downlink training \(\mathbf{y=ADx+n}\), \(\mathbf{AD}\) has dimension \(T\times N\) where \(N\) denotes the number of antennas and \(T\) is the training duration. To recover \(\mathbf{h}\) from \(\mathbf{y}\) and \(\mathbf{P}\) (or \(\mathbf{AD}\)) using conventional channel estimation such as LS or LMMSE estimator, it is required that \(\mathbf{P}\) needs to be a tall matrix, i.e., $T\geq N$.
%       %having more rows than columns.
%    In order to do so, the training duration needs to be larger than the maximal delay spread or the number of antennas, thus requiring substantial training overhead.  For the environment that has a long delay spread, or the Massive MIMO system that deploys a large antenna array, too much training resources are needed for the conventional channel estimation.
    %
    However, by employing the CS techniques, it only requires \(m\) to be proportional to the number of nonzero elements in the sparse vector. %(it is better to give some simple introduction in section 2)
 %  Assuming that the number of nonzero elements is $k$
%    %\(\|\mathbf{h}\|_0=s\) or
%    \(\|\mathbf{x}\|_0=s\).
%    In \cite{cs_magazine}, it has been shown that $\mathbf{x}$ can be estimated reliably via the CS technique as long as \(T\propto C\cdot s \text{log}(\frac{N}{s})\), where \(C\) is some constant depending on the system matrix.
    %, the \(l_1\) recovery algorithm can robustly estimate \(\mathbf{\beta}\).
    %
    In other words, when the channel can be sparsely represented, the channel can be acquired by the CS technique using much smaller amount of resources than what conventional methods require \cite{gao,lau,gao2,naffouri,chan_choi}.
    %Since this number is much less than the dimension $n$ of the channel vector, we achieve great reduction in the pilot training overhead .
    %\textbf{=====================YC=======================================}

    In  Fig.~\ref{fig:mmWave}, we compare the performance of the CS algorithm (OMP) and the conventional LMMSE technique for channel estimation in massive MIMO systems. We set the number of antennas $n$ to $64$,  the training overhead $m$ to $16$, and the number of multipaths $k$ is  4.
  As shown in the figure, the CS algorithm outperform the conventional LMMSE channel estimator by a large margin except low signal to noise power ratio (SNR) regime.
\begin{figure}[h]
\begin{center}
\hspace{-3mm}\includegraphics[width=110mm,height=95mm]{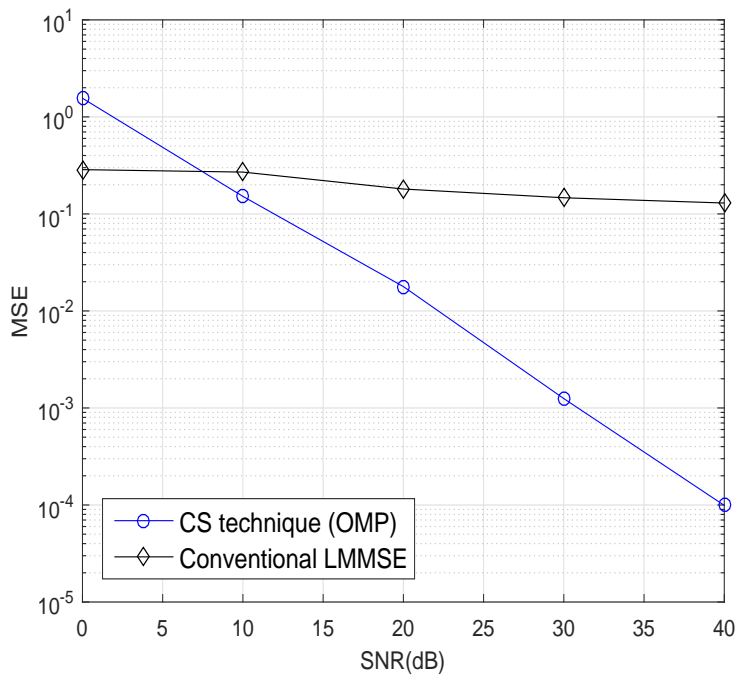}
\caption{Performance of the CS technique and the conventional LMMSE for  angular domain
channel estimation in massive MIMO systems.}
\label{fig:mmWave}
\end{center}
\end{figure}

    \vspace{0.5cm}
    \subsubsection{Impulse Noise Cancellation in OFDM Systems}
    %----------------------------------------
    While OFDM is well suited for frequency selective channel with Gaussian noise, when the unwanted impulsive noise is added, the performance would be degraded severely. In fact, since the impulse in the time domain corresponds to the constant in the frequency domain, very strong time domain impulses will give negative impact to most of frequency domain symbols.
    Since the span of impulse noise is short in time and thus can be considered as a sparse vector, we can use the CS technique to mitigate this noise \cite{Tareq_Al_Naffouri11,ana}. First, the discrete time complex baseband equivalent channel model for the OFDM signal is expressed as
    %
    %--------------------------------------
    \begin{eqnarray}
        \mathbf{y} = \mathbf{H} \mathbf{x}_t + \mathbf{n} \label{eq:basic_ofdm}
    \end{eqnarray}
    %--------------------------------------
    %
    where $\mathbf{y}$ and $\mathbf{x}_t$ are the time-domain receive and transmit signal blocks (after the cyclic prefix removal), $\mathbf{H}$ is the circulant matrix generated by the cyclic prefix, and $\mathbf{n}$ is additive Gaussian noise vector.
    When the impulse noise $\mathbf{e}$ is added, the received signal vector becomes
    %
    %--------------------------------------
    \begin{eqnarray}
        \mathbf{y} = \mathbf{H} \mathbf{x}_t + \mathbf{e} + \mathbf{n}. \label{eq:basic_ofdm}
    \end{eqnarray}
    %--------------------------------------
    %

    Note that the circulant matrix $\mathbf{H}$ can be eigen-decomposed by the DFT matrix $\mathbf{F}$, i.e., $\mathbf{H} = \mathbf{F}^H \mathbf{\Lambda} \mathbf{F}$ \cite{Moon}. Also, the time-domain transmit signal in OFDM systems is expressed as
    $\mathbf{x}_t = \mathbf{F}^H \mathbf{x}_f$.
    where $\mathbf{x}_f$ is the frequency-domain symbol vector.
    Let $q$ be the number of subcarriers in which symbols are being transmitted, then the relationship between $\mathbf{x}_f$ and the true (nonzero) symbol vector $\mathbf{s}$ with dimension $q (\leq n)$ is $$\mathbf{x}_f = \mathbf{\Pi} \mathbf{s}$$ where $\mathbf{\Pi}$ is $n \times q$ selection matrix containing only one element being one in each column and rest being zero.
    %, and $\mathbf{s}$ is frequency-domain symbol vector of dimension $q \leq n$.
    With these, \eqref{eq:basic_ofdm} can be rewritten as
    %
    %--------------------------------------
    \begin{eqnarray}
        \mathbf{y} &=&
        (\mathbf{F}^H \mathbf{\Lambda} \mathbf{F}) \mathbf{x}_t + \mathbf{e} + \mathbf{n} \nonumber \\ &=&
        (\mathbf{F}^H \mathbf{\Lambda} \mathbf{F}) (\mathbf{F}^H \mathbf{\Pi} \mathbf{s}) + \mathbf{e} + \mathbf{n} \nonumber \\ &=& \mathbf{F}^H \mathbf{\Lambda} \mathbf{\Pi} \mathbf{s} + \mathbf{e} + \mathbf{n}.
    \end{eqnarray}
    %--------------------------------------

    Let $\mathbf{y}'$ be the received vector after the DFT operation ($\mathbf{y}' = \mathbf{F y}$). Then, we have
    %
    %--------------------------------------
    \begin{eqnarray}
    \mathbf{y}' = \mathbf{\Lambda} \mathbf{\Pi} \mathbf{s} + \mathbf{F e} + \mathbf{n}' \label{eq:in4}
    \end{eqnarray}
    %--------------------------------------
    %
    where $\mathbf{n}' = \mathbf{F n}$ is also Gaussian having the same statistic of $\mathbf{n}$.
    In removing the impulse noise, we use the subcarriers free of modulation symbols. By projecting $\mathbf{y}'$ onto the space where symbol is not assigned (i.e., orthogonal complement of the signal subspace), we obtain\footnote{As a simple example, if $\mathbf{F}$ is $4\times 4$ DFT matrix and the first and third subcarrier is being used $\mathbf{s} = \matc{s_1 \\ s_3}$, then the selection matrix is $\mathbf{\Pi} = \matc{1 & 0  \\ 0 & 0 \\ 0 & 1  \\ 0 & 0 }$ and the projection operator is $\mathbf{P} = \matc{0 & 0 & 0 & 0 \\ 0 & 1 & 0 & 0 \\ 0 & 0 & 0 & 0   \\ 0 & 0 & 0 & 1  }$. Since the $1$st and $3$rd elements of $\mathbf{\Lambda} \mathbf{\Pi} \mathbf{s}$ are nonzero and $2$nd and $4$th columns contain nonzero elements, we have $\mathbf{P} \mathbf{\Lambda} \mathbf{\Pi} \mathbf{s} = \mathbf{0}$.}
    %
    %%%%%%%%%%%%%%%%%%%%%%%%%%%%%%%%%%%%%%%%%%%%%%%%%%
\begin{figure}[t]
\begin{center}
%\ifCLASSOPTIONonecolumn
	\includegraphics[width=110mm,height=95mm]{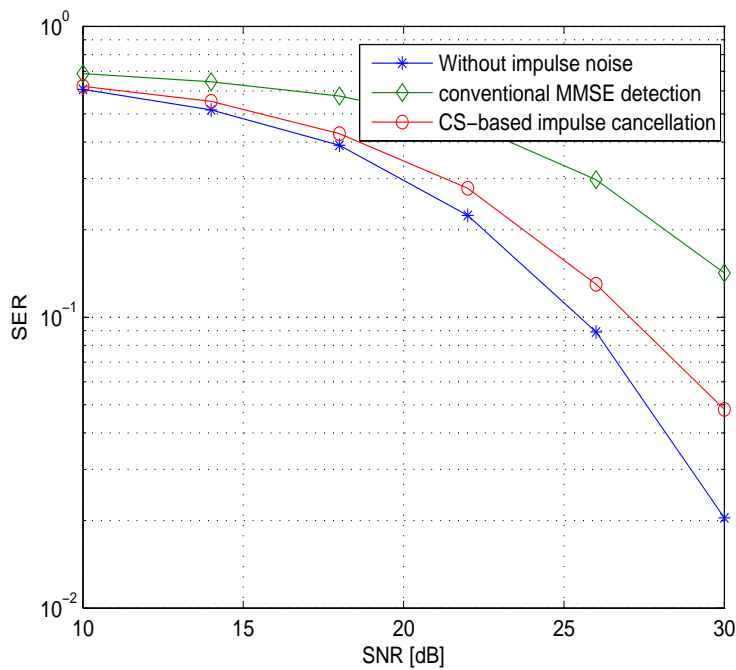}
\caption{Performance of OFDM systems with impulsive noise ($n = 64$, $q = 12$). Span of impulse is $2$ among $64$ samples. and the impulse noise power is set to $20$ dB higher than the average sample power. Symbols are generated by the quadrature phase shift keying (QPSK) modulation and the OMP algorithm is used for the sparse impulse noise recovery.}
\label{fig:cs_impulse}
\end{center}
\end{figure}
%%%%%%%%%%%%%%%%%%%%%%%%%%%%%%%%%%%%%%%%%%%%%%%%%%
    %--------------------------------------
    \begin{eqnarray}
    \mathbf{y}'' = \mathbf{P} \mathbf{y}' = \mathbf{P} \mathbf{F e} + \mathbf{v} \label{eq:in5}
    \end{eqnarray}
    %--------------------------------------
    where $\mathbf{v} = \mathbf{P} \mathbf{n}'$ is the sub-sampled noise vector.
    Note that $\mathbf{y}''$ is a projection of $n$-dimensional impulse noise vector onto a subspace of dimension $m (\ll n)$.
    In the CS perspective, $\mathbf{y}''$ and $\mathbf{P F}$ are the observation vector and the sensing matrix, respectively, and hence the task is to estimate $\mathbf{e}$ from $\mathbf{y}''$.
    %
    %Once the impulse noise estimate $\hat{\mathbf{e}}$ is generated from the CS technique, then we
    %
    Then, by subtracting $\mathbf{F} \hat{\mathbf{e}}$ from the received vector $\mathbf{y}'$, we obtain the modified received vector
    %--------------------------------------
    \begin{eqnarray}
    \hat{\mathbf{y}}' &=& \mathbf{y}' -  \mathbf{F} \hat{\mathbf{e}} \nonumber \\
                      &=& \mathbf{\Lambda} \mathbf{\Pi} \mathbf{s} + \mathbf{F} ( \mathbf{e} - \hat{\mathbf{e}} ) + \mathbf{n}' . \label{eq:in6}
    \end{eqnarray}
    %--------------------------------------
    %Due to the reduction of impulsive noise, we can achieve the improved detection performance.
    As show in Fig. \ref{fig:cs_impulse}, as a result of the impulsive noise cancellation, we obtain the noise mitigated observation and improved detection performance.
    %can detect symbol with less noisy observation vector.

% =============================================
\subsection{Support Identification}
\label{sec:support_identify}
% =============================================
    Set of indices corresponding to nonzero elements in $\mathbf{s}$ is called the support $\Omega_{\mathbf{s}}$ of $\mathbf{s}$\footnote{If $s = [~0 ~0 ~1 ~0 ~2]$, then $\Omega_{\mathbf{s}} = \{3, 5\}$.} and the problem to identify the support is called support identification problem. In some applications, support identification and sparse signal recovery are performed simultaneously. In other case, signal estimation is performed after the support identification. Support identification is important as a sub-step for the sparse estimation algorithm. It is also useful in its own right when an accurate estimation of nonzero values is unnecessary.

    %----------------------------------------
    \subsubsection{Spectrum Sensing}
    %----------------------------------------
    As a means to improve the overall spectrum efficiency, cognitive radio (CR) has received much attention recently \cite{cr_mag}. CR technique offers a new way of exploiting {\it temporarily} available spectrum. Specifically, when a primary user (license holder) does not use a spectrum, a secondary user may access it in such a way that they do not cause interference to primary users.
    Clearly, key to the success of the CR technology is the accurate sensing of the spectrum (whether the spectrum is empty or used by a primary user) so that secondary users can safely use the spectrum without hindering the operation of primary users.
    Future CR systems should have a capability to scan a wideband of frequencies, say in the order of a few GHz. In this case, design and implementation of high-speed analog to digital converter (ADC) become a challenge since the Nyquist rate might exceed the sampling rate of the state-of-the-art ADC devices, not to mention the huge power consumption. One can therefore think of an option of scanning each narrowband spectrum using the conventional technique. However, conventional approach is also undesirable since it takes too much time to process the whole spectrum (if done in sequential manner) or it is too expensive in terms of cost, power consumption, and implementation complexity (if done in parallel).
    %
    %
%%%%%%%%%%%%%%%%%%%%%%%%%%%%%%%%%%%%%%%%%%%%%%%%%%
\begin{figure}[t]
\begin{center}
%\ifCLASSOPTIONonecolumn
	\includegraphics[width=100mm]{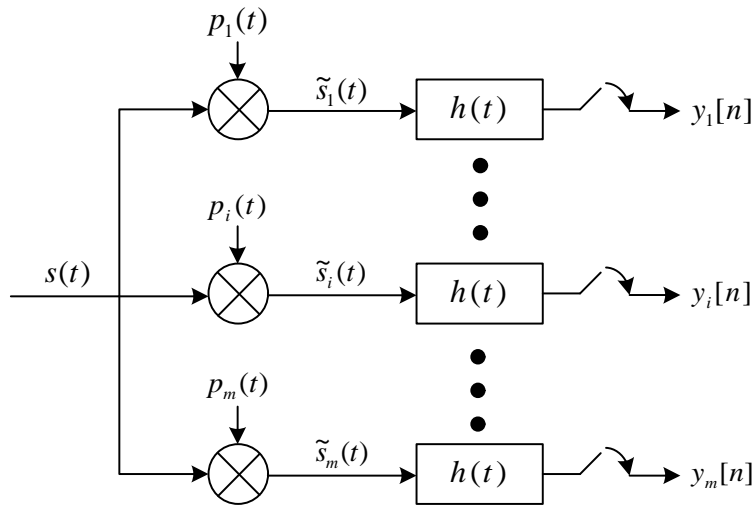}
\caption{Block diagram of modulated wideband converter for spectrum sensing.}
\label{fig:sparse_det}
\end{center}
\end{figure}
%%%%%%%%%%%%%%%%%%%%%%%%%%%%%%%%%%%%%%%%%%%%%%%%%%

    Recently, CS-based spectrum sensing technique has received much attention for its potential to alleviate the sampling rate issue of ADC and the cost issue of RF circuitry. In the CS perspective, the spectrum sensing problem can be translated into the problem to find the nonzero position of vector, which is often referred to as the {\it support identification} or {\it model selection} problem.
    One popular approach of the CS-based spectrum sensing problem, called modulated wideband converter (MWC), is formulated as follows \cite{eldar,geert}. First, we multiply a pseudo random function $p(t)$ with period $T_p$ to the time domain continuous signal $s(t)$. Since $p(t)$ is a periodic function, it can be represented as a Fourier series ($p(t) = \sum_k c_k e^{j 2\pi k / T_p}$),
    Let $s(f)$ be the frequency domain representation of $s(t)$, then the Fourier transform of the modulated signal $\tilde{s}(t) = p(t) s(t)$ is expressed as
    %
    %--------------------------------------------------
    \begin{eqnarray}
        \tilde{s}(f) = \sum_{k=-\infty}^{\infty} c_k s (f - k f_p) , \label{eq:mwc1}
    \end{eqnarray}
    %--------------------------------------------------
    %
    where $f_p = 1/T_p$.
    %
    % Jinhong, can you add one sentence?
    % The reason to perform the random modulation is to xxx. All information is now in the low frequency regime in an aliased form.
    %
    The low-pass filtered version $\tilde{s}'(f)$ will be expressed as $\tilde{s}'(f) = \sum_{k=-L}^{L} c_k s (f - k f_p)$.
    Denoting $y[n]$ as the discretized sequence of $\tilde{s}'(t)$ after the sampling (with rate $T_s$), we obtain the frequency domain relationship as\footnote{If $u[n] = w(t)$ at $t = n T_s$, then $u(e^{j\Omega}) = w(f)$ where $\Omega = 2 \pi f T_s$. }
    %
    %--------------------------------------------------
    \begin{eqnarray}
        y( e^{j 2 \pi f T_s}) = \sum_{k = -L}^{L} c_k s (f - kf_p).  \label{eq:mwc2}
    \end{eqnarray}
    %--------------------------------------------------
    %
    When this operation is done in parallel for different modulating functions $p_i (t)$ ($i = 1, 2, \cdots, m$), we have multiple measurements $y_i( e^{j 2 \pi f T_s})$. After stacking these, we obtain $$\mathbf{y} = [ y_1 ( e^{j 2 \pi f T_s}) ~ \cdots ~ y_m ( e^{j 2 \pi f T_s}) ]^T$$ and the corresponding matrix-vector form $\mathbf{y} = \mathbf{H s}$ where $\mathbf{s} = [s(f - Lf_p) ~ \cdots ~ s(f + L f_p) ]^T$ and $\mathbf{H}$ is the sensing matrix relating $\mathbf{y}$ and $\mathbf{s}$. %($\mathbf{H}_{i,j} = c_{xxx}$).
    Since large portion of the spectrum band is empty, $\mathbf{s}$ can be readily modeled as a sparse vector, and the task is summarized as a problem to find $\mathbf{s}$ from $\mathbf{y} = \mathbf{H} \mathbf{s}$.
    It is worth mentioning that this problem is distinct from the sparse estimation problem since an accurate estimation of nonzero values is unnecessary. Recalling that the main purpose of the spectrum sensing is to identify the empty band, not the occupied one, it would not be a serious problem to slightly increase the false alarm probability (by false alarm we mean that the spectrum is empty but decided as an occupied one). However, special attention should be paid to avoid the misdetection event since the penalty would be severe when the occupied spectrum is falsely declared to be an empty one.
%
%
%
%%%%%%%%%%%%%%%%%%%%%%%%%%%%%%%%%%%%%%%%%%%%%%%%%%
\begin{figure}[t]
\begin{center}
%\ifCLASSOPTIONonecolumn
	\includegraphics[width=110mm,height=95mm]{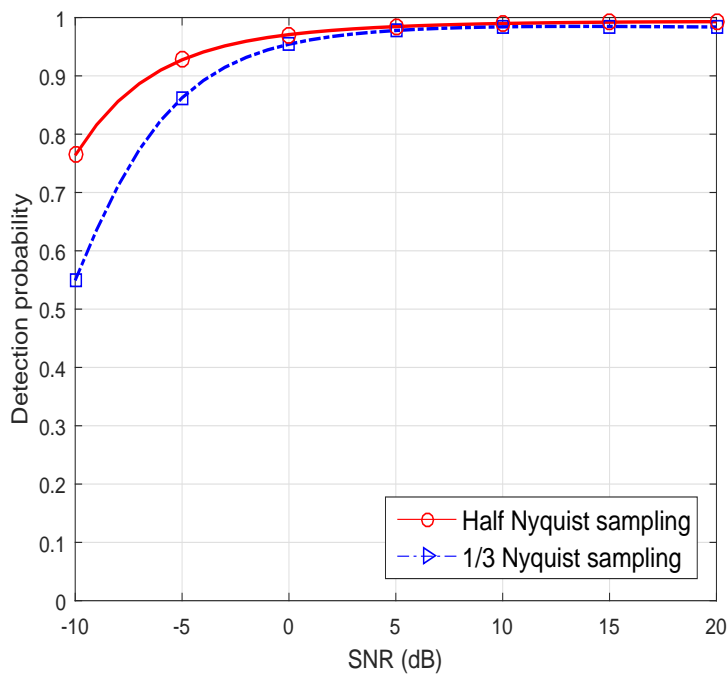}
\caption{Probability of correct support detection in the CS-based spectrum sensing. The signal occupies only 3 bands, each of width 50MHz, in a wide spectrum $5$GHz.}
%The Gaussian noise is added and scaled so that the signal has the desired signal-to-noise ratio (SNR).}
%Misdetection and false alarm performance of CS-based spectrum sensing techniques.}
\label{fig:sparse_det}
\end{center}
\end{figure}
%%%%%%%%%%%%%%%%%%%%%%%%%%%%%%%%%%%%%%%%%%%%%%%%%%

% =============================================
\vspace{0.5cm}
\subsubsection{Detection of Active Devices for IoT Systems}
\label{sec:sparse_det}
% =============================================

    In recent years, internet of things (IoT), providing network connectivity of almost all things at all times, has received much attention for its plethora of applications such as healthcare, automatic metering, environmental monitoring (temperature, humidity, moisture, pressure), surveillance, automotive systems, etc \cite{iot2}. Recently, it has been shown that the CS can be useful for various IoT applications including energy saving \cite{dan}, data acquisition and streaming \cite{shancang,melodia}, positioning \cite{feng,bowu,rss_cs,lcs},  monitoring \cite{congwang}, etc. See \cite{huang} for comprehensive survey.
    %
    %Recently, CS techniques have been used for the design of IoT systems \cite{wan,huang,shancang,melodia,feng,dan,bowu,lcs}.
    %
    Common feature of the IoT networks is that the node density is much higher than the cellular network, yet the data rate is very low and not every device transmits information at a given time. Hence, the number of active devices in IoT environments is in general much smaller than that of the inactive devices. Since it is not easy to acquire the active user information from complicated message handshaking, it is of importance to identify what devices are active at the access point (AP). This problem is one of key issues in massive machine type communications (mMTC) in 5G wireless communications \cite{iotlte,metis}.
    %
    % Due to this reason, when we consider the uplink of IoT networks, dimension of a transmit vector $\mathbf{s}$ (i.e., number of devices) is large but the number of nonzero elements of $\mathbf{s}$ (i.e., number of active devices) is small so that the transmit vector $\mathbf{s}$ can be readily modeled as a sparse vector.
    %
    %
    % Furthermore, since the available time/frequency resources of IoT systems is limited due to the limitation of bandwidth, cost of RF circuits and antenna, and power consumption,\footnote{In fact, a duty cycle based energy management is required for IoT sensors whose power consumption is very small, and hence they are sustainable by the energy harvesting from renewable resources, such as solar, wind, motion, and RF signals. In this regard, CS technique fits well into the ``opportunistic" harvesting and transmission of the IoT sensors to meet the ``bursty" energy and traffic arrivals, unlike the existing cellular network.} the number of resources $m$ is smaller than the number of total devices $n$.

\begin{figure}[h]
\begin{center}
\hspace{-3mm}\includegraphics[width=110mm,height=95mm]{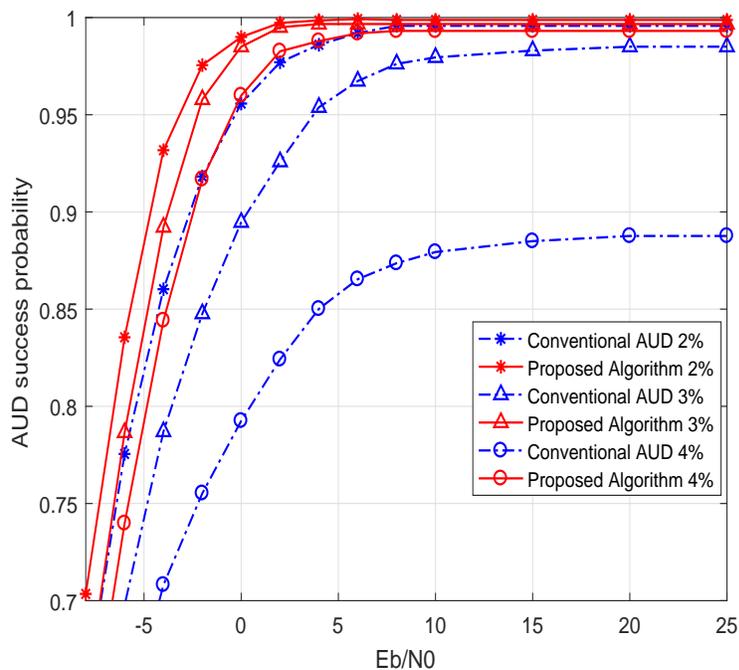}
\caption{Performance of AUD (FFT size is 512, length of CIR is 36, and among them three taps are dominant.)}
\label{fig:aud}
\end{center}
\end{figure}

    Suppose that there are $n$ devices in the networks and only $k (\ll n)$ devices among them are trying to access to AP.
    In each device, the quasi-orthogonal signature (codeword)  with length $m ( < n)$ is assigned from the codebook $Q = \{\mathbf{q}_1, \cdots, \mathbf{q}_{n}\}$.
   Since $m<n$,  orthogonality among code sequences cannot be guaranteed (i.e., $\mathbf{q}_i \mathbf{q}_j \neq 0 $ for $i \neq j$).
   When multiple devices send their signature sequences to AP, the first thing to be done at the AP is to identify which devices are transmitting information to the AP. Channel estimation, symbol detection, and decoding are performed after this active user detection (AUD).
   Under the flat fading channel assumption, which is true for narrowband IoT systems (e.g., eMTC and NB-IoT in 3GPP Rel. 13 \cite{iotlte}), the received vector at the AP is expressed as
%===================================
\begin{align}
\mathbf{y} &= \sum_{i=1}^{n} h_i \mathbf{q}_{i} p_i + \mathbf{v} \nonumber \\ %   + \cdots + h_{m} \mathbf{q}_{m} p_{m}  + \mathbf{n} \nonumber \\
&= [\mathbf{q}_{1} ~ \cdots ~ \mathbf{q}_{n}] \matc{h_1 p_1 \\ \vdots \\ h_n p_n } + \mathbf{v} \nonumber \\
& = \mathbf{H} \mathbf{s} + \mathbf{v} \label{eq:hsn}
\end{align}
%===================================
where $p_i$ is the symbol transmitted from the $i$th device, $h_i$ is the (scalar) channel from the $i$th device to the AP,  $\mathbf{H} = [\mathbf{q}_{1} ~ \cdots ~ \mathbf{q}_{n} ]$ is the $m \times n$ matrix generated by the codeword vector $\mathbf{q}_{i}$\footnote{There are various approaches to generate the codebook for machine type devices. See, e.g., sparse coding multiple access (SCMA) \cite{scma}.},
and $\mathbf{v}$ is the $m \times 1$ noise vector.
%
%With this setup, the symbol vector $\mathbf{s}$ contain $K$ nonzero elements and can be modeled as a sparse vector.
%Since only $K$ elements of $\mathbf{s}$ are nonzero, $\mathbf{s}$ can be readily modeled as a sparse vector.
In this setup, the problem to identify which devices are active is translated into the problem to find out the support of $\mathbf{s}$. As depicted in Fig.~\ref{fig:aud}, the CS-based AUD performs better than conventional scheme. After the AUD, columns $\mathbf{q}_i$ and symbol element $s_i$ corresponding to inactive users can be removed from the system model and as a result, we obtain the overdetermined system.
 Specifically, if the set $S = \{ s_1, ~ \cdots, ~ s_k \}$ is the support (index set of active users), then the system model to estimate the channel is
%===================================
\begin{eqnarray}
\mathbf{y} = [ p_{s_1} \mathbf{q}_{s_1} & \cdots & p_{s_k} \mathbf{q}_{s_k} ] \matc{ h_{s_1}  \\ \vdots \\ h_{s_k}  } + \mathbf{v} \nonumber \\
= \mathbf{T} \mathbf{h} + \mathbf{v}. \label{eq:hsn2}
\end{eqnarray}
%===================================
%
As long as $k < m$, \eqref{eq:hsn2} is modeled as an overdetermined system so that one can estimate the channel vector $\mathbf{h}$ accurately using the conventional LS or LMMSE technique.

%
%In essence, detection of active devices  boils down to detection of support in $\mathbf{s}$ based on the observation $\mathbf{y}$ which can be effectively solved by compressed sensing techniques.
%
%
% Traditional way to handle this problem to treat all interfering signals as a noise. Denoting $s_1$ as the desired symbol, the system model for this approach is given by $\mathbf{y} =h_1 \mathbf{q}_1 s_1 + (\sum_{i \neq i} h_i \mathbf{q}_i s_i + \mathbf{n})$ where the quantities inside the parenthesis correspond to an effective noise (sum of noise and interferences).
%
% This strategy is simple to implement but it is not so appealing since the signal recovery operation is performed in a very low signal-to-interference-noise ratio (SINR) regime (SINR $= \frac{ E\| \mathbf{h}_i \mathbf{s}_i \|_2^2}{  \sum_{j \neq i} E \| \mathbf{h}_i \mathbf{s}_i \|_2^2 + \sigma_v^2 }$).
% It is also possible to solve (\ref{eq:hsn}) directly via least square or linear minimum mean square error estimations. However, since the system is underdetermined, the solution would not be satisfactory.
% Clearly, compressed sensing would yield much better solution to these methods by exploiting the sparsity underlying in the device activities.
%

\subsubsection{Localization for wireless sensor network}

 In wireless sensor networks,  location information of the sensor nodes is needed to perform the location-aware resource allocation and scheduling. Location information is also important ingredient for the location-based service. In a nutshell, the localization problem is to estimate the position of the target nodes using the received signals at multiple APs \cite{correal}. For example, we consider the wireless sensor network equipped with multiple APs at different places to receive the signals from multiple targets. These APs collect the received signals and then send them to the server where localization task is performed (see Fig.~\ref{fig:iot_loc}). Assume that there are $k$ targets and $m$ APs measure the received signal strength (RSS). The location of the targets is represented in the two-dimensional location grid of size $n$. The conventional localization algorithm calculates the average RSS for  all candidates over the grid, and searches for the one closest to the measurements. Due to a large number of elements in the location grid, this approach will incur high computational complexity. Since the number of target nodes is much fewer than the number of elements in the grid,  the location information of the targets is sparse, and thus the CS technique can be used to perform efficient  localization  \cite{feng,bowu,rss_cs}. Let the average RSS from the $j$th grid element to the $i$th AP be $p_{i,j}$.
 \begin{figure}[t]
\begin{center}
\hspace{-3mm}\includegraphics[scale = 0.9]{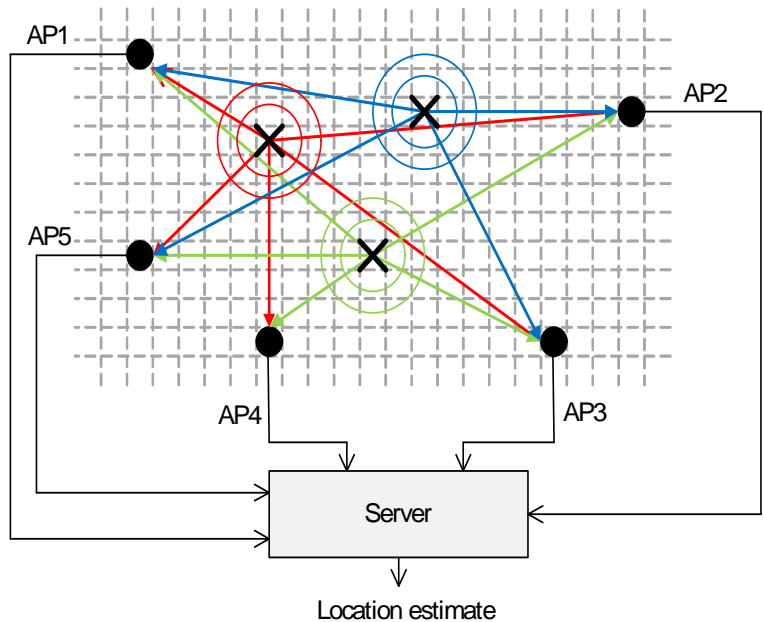}
\caption{Illustration of localization in wireless sensor network.}
\label{fig:iot_loc}
\end{center}
\end{figure}
The average RSS can be obtained from offline measurement at certain reference points (i.e., fingerprinting method \cite{rss,fingerprinting}) or calculated via appropriate path loss models. In the fingerprinting method, we obtain the average RSS at non-reference points via two dimensional interpolation. The measurement received by the $m$ APs $\mathbf{y} = [y_1, \cdots, y_{m}]^{T}$ is expressed as
 \begin{align}
 \mathbf{y} = \mathbf{P} \mathbf{s}  + \mathbf{v},
 \end{align}
 where $\mathbf{P}$ is the $m \times n$ matrix and the $i$th element of $\mathbf{s}$ has zero value if the target does not lie at the $i$th grid element, and $\mathbf{v}$ is the perturbation term capturing the noise and modeling error. The support of $\mathbf{s}$ represents the location of the target nodes and the CS technique can be employed to identify the support.

    \vspace{0.3cm}
    %----------------------------------------
    \subsubsection{Direction Estimation}
    %----------------------------------------
    Yet another important application of the support identification problem is the estimation of the angle of arrival (AoA) and angle of departure (AoD)  in wireless communications.
   % AoA and AoD estimation can also be well modeled as the support identification problem.
    In the millimeter wave (mmWave) communications, the carrier frequency increases up to tens or hundreds of GHz so that the transmit signal power decays rapidly with distance and wireless channels exhibit a few strong multipaths components caused by the small number of dominant scatterers.
    In fact, signal components departing and arriving from particular angles are very few compared to the total number of angular bins \cite{rappaport}. When the estimates of AoA and AoD  are available, beamforming with high directivity is desirable to overcome the path loss of mmWave wireless channels. The sparsity of the channel in the angular domain can be exploited for an effective estimation of the AoA and AoD via CS techniques \cite{el2014spatially,mmwave,mmwave_overview}.
    Consider a mmWave system where the transmitter and the receiver are equipped with $N_t$ and $N_r$ antennas, respectively.
    When employing the uniform linear array antennas, MIMO channels in the angular domain is expressed as
    \begin{align}
    \mathbf{H} = \mathbf{A}_{r} \mathbf{\Phi}_a \mathbf{A}_{t}^{H},
    \end{align}
    where  $\mathbf{A}_{r} = \left[\mathbf{a}_r(\phi_1), ...,\mathbf{a}_r(\phi_{L_r})  \right] \in {\mathcal{C}^{N_r \times L_r}}$, $\mathbf{A}_{t} = \left[\mathbf{a}_t(\phi_1), ...,\mathbf{a}_t(\phi_{L_t})  \right] \in {\mathcal{C}^{N_t \times L_t}}$, $L_r$ and $L_t$ are the number of total angular bins for the receiver and the transmitter, respectively, $\mathbf{a}_r(\phi_i)$ and $\mathbf{a}_t(\phi_i)$ are the steering vectors corresponding to the $i$-th angular bin for AoA and AoD, respectively, %, ($\mathbf{a}_r(\phi_i) = [1, e^{-j(2\pi/\lambda)d\sin(\phi_i)}, ...,e^{-j(N_r-1)(2\pi/\lambda)d\sin(\phi_i)} ] $),
    and
    $\mathbf{\Phi}_a$ is the $L_r\times L_t$ path-gain matrix whose $(i,j)$th entry contains the path gain from $j$th angular bin for the transmitter to $i$th angular bin for the receiver.
    Note that due to the sparsity of the channel in the angular domain, only a few elements of $\mathbf{\Phi}_a$ (i.e., $k$) are nonzero.
    In order to estimate AoA and AoD, the base-station transmits the known symbols  at the designated directions over $T$ time slots, where
    the transmit beamforming vector $\mathbf{f}_i \in {\mathcal{C}^{N_t \times 1}}$ is used for the $i$th beam transmission.
    Then, the received vector $\mathbf{z}_i \in \mathcal{C}^{N_r \times 1}$ corresponding to the $i$th beam transmission  can be expressed as
    \begin{align}
     \mathbf{z}_i = \mathbf{H} \mathbf{f}_i {x}_i + \mathbf{n}_i,
    \end{align}
    where $\mathbf{n}_i$ is the noise vector and $x_i$ is the known symbol.
       When the receiver applies the combining matrix $\mathbf{W} \in {\mathcal{C}^{N_r \times   Q}}$, we can obtain the measurement vector  $\mathbf{y}_i \in \mathcal{C}^{Q \times 1}$ as
     \begin{align}
    \mathbf{y}_i &= \mathbf{W}^{H} \mathbf{H} \mathbf{f}_i {x}_i + \mathbf{W}^{H} \mathbf{n}_i, \\
    &= \mathbf{W}^{H} \mathbf{H} \mathbf{f}_i {x}_i + \mathbf{v}_i,
 %   &= \mathbf{A}_{r} \mathbf{\Phi}_a \mathbf{A}_{t}^{H} \mathbf{x} + \mathbf{n} \nonumber \\
%    & = \sum_{i=1}^{N}\sum_{i=j}^{N}  \mathbf{a}_r(\theta_i)\mathbf{a}_t(\theta_j)^{H} \mathbf{x}  \phi_{i,j}  + \mathbf{n} \nonumber \\
%    &= \mathbf{H}\mathbf{s}  + \mathbf{n}, \label{eq:rrrx}
    \end{align}
     If we collect the received signals over $T$ beam transmissions, i.e., $\mathbf{y} = [\mathbf{y}_1^{T}, ..., \mathbf{y}_T^{T}]^T$ and let $x_i=1$,  then we have \cite{mmwave}
 \begin{align}
 \mathbf{y} &= {\rm vec} \left( \mathbf{W}^{H} \mathbf{H} \mathbf{F}  \right) + \mathbf{v} \\
 &= \left( \mathbf{F}^{T} \otimes \mathbf{W}^{H} \right) {\rm vec} (\mathbf{A}_{r} \mathbf{\Phi}_a \mathbf{A}_{t}^{H}) + \mathbf{v} \\
 %&= \left( \mathbf{F}^{T} \otimes \mathbf{W}^{H} \right) \left({\rm conj}(\mathbf{A}_t \circ \mathbf{A}_r  \right) \mathbf{h}+ \mathbf{v} \\
 & = \left( \left( \mathbf{F}^{T} (\mathbf{A}_t)^{*} \right) \otimes \left( \mathbf{W}^{H} \mathbf{A}_r \right) \right)  \mathbf{s}+ \mathbf{v}, \\
 &=\mathbf{P} \mathbf{s} + \mathbf{v}
 \end{align}
 where ${\rm vec}(\cdot)$ and $()^*$ are the vectorization and the conjugation operations, respectively, $\otimes$ is the  Kronecker  product,
%the input  each column of the matrix $(\mathbf{A}^{*}_{b}\circ \mathbf{A}_{m})$ is composed of the $(\mathbf{a}^{*}_{b}(\frac{2\pi}{M}i)\otimes \mathbf{a}_{m}(\frac{2\pi}{M}j))$ and
$\mathbf{v}=  [\mathbf{v}_1^{T}, ..., \mathbf{v}_T^{T}]^T$, $\mathbf{F} = \left[\mathbf{f}_1, \cdots, \mathbf{f}_{T}\right]$,   and $\mathbf{s} ={\rm vec}( \mathbf{\Phi}_a)$.
     Note that the indices of the nonzero elements in $\mathbf{s}$ correspond to the AoA and AoD information. The dimension of the vector $\mathbf{s}$ (i.e., $n = L_r \times L_t$) needs to be large in order to estimate the AoA and AoD at high resolution. On the other hand, the size of the measurement vector $n(=T \times Q)$ cannot be as large as the value of $m$ due to the resource constraints.  Since the number of nonzero elements is small, $\mathbf{s}$ is modeled by the  sparse vector and the CS technique is useful to find the support of $\mathbf{s}$.

\begin{figure}[t]
\begin{center}
\hspace{-3mm}\includegraphics[width=110mm,height=95mm]{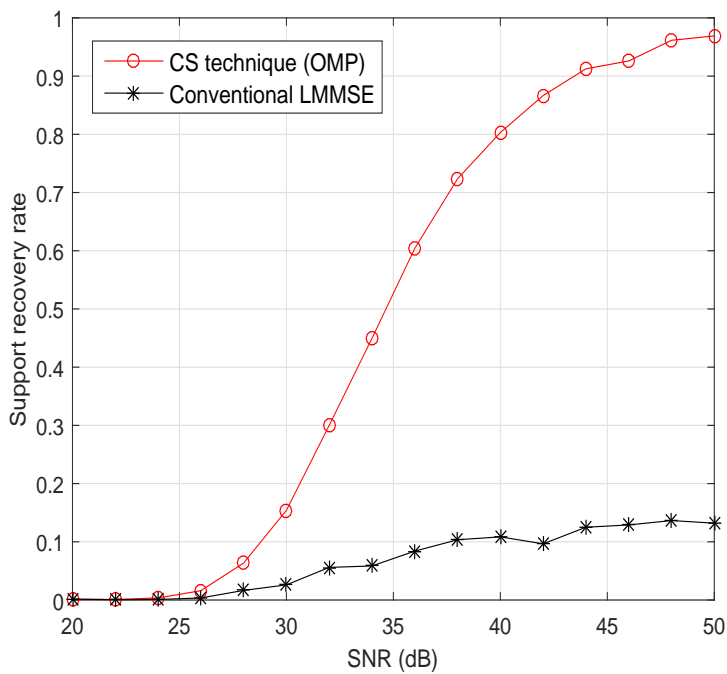}
\caption{Recovery performance of the CS technique in comparison with the conventional LMMSE method for mmWave communications.}
\label{fig:mmwave_sim}
\end{center}
\end{figure}

    Fig.~\ref{fig:mmwave_sim} shows the performance of the CS technique and the conventional LMMSE method.  We assume that the transmitter and the receiver use the beamforming and combining vectors steered to equally spaced directions with  $T=Q=16$.  We consider mmWave channels with three multi-path components (i.e., $k=3$) and both the transmitter and the receiver have 16 antennas. The resolutions for AoA and AoD are set to $L_r=32$ and $L_t=32$, respectively. In this setup, the size of the system matrix $\mathbf{P}$ is  $256 \times 1024$. With the LMMSE method, we find the support by picking the largest $k$ absolute values of the LMMSE estimate of $\mathbf{s}$. While the conventional LMMSE does not work properly, the CS technique performs accurate reconstruction of the support by exploiting the sparsity structure of the angular-domain channel.

\subsection{Sparse Detection}
\label{sec:non_sparse_det}
% =============================================

%%%%%%%%%%%%%%%%%%%%%%%%%%%%%%%%%%%%%%%%%%%%%%%%%%%%%%%
\begin{figure}[t]
\begin{center}
\subfigure[Block diagram]{\label{fig:det1}
\psfig{file=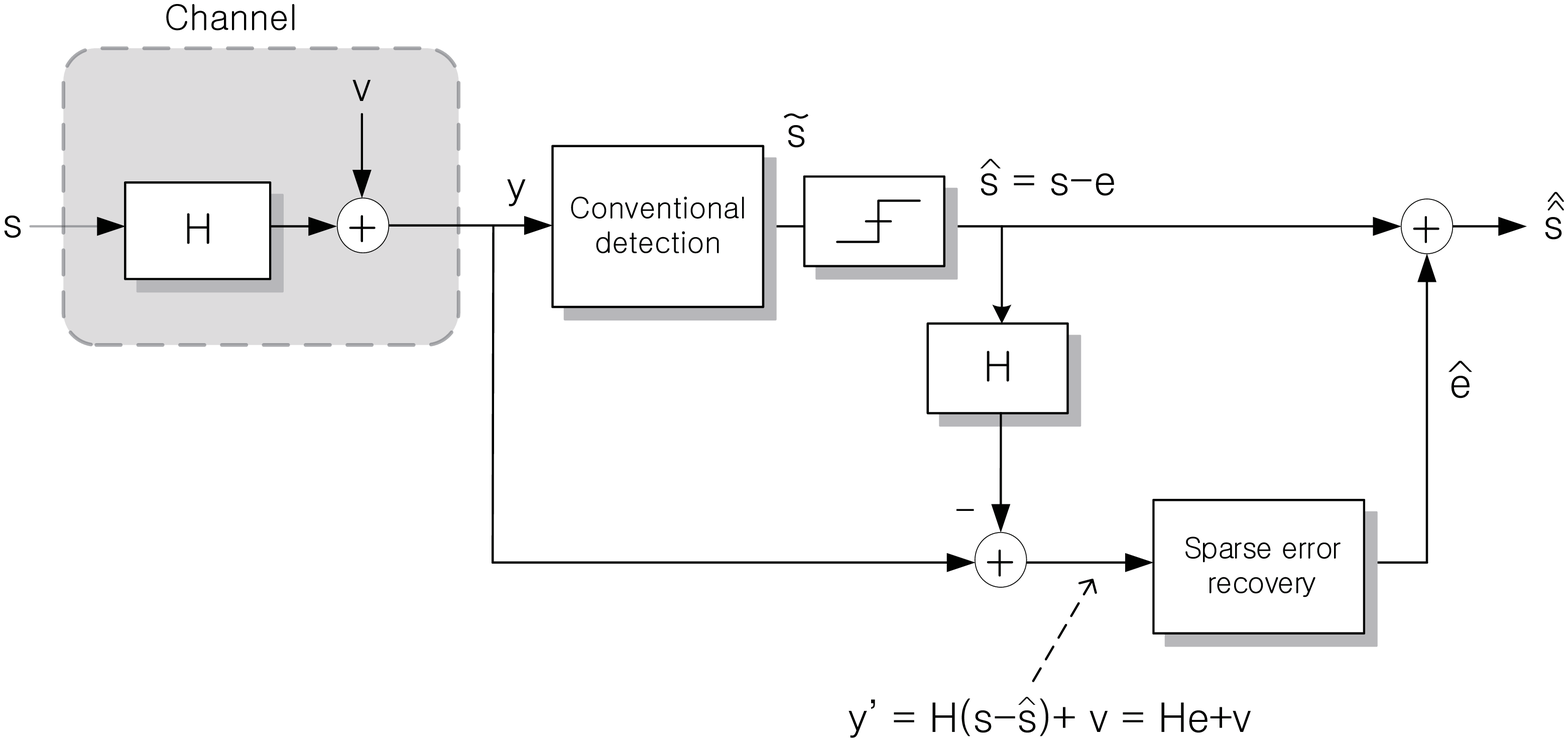,width=110mm,height=65mm} }\\\subfigure[SER performance]{\label{fig:det2}
\psfig{file=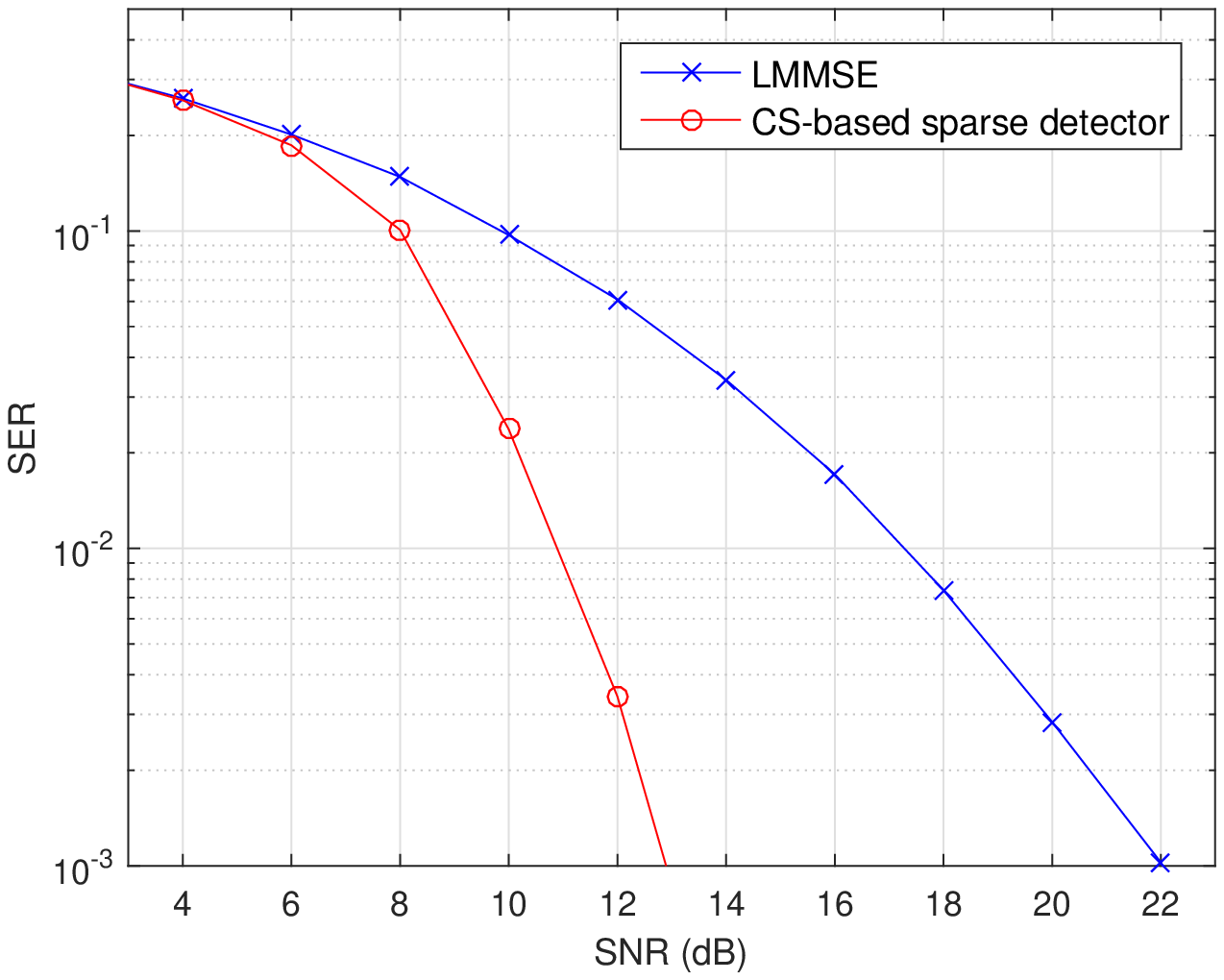,width=100mm, height=85mm} }\\
\caption{Sparse detection using the CS technique: (a) block diagram of CS-based sparse error detection (SED) technique and (b) symbol error rate (SER) performance.} \label{fig:sparse_det}
\end{center}
\end{figure}
%%%%%%%%%%%%%%%%%%%%%%%%%%%%%%%%%%%%%%%%%%%%%%%%%%%%%%%

    Even in the case where the transmit symbol vector is non-sparse, we can still use CS techniques to improve the performance of the symbol detection. There are a number of applications where the transmit vector cannot be modeled as a sparse signal.
    Even in this case, by the deliberate combination of conventional linear detection and sparse recovery algorithm, one can improve the detection performance \cite{globe14}. In this scheme, conventional linear detection such as LMMSE is performed initially to generate a rough estimate (denoted by $\tilde{\mathbf{s}}$ in Fig. \ref{fig:sparse_det}) of the transmit symbol vector. Since the quality of detected (sliced) symbol vector $\hat{\mathbf{s}}$ is generally acceptable in the operating regime of the target systems, the error vector $\mathbf{e} = \mathbf{s} - \hat{\mathbf{s}}$ after the detection would be readily modeled as a sparse signal.
    %\footnote{Error vector is defined as $\mathbf{e} = \mathbf{s} - \hat{\mathbf{s}} = \mathbf{s} - Q(\tilde{\mathbf{s}})$ where $\tilde{\mathbf{s}}$ is the output of the linear detector and $Q(\cdot)$ is the symbol slicer.}
    %
    %
    Now, by a simple transform of this error vector, one can obtain the new measurement vector $\mathbf{y}'$ whose input is the sparse error vector $\mathbf{e}$. This task is accomplished by the retransmission of the detected symbol $\hat{\mathbf{s}}$ followed by the subtraction (see Fig. \ref{fig:sparse_det} (a)). As a result, newly obtained received vector $\mathbf{y}'$ is expressed as $\mathbf{y}' = \mathbf{y} - \mathbf{H}  \hat{\mathbf{s}} = \mathbf{H} \mathbf{e} + \mathbf{v}$, where $\hat{\mathbf{s}}$ is the estimate of $\mathbf{s}$ obtained by the conventional detector.
    In estimating the error vector $\mathbf{e}$ from $\mathbf{y}'$, a sparse recovery algorithm can be employed. By adding the estimate $\hat{\mathbf{e}}$ of the error vector to the sliced symbol vector $\hat{\mathbf{s}}$, more reliable estimate of the transmit vector $\hat{\hat{\mathbf{s}}} = \hat{\mathbf{s}} + \hat{\mathbf{e}} = \mathbf{s} + (\hat{\mathbf{e}} - \mathbf{e})$ can be obtained.
    %The block diagram of the CS-based sparse error detection (SED) scheme is depicted in Fig. \ref{fig:sparse_det}(a) and its performance is plotted in Fig. \ref{fig:sparse_det}(b).
    The performance of CS-based sparse detection scheme is plotted  in Fig. \ref{fig:sparse_det} (b).

%
%
%In Table \ref{tb:summary}, we summarize three distinct CS subproblems related to the wireless communications.
%%
%\begin{table*}
%\begin{center}
%\caption{Summary of CS subproblems related to wireless communications}
%%\begin{tabular}{ll} \hline
%\begin{tabular}{c|c|c} \hline
%  % after \\: \hline or \cline{col1-col2} \cline{col3-col4} ...
%  Sub-problem & Function & Applications \\ \hline
%  Sparse estimation & estimate real/complex signals & channel estimation, impulsive noise cancellation \\ \hline
%  Support identification & identify support of the sparse vector & Spectrum sensing, direction estimation \\ \hline
%  Sparse detection & estimate real/complex signals & Symbol detection \\  \hline
%\end{tabular}
%\end{center}
%\label{tb:summary}
%\end{table*}

% =============================================
%
%
%

% =============================================
%
%
%
\section{Issues to Be Considered When Applying CS techniques to Wireless Communication Systems}
\label{sec:issues}
%
%
%
% =============================================
%
As more things should be considered in the design of wireless communication systems, such as wireless channel environments, system configurations (bandwidth, power, number of antennas), and design requirements (computational complexity, peak data rate, latency), the solution becomes more challenging, ambitious, and complicated.
As a result, applying the CS techniques to wireless applications is not any more copy-and-paste type task and one should have good knowledge on fundamental issues.
Some of the questions that wireless researchers can come across when they design a CS-based technique are listed as follows:
\begin{itemize}
\item Is sparsity important for applying CS technique? If yes, then what is the desired sparsity level?
\item How can we convert non-sparse vector into sparse one? Should we know the sparsity a priori?
\item What is the desired property of the system (sensing) matrix?
\item What kind of recovery algorithms are there and what are pros and cons of these?
\item What should we do if  more than one observations are available?
\item Can we do better if the input vector consists of finite alphabet symbols?
\end{itemize}
In this section, we address these issues in a way of answering to these questions. In each issue, we provide useful tips and tricks for the successful development of CS techniques for wireless communication systems.

% =============================================
\subsection{Is Sparsity Important?}
\label{sec:important}
% =============================================
%
If you have an application that you think CS-based technique might be useful, then the first thing to check is whether the signal vector to be recovered is sparse or not.
Many natural signals, such as image, sound, or seismic data are in themselves sparse or can be sparsely represented in a properly chosen basis.
Even though the signal is not strictly sparse, often it can be well approximated as a sparse signal.
For example, as mentioned, most of wireless channels exhibit power-law decaying behavior due to the physical phenomena of waves (e.g., reflection, diffraction, and scattering) so that the received signal is expressed as a superposition of multiple attenuated and delayed copies of the original signal. Since a few of delayed copies contain most of the energy, a vector representing the channel impulse response can be readily modeled as a sparse vector. Obviously, stronger sparsity (i.e. lower $k$) leads to better recovery performance for CS techniques as compared to the conventional estimator (see Fig. \ref{fig:perf_sparsity}).
Regarding the sparsity, an important question that one might ask for is what level of sparsity is enough to apply the CS techniques? Put it alternatively, what is the desired dimension of the observation vector when the sparsity $k$ is given? Although there is no clean-cut boundary on the measurement size under which CS-based techniques do not work properly,\footnote{In fact, this is connected to many parameters such as dimension of vector and quality of system matrix.} it has been shown that one can recover the original signals using $m = O ( k \log (\frac{n}{k} ))$ measurements via many of state-of-the-art sparse recovery algorithms.
Since the logarithmic term can be approximated to a constant, one can set $m = \epsilon k$ as a starting point (e.g., $\epsilon = 4$ by four-to-one practical rule \cite{cs_magazine}). This essentially implies that measurement size is a linear function of $k$ and unrelated to $n$. Note, however, that if the measurement size is too small and comparable to the sparsity (e.g., $m < 2k$ in Fig. \ref{fig:perf_sparsity}), performance of the sparse recovery algorithms might not be appealing.

%%%%%%%%%%%%%%%%%%%%%%%%%%%%%%%%%%%%%%%%%%%%%%%%%
\begin{figure}[t]
\begin{center}
	\includegraphics[width=110mm,height=95mm]{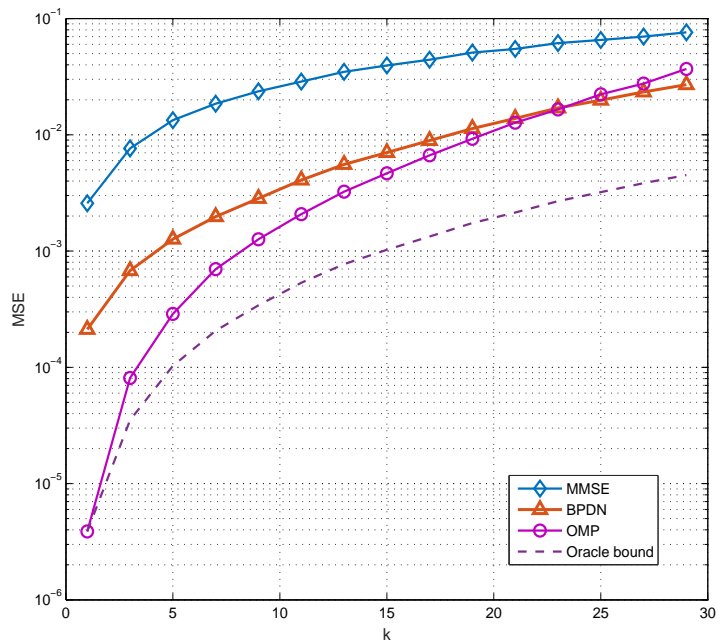} %comm_k_variation2}
\caption{Recovery performance as a function of sparsity. Elements of the sensing matrix is generated by the Gaussian random variable. We can clearly see that the CS techniques outperform the conventional MMSE technique by a large margin ($m=100$, $n=256$, SNR$=20$ dB).}
\label{fig:perf_sparsity}
\end{center}
\end{figure}
%%%%%%%%%%%%%%%%%%%%%%%%%%%%%%%%%%%%%%%%%%%%%%%%%
In summary, if the desired vector is sparse, then it is natural to consider the CS technique. As a rule of thumb, if the sparsity $k$ satisfies $m=4k$, one can try the CS technique.  We will discuss in the next subsection that even though the target vector is non-sparse, we can still use the CS technique by finding the basis over which the target vector can be sparsely represented.
%In short, bottom line to use the CS technique in this section is that the desired vector is sparse or can be made to be sparse.

% =============================================
\subsection{Predefined Basis or Learned Dictionary?}
\label{sec:important}
% =============================================
%
As discussed, to use CS techniques in wireless communication applications, we should ensure that the target signal has a sparse representation. Traditional CS algorithm is performed when the signal can be sparsely represented in an orthonormal basis, and many robust recovery theories are based on this assumption \cite{cs_magazine}. Although such assumption is valid in many applications, there are still plenty of scenarios where the target signal may not be sparsely represented in an orthonomal basis, but in an overcomplete dictionary \cite{candes2011compressed}. %olshausen1997sparse, }.  <-- Yacong ***
Overcomplete dictionary refers to a dictionary having more columns than rows. Since such dictionary is usually unknown beforehand, it should be learned from a set of training data.
This task, known as {\it dictionary learning} \cite{kreutz2003dictionary, aharon2006img}, is to learn an overcomplete dictionary
%Assume for a set of specific signals $\mathbf{x}_i \in\mathbb{C}^{N\times 1}$, we want to learn an overcomplete dictionary
$\mathbf{D} \in\mathbb{C}^{n \times m}$ ($n<m$) from a set of signals $\mathbf{x}_i$
%
%$\mathbf{x}_i \approx \mathbf{D}\mathbf{s}_i$
%
such that $\mathbf{x}_i$ can be approximated as $\mathbf{x}_i \approx \mathbf{D}\mathbf{s}_i$, where $\mathbf{s}_i \in\mathbb{C}^{m\times 1}$ and $\|\mathbf{s}_i\|_0\ll m$. Specifically, by using a training set $\mathbf{X}$ which contains $L$ realizations of $\mathbf{x}$, i.e., $\mathbf{X} = [\mathbf{x}_1,\mathbf{x}_2,\ldots,\mathbf{x}_L]$, we solve the following optimization problem
%
%--------------------------------------------------------
\begin{equation}\label{equ:dictionary learning}
\min_{\mathbf{D},\mathbf{s}_1,\ldots,\mathbf{s}_L}\lambda\|\mathbf{X}-\mathbf{DS}\|^2_F+\sum\limits_{i=1}^L\|\mathbf{s}_i\|_0
\end{equation}
%--------------------------------------------------------
%
where $\mathbf{S}=[\mathbf{s}_1,\mathbf{s}_2,\ldots,\mathbf{s}_L]$ is the matrix formed from all sparse coefficients satisfying $\mathbf{x}_k\approx \mathbf{D}\mathbf{s}_k$. Note that $\lambda$ is the parameter that trades off the data fitting error $\|\mathbf{X}-\mathbf{DS}\|^2_F$ and sparsity of the representation $\sum\limits_{i=1}^L\|\mathbf{s}_i\|_0$.
As a result of this, we can express the system model $\mathbf{y}=\mathbf{A}\mathbf{x} + \mathbf{n}$ into $\mathbf{y} \approx \mathbf{A} \mathbf{D} \mathbf{s} + \mathbf{n} = \mathbf{H} \mathbf{s} + \mathbf{n} $ where $\mathbf{H} = \mathbf{A} \mathbf{D}$.
After obtaining $\hat{\mathbf{s}}$ from the CS technique, we generate the estimated signal \(\hat{\mathbf{x}}=\mathbf{D}\hat{\mathbf{s}}\).
%
%Notice that due to the non-orthogonality of \(\mathbf{D}\), new theories are required to guarantee robust recovery \cite{candes2011compressed, rauhut2008compressed}.

%In practice, many signals can only be sparsely represented in an overcomplete dictionary. For example, sparsity of wireless channels  can be observed in adequately chosen overcomplete basis (e.g.,  delay-Doppler domain \cite{berger_mag} and spatial domain \cite{rao2014distributed}).
%
%
As an example to show the benefit of the dictionary learning, we consider the downlink channel estimation of the massive MIMO systems.
Consider the pilot-aided channel estimation where the basestation sends out pilots symbols \(\textbf{A}\in\mathbb{C}^{T\times n}\) during the training period $T$.
%
% and the user estimates the channel using the received vector \(\textbf{y}=\textbf{Ah}+\textbf{n}\).
%
In the multiple-input-single-output (MISO) scenario where the basestation has $n$ antennas and the mobile user has a single antenna, the received vector is $\textbf{y}=\textbf{Ah}+\textbf{n}$ where $\textbf{h}$ is the channel vector.
Traditional channel estimation schemes such as LS or MMSE estimation require more than or at least equal to $n$ measurements to reliably estimate the channel. In the massive MIMO regime where $n$ is in the order of hundred or more, this approach might not be practical since it consumes too much downlink resources. From our discussion, if \(\textbf{h}\) is sparse in some basis or dictionary \(\textbf{D}\) (i.e., \(\textbf{h}\approx \textbf{Ds}\), \(\|\textbf{s}\|_0\ll n\)), then with the knowledge of \(\textbf{A}\) and \(\textbf{D}\), \(\textbf{s}\) can be recovered from \(\textbf{y}=\textbf{ADs}+\textbf{n}\) using the CS technique, and subsequently the channel \(\textbf{h}\) is estimated as \(\hat{\textbf{h}}=\textbf{D}\hat{\textbf{s}}\).
Since the training period proportional to the sparsity of $\mathbf{s}$ ($T\propto \|\textbf{s}\|_0$) is enough, channel estimation using the CS technique is effective in reducing the pilot training overhead.

One can easily see that key issue in this process is to find $\textbf{D}$ such that \(\textbf{h}\) can be sparsely represented, that is, to express $\textbf{h}$ as \(\textbf{h} \approx \textbf{Ds}\) where $\textbf{s}$ is the sparse vector.
A commonly used basis  is the DFT basis which is derived from a uniform linear array deployed at the basestation \cite{bajwa2010compressed, el2014spatially}. However, the sparsity assumption under the orthogonal DFT basis is valid only when the scatters in the environment are extremely limited (e.g., a point scatter) and the number of antennas at the basestation goes to infinity \cite{sayeed2002deconstructing}, which is not applicable in many cases. %
Fig. \ref{fig:MSE_bar} depicts the model mismatch error \(\sum_{i=1}^L\|\textbf{h}_i-\textbf{D}\textbf{s}_i\|_2^2/L\) as a function of the number of atoms in \(\textbf{D}\) being used.
%
%(we test regarding to $k = 15, 30, 50$)
For each $k$, we set the constraint $\|\textbf{s}_i\|_0\le k$ for all $i$ and then compare three types of \(\textbf{D}\): orthogonal DFT basis, overcompleted DFT dictionary and overcompleted learned dictionary.
We observe that an approach using overcomplete dictionary achieves much smaller model mismatch error than the orthogonal basis, while the learned dictionary is even better than overcomplete DFT dictionary (see also \cite{ding2015channel}).
In this figure, we also observe that even with relatively less sparsity (e.g., $\|\textbf{s}_i\|_0 = 30$), the model mismatch error of the learned dictionary is pretty small \(\|\textbf{h}_i-\textbf{D}\textbf{s}_i\|_2^2 < 10^{-3}\). %so that \(\textbf{h}_i\) can be well approximated by \(\textbf{D}\textbf{s}_i\).
This example clearly demonstrates that the essential degree of freedom of the channel is much smaller than the dimension of the channel ($k=30 \ll n=100$).

%%%%%%%%%%%%%%%%%%%%%%%%%%%%%%%%%%%%%%%%%%%%%%%%%%
\begin{figure}[!t]
\begin{center}
	\includegraphics[width=110mm,height=95mm]{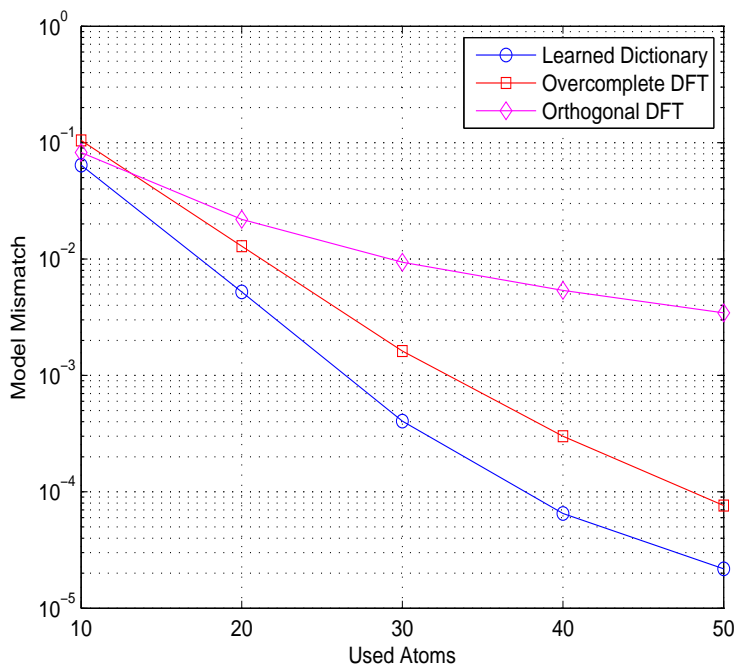}
\caption{MSE comparison of overcomplete learned dictionary, overcomplete DFT dictionary and orthogonal DFT basis. The channel \(\textbf{h}_i\) is generated using 3GPP spatial channel model (SCM) \cite{3gpp.25.996}, each \(\textbf{h}_i\) consists of 6 scatter clusters (3 far scatter and 3 local scatter).}
\label{fig:MSE_bar}
\end{center}
\end{figure}
%%%%%%%%%%%%%%%%%%%%%%%%%%%%%%%%%%%%%%%%%%%%%%%%%%
%
The main point in this subsection is that even though the vector is non-sparse, one can sparsify the desired vector using a technique such as dictionary learning. Using simple dictionary like DFT matrix would be handy but finding an elaborate dictionary is a bit challenging task and also interesting open research problem.
%Also, it is clear that computational cost associated with the learning will not be trivial. However,  it would not be a serious problem if this task is done in off-line.
In general, computational cost associated with the learning in (\ref{equ:dictionary learning}) will not be trivial. However, this would not be a serious problem since this task is done in offline in most cases.

% =============================================
\subsection{What is the Desired Property for System Matrix?}
\label{sec:property}
% =============================================
From the discussion above, one might think that the desired signal can be recovered accurately as long as the original signal vector is sparse.
Unfortunately, this is not always true since
%the condition that the desired vector is sparse is only necessary since
the accurate recovery of the sparse vector would not be possible when a poorly designed system matrix is used.
For example, suppose the support $\Omega$ of $\mathbf{s}$ is $\Omega = \{1, 3\}$ and the first and third columns of $\mathbf{H}$ are the same, then by no means the recovery algorithm will work properly. This also motivates that the columns in $\mathbf{H}$ should be designed to be as orthogonal to each other as possible. %Then the received vector $\mathbf{y}$ will not preserve any information on $\mathbf{s}$ and hence the recovery of the original signal is by no means possible.
Intuitively, the more the system matrix preserves the energy of the original signals, the better the quality of the recovered signal would be. The system matrices supporting this idea need to be designed such that each element of the measurement vector contains similar amount of information on the input vector $\mathbf{s}$.
That is the place where the random matrix comes into play. Although an exact quantification of the system matrix is complicated (see also next subsection), good news is that most of random matrices, such as Gaussian ($\mathbf{H}_{i,j} \sim N(0, \frac{1}{m})$) or Bernoulli with equal probability ($\mathbf{H}_{i,j} = \pm\frac{1}{m}$), well preserve the energy of the original sparse signal.

When the CS technique is applied to the wireless communications, the system matrix $\mathbf{H}$ can be determined by the process of generating the transmit signal and/or wireless channel characteristics.
%
%In certain systems, the structure of the sensing matrix $\mathbf{H}$ can be controlled by a system designer such that good recovery performance is ensured \cite{choi_chan}.
%
Fortunately, many of system matrices in wireless communication systems
%, such as Walsh-Hadamard matrix in CDMA and subsampled Fourier matrix in OFDMA,
behave like a random matrix.
Similarly, the system matrix is modeled by a Bernoulli random matrix when the channel estimation is performed for code division multiplexing access (CDMA) systems. Fading channel is often modeled as Gaussian random variables so that the channel matrix whose columns correspond to the channel vectors between mobile terminal and the base-station can be well modeled as a Gaussian random matrix.
%
%Hence, when the user activity is detected for IoT applications or transmitted symbol symbols are estimated via the method described in Fig.~\ref{fig:sparse_det}, the system matrix will follow Gaussian random matrix.
%
%In Fig. \ref{fig:perf_sparsity}, we plot the performance of two well-known sparse recovery algorithms (BPDN and OMP) and MMSE estimator for two distinct system matrices. In these results, we observe that the performance using subsampled DFT and linear feedback shift register (LFSR)-based system matrix  is not much different from that using pure random matrices. % (Gaussian and Bernoulli random matrices.

While the system matrix of many wireless applications works well in many cases, in some case we can also design the system matrix to improve the reconstruction quality. This task, called {\it sensing matrix design}, is classified into two approaches.
In the first approach, we assume that the desired signal \(\mathbf{x}\) is sparsely represented in a dictionary \(\mathbf{D}\). Then, the system model is expressed as $\mathbf{y} = \mathbf{Hx} = \mathbf{HDs}$.  In this setup, the goal is to design $\mathbf{H}$ adapting to dictionary \(\mathbf{D}\) such that columns in the combined dictionary $\mathbf{E} = \mathbf{HD}$ has good geometric properties \cite{elad2007optimized, duarte2009learning, yu2011measurement}. In other words, we design $\mathbf{H}$ such that the columns in \(\mathbf{E}\) are as orthogonal to each other as possible.
In the second type, rows of \(\mathbf{H}\) are \emph{sequentially} designed using previously collected measurements as guidance \cite{malloy2014near, haupt2012sequentially}. The main idea of this approach is to estimate the support from previous measurements and then allocate the sensing energy to the estimated support element.
%and the benefits of those adaptive sensing are decrement of the required magnitude of the smallest element in \(\mathbf{s}\) for robust recovery.  <-- Hard to understand
%
Recently, the system design strategies to generate a nice structure of system matrix in terms of recovery performance for massive MIMO systems were proposed \cite{qi,chan_choi}.

In summary, in order to make the CS technique effective, columns of the system matrix should be as uncorrelated as possible. Simple and effective way to do so is to design the elements of the system matrix as random as possible.
One can check this by computing the mutual coherence $\mu (\mathbf{H})$ or average coherence (taking average of all inner product instead of finding maximum). Also, one can learn from the bound in Section II.D ($\mu (\mathbf{H}) \geq \frac{1}{\sqrt{m}}$) that the coherence of the system matrix improves (decreases) with the dimension and so would be the resulting recovery performance.
%Another important point in this subsection is that the system matrix can be designed to improve the recovery performance.
Also,  similar to the dictionary learning, we can design (learn) the system matrix to achieve better recovery performance. This is interesting research topic since the recovery performance depends strongly on the system matrix, simple and efficient design strategy would be beneficial to improve the performance.

\subsection{What Recovery Algorithm Should We Use?}
\label{sec:recoveryalg}
% =============================================
%Mapping the system into the CS problem and using the right solution would be two important things to do.
When the researchers consider the CS techniques in their applications, they can be confused by a plethora of algorithms. There are hundreds of sparse recovery algorithms in the literatures, and still many new ones are proposed every year. The tip for not being flustered in a pile of algorithms is to clarify the main issues like the target specifications (performance requirements and complexity budget), system environments (quality of system matrix, operating SNR regime), dimension of measurements and signal vectors, and also availability of the extra information.
Perhaps two most important issues in the design of CS-based wireless communication systems are the mapping of the wireless communication problem into an appropriate CS problem and the identification of the right recovery algorithm.
Often, one should modify the algorithm to meet the system requirements.
Obviously, identifying the best algorithm for the target application is by no means straightforward and one should have basic knowledge of the sparse recovery algorithm.
In this subsection, we provide a brief overview on four major approaches: {\it $\ell_1$-norm minimization, greedy algorithm, iterative algorithm}, and {\it statistical sparse recovery technique}. Although not all sparse recovery algorithms can be grouped into these categories, these four are important in various standpoints such as popularity, effectiveness, and historical value.

\begin{itemize}

%%%%%%%%%%%%%%%%%%%%%%%%%%%%%%%%%%%%%%%%%%%%%%%%%%%%%%%%%%%%%%%%%%%%%%%%%
\item {\bf Convex optimization approach ($\ell_1$-norm minimization)}:
As mentioned, with the knowledge of the signal $\mathbf{s}$ being sparse, the most natural way is to find a sparse input vector
under the system constraint ($\arg\min \| \mathbf{s} \|_0 ~\mbox{s.t.}~ \mathbf{y} = \mathbf{H s}$).
% \end{align}
%In the presence of the measurement noise (i.e., typical case for communication scenarios),   we solve
%  \begin{align}
% \arg\min \| \mathbf{y} = \mathbf{H s}\|_2 + \lambda \| \mathbf{s} \|_0
% \end{align}
%
%
Since the objective function $\| \mathbf{s} \|_0$ is non-convex, solution of this problem can be found in a combinatoric way.
%To be specific, starting from $K =1$ assumption, solution satisfying $\mathbf{y} = \mathbf{H s}$ is found. If $K = 1$, then only one entry of $\mathbf{x}$ is nonzero and thus the equality constraint is simply expressed as $\mathbf{y} = \mathbf{h}_i s_i$. If the solution $\hat{s}_i$ satisfying this equality is found among ${n \choose 1}$ choices, we finish the algorithm and return $\hat{\mathbf{s}} = [0 ~ \cdots 0 ~ \hat{x}_i ~0 ~ \cdots ~ 0]^T$. Otherwise, we need to move on $K=2$ assumption and then investigate ${n \choose 2}$ possible choices satisfying the constraint $\mathbf{y} = [\mathbf{h}_i ~ \mathbf{h}_j ] [x_i ~ x_j ]^T$. This operation is repeated until the solution satisfying the equality constraint is found.\footnote{The way to solve the problem in a noisy scenario is similar to the noiseless scenario (equality constraint is changed to inequality constraint $\| \mathbf{y} - \mathbf{H x} \|_2 < \epsilon$).}
% For a practical systems, clearly, this approach is too expensive and impractical. As an approximate way, an algorithm based on the tree search has been proposed recently. MMP}
%
%
As an approach to overcome the computational bottleneck of $\ell_0$-norm minimization, $\ell_1$-norm minimization has been used. If the noise power is bounded to $\epsilon$, $\ell_1$-minimization problem is expressed as
$$\arg\min \| \mathbf{s} \|_1 ~ \mbox{s.t.} ~ \| \mathbf{y} - \mathbf{H s} \|_2 < \epsilon.$$
%
%
%Note that $\epsilon = 0$ for the noiseless setting.
%
%
%In \cite{stable}, it is shown that if the noise power is limited to $\epsilon$ and the number of observations is sufficiently large, $\ell_2$-norm of the reconstruction error is within the constant multiple of $\epsilon$ (i.e., $\| \hat{\mathbf{s}} - \mathbf{s} \|_2 < c_0 \epsilon$ where $c_0$ is a constant).

Basis pursuit de-noising (BPDN) \cite{bpdn}, also called Lasso \cite{lasso}, relaxes the hard constraint on the reconstruction error by introducing a soft weight $\lambda$ as
%=====================================================
\begin{eqnarray}
  \hat{\mathbf{s}} = \arg \min_{\mathbf{s}} \| \mathbf{y} - \mathbf{H s} \|_2 + \lambda \| \mathbf{s} \|_1.
  \label{eq:bpdn}
\end{eqnarray}
%=====================================================
%Note that $\lambda > 0$ controls the tradeoff between the approximation accuracy and the sparsity of the solution\cite{eldar_book}.
The recovery performance of such $\ell_1$-minimization method can be further enhanced by solving a sequence of weighted  $\ell_1$ optimization \cite{reweighted_candes}.
Although the $\ell_1$-minimization problem is convex optimization problem and thus efficient solvers exist,
%
%based on the linear programming (LP) exist (e.g., interior-point method \cite{boyd}). However, when
computational complexity of this approach is still burdensome in implementing real-time wireless communication systems.

\item {\bf Greedy algorithm}: In principle,  main principle of the greedy algorithm is to successively identify the subset of support (index set of nonzero entries) and refine them until a good estimate of the support is found. Suppose the support is found accurately, then the estimation of support elements would be straightforward since one can convert the underdetermined system into overdetermined one by removing columns corresponding to the zero element in $\mathbf{s}$ and then use a conventional estimation scheme like MMSE or least squares (LS) estimator.
    In many cases, greedy algorithm attempts to find the support in an iterative fashion, obtaining a sequence of estimates $(\hat{\mathbf{s}}_1 , \cdots , \hat{\mathbf{s}}_n)$.
    While the OMP algorithm picks an index of column of $\mathbf{H}$ one at a time using a greedy strategy \cite{omp}, recently proposed variants of OMP, such as generalized OMP (gOMP) \cite{gomp}, compressive sampling matching pursuit (CoSaMP) \cite{cosamp}, subspace pursuit (SP) \cite{sp}, and multipath matching pursuit (MMP) \cite{kwon}, have the refined step to improve the recovery performance. For example, gOMP selects multiple promising columns in each iteration. CoSaMp \cite{cosamp} and SP \cite{sp} incorporate special procedures to refine the set of column indices by 1) choosing more than $k$ columns of $\mathbf{H}$, 2) recovering the signal coefficients based on the projection onto the space of the selected columns, and 3) rejecting those might not be in the true support. MMP performs the tree search and then find the best candidate among multiple promising candidates obtained from the tree search.
    In general, these approaches outperform the OMP algorithm at the cost of higher computational complexity.
   In summary, the greedy algorithm has computational advantage over the convex optimization approach while achieving comparable (sometimes better) performance.

% --------------------------------------------------
% I don't think this is needed at this pointS
%\item {\bf Hard thresholding} We need a method to search among the classifiers in the language for the highest-scoring one. The choice of optimization technique is key to the efficiency of the learner, and also helps determine the classifier produced if the evaluation function has more than one optimum. It is common for new learners to start out using off-the-shelf optimizers, which are later replaced by custom-designed ones.
% --------------------------------------------------

\item {\bf Iterative algorithm}:    Sparse solution can be found by refining  the sparse signal estimate in an iterative fashion.  This approach called iterative hard thresholding (IHT) \cite{iht,maleki} performs the following update step iteratively
    %--------------------------------------------
    \begin{align} \label{eq:iht}
    \hat{\mathbf{s}}^{(i+1)} = T\left(\hat{\mathbf{s}}^{(i)} + \mathbf{H}^{H}(\mathbf{y} - \mathbf{H} \hat{\mathbf{s}}^{(i)} \right),
    \end{align}
    %--------------------------------------------
    %
    where $\hat{\mathbf{s}}^{(i)}$ is the estimate of the signal vector $\mathbf{s}$ at the $i$th iteration and $T(\cdot)$ is the thresholding operator.
     %
     %Though the performance of the IHT is slightly worse than that of the OMP, it is easier to implement in hardware since it does not need to perform least square projection requiring matrix inversion.
    Algorithms based iterative thresholding  yet exhibiting improved performance have been proposed  in \cite{amp1,amp2}.

  \item {\bf Statistical sparse recovery}: Statistical sparse recovery algorithms treat the signal vector $\mathbf{s}$ as a random vector and then infer it using the Bayesian framework. In the maximum-a-posteriori (MAP) approach, for example, an estimate of $\mathbf{s}$ is expressed as
      $$\hat{\mathbf{s}} = \arg \max_{\mathbf{s}} \ln f(\mathbf{s}|\mathbf{y}) = \arg \max_{\mathbf{s}}\ln f(\mathbf{y}|\mathbf{s}) + \ln f(\mathbf{s}),$$
      where $f(\mathbf{s})$ is the prior distribution of $\mathbf{s}$.
      To model the sparsity nature of the signal vector $\mathbf{s}$, $f(\mathbf{s})$ is designed in such a way that it decreases with the magnitude of $\mathbf{s}$. Well-known examples include i.i.d. Gaussian and Laplacian distribution. For example, if i.i.d. Laplacian distribution is used, then the prior distribution $f(\mathbf{s})$ is expressed as
      %%%%%%%%%%%%%%%%%%%%%%%%%%
      $$f(\mathbf{s}) = \left(\frac{\lambda}{2}\right)^n \exp\left(-\lambda \sum_{i=1}^{n} |s_i| \right).$$
      %%%%%%%%%%%%%%%%%%%%%%%%%%
      Note that the MAP-based approach with the Laplacian prior model leads to the algorithm similar to the BPDN in \eqref{eq:bpdn}. When one chooses other super-Gaussian priors, the model reduces to a regularized least squares problem \cite{gorodnitsky1997sparse, candes2008enhancing, chartrand2008iteratively}, which can be solved by a sequence of reweighted  $\ell_1$ or $\ell_2$ algorithms.
      %
      %
      %{\bf Besides, statistical inference on sparse signals can be accomplished by describing the statistical dependency of random variables  in the joint distributions $Pr(\mathbf{y},\mathbf{s})$ using a factor graph and developing the message passing algorithm over the graph. A regular CS setup yields a  dense factor graph and thus the propagation of the messages from a number of neighboring nodes often involves substantial computational complexity.
      %It was shown in \cite{amp_donoho} that some approximations made in manipulating the messages can lead to asymptotically optimal yet computationally effective recovery algorithm. This approach has produced a family of practical statistical recovery algorithms referred to as approximate message passing (AMP).}
%
%
%%%%%%%%%%%%%%%%%%%%%%%%%%%%%%%%%%%%%%%%%%%%%%%%%%%
%%%%%%%%%%%%%%%%%%%%%%%%%%%%%%%%%%%%%%%%%%%%%%%%%%%
\begin{table*}
\begin{center}
\caption{Summary of sparse recovery algorithms} \label{tb:smv} \small
\begin{tabular}{ |p{2cm}|p{1.8cm}|p{6cm}|p{5cm}| }
\hline
%\multicolumn{2}{ |c| }{Sparse recovery algorithm ({\bf remove this row})} \\
\hline
 Approach & Algorithm & Features & Computational Complexity \\ \hline
\multirow{2}{3cm}{Convex optimization}
   & BPDN \cite{bpdn} & Reconstruction error $\|\mathbf{y}-\mathbf{H}\mathbf{s}\|_2$ regularized with $\ell_1$ norm $\|\mathbf{s}\|_1$ is minimized.
   & Complexity is  $\mathcal{O}(m^2n^3)$. Recently, the speed of optimization for the BPDN has been substantially improved   \cite{boyd}.
   \\ \cline{2-4}
    & Reweighted $\ell_1$ minimization\cite{reweighted_candes} &  The BPDN can be improved via iterative reweighted $\ell_1$-minimization.
    & Complexity is  $\mathcal{O}(m^2n^3\cdot \mbox{iter})$. Computational complexity of this approach is higher than the BPDN.
    \\ \hline
\multirow{3}{3cm}{Greedy \\ algorithm}
   & OMP \cite{omp}, gOMP \cite{gomp} & The indices of nonzero elements of $\mathbf{s}$ are identified in an iterative fashion.
    & Complexity is $O(mn
    k)$.  Complexity is low for small $k$.
    \\ \cline{2-4}
   & CoSaMp \cite{cosamp}, SP \cite{sp} & More than $k$ indices of the nonzero elements of $\mathbf{s}$ are found and then candidates of poor quality are pruned afterwards.
   & Complexity is $O(mn\cdot \mbox{iter})$. It requires higher complexity than the OMP.
   \\ \cline{2-4}
      & MMP \cite{kwon} & Tree search algorithm is adopted to search for the indices of the nonzero elements in $\mathbf{s}$ efficiently.
   & The complexity is higher than that of the OMP.  The tree-based search offers the trade-off between performance and complexity. \\ \hline
   \multirow{2}{3cm}{Iterative \\ algorithm}
   & IHT \cite{iht}  &  Iterative thresholding step in (\ref{eq:iht}) is performed repeatedly. It works well under limited favorable scenarios.
   &
   Complexity is $\mathcal{O}(mn \cdot \mbox{iter})$.  Implementation cost is low.
    \\ \cline{2-4}
   & AMP \cite{amp2}  & Gaussian approximations in message passing are used to derive the algorithm.  %When the prior knowledge on the signal and noise distribution is unavailable or incorrect, it might not perform well
    & Complexity is comparable to that of the IHT.
    \\ \hline
   \multirow{2}{3cm}{Statistical \\ sparse \\ recovery}
   & MAP with Laplacian prior \cite{jeffrey}   &  MAP estimation of the sparse vector is derived using Laplacian distribution as sparsity-promoting prior distribution. &
   Complexity is comparable to that of the BPDN.   \\ \cline{2-4}
   & SBL \cite{wipf2004sparse}, BCS \cite{bcs}  & Hyper-parameter is used to model the sparsity of the sparse signals. The EM algorithm is used to find the hyper-parameter and signal vector iteratively.
   & Complexity is $\mathcal{O}(n^3)$. Computational complexity can be high for large size problems due to matrix inversion operation.
    \\
   \hline
   \end{tabular}
\end{center}
\end{table*}
%%%%%%%%%%%%%%%%%%%%%%%%%%%%%%%%%%%%%%%%%%%%%%%%%%%
%
%%%%%%%%%%%%%%%%%%%%%%%%%%%%%%%%%%%%%%%%%%%%%%%%%%%

%%%%%%%%%%%%%%%%%%%%%%%%%%%%%%%%%%%%%%%%%%%%%%%%%%%
%
%
  Different type of statistical sparse recovery algorithms are sparse Bayesian learning (SBL)  \cite{wipf2007empirical} and Bayesian compressed sensing \cite{bcs}. In these approaches, the priori distribution of the signal vector $\mathbf{s}$ is modeled as zero-mean Gaussian with the variance parameterized by a hyper-parameter. In SBL, for example, it is assumed that each element of \(\mathbf{s}\) is a zero-mean Gaussian random variable with variance \(\gamma_k\) (i.e., $s_k \sim \mathcal{N}(0,\gamma_k)$).
  %
  %, and the CS measurement \(\mathbf{y}=\mathbf{H}\mathbf{s}+\mathbf{n}\) where \(\mathbf{n}\sim \mathcal{N}(\mathbf{0},\sigma^2\mathbf{I})\).
  %
  A suitable prior on the variance $\gamma_k$ allows the modeling of several super-Gaussian densities. Often a non-informative prior is used and found to be effective. Let \(\boldsymbol{\gamma}=\{\gamma_k, \forall k\}\), then the hyperparameters \(\Theta=\{\boldsymbol{\gamma}, \sigma^2\}\) which control the distribution of \(\mathbf{s}\) and \(\mathbf{y}\) can be estimated from data by marginalizing over \(\mathbf{s}\) and then performing evidence maximization or Type-\RN{2} maximum-likelihood estimation \cite{tipping2001sparse}:
      %%%%%%%%%%%%%%%%%%%%%%%%%%%%%%%%%%%%%%%%%%%%%%%
      \begin{equation}\label{eq: sbl ml}
      \begin{aligned}
      \hat{\Theta}
      &=\arg\max_{\Theta}p(\mathbf{y};\boldsymbol{\gamma},\sigma^2) \\
      &=\arg\max_{\Theta}\int p(\mathbf{y}|\mathbf{s};\sigma^2)p(\mathbf{s};\boldsymbol{\gamma})d\mathbf{s}.
      \end{aligned}
      \end{equation}
      %%%%%%%%%%%%%%%%%%%%%%%%%%%%%%%%%%%%%%%%%%%%%%%
     The signal \(\mathbf{s}\) can be inferred from the maximum-a-posterior (MAP) estimate after obtaining \(\hat{\Theta}\):
      %%%%%%%%%%%%%%%%%%%%%%%%%%%%%%%%%%%%%%%%%%%%%%%
      \begin{equation}
      \mathbf{s}=\arg\max_{\mathbf{s}} p(\mathbf{s} | \mathbf{y} ; \hat{\Theta}).
      \end{equation}
      %%%%%%%%%%%%%%%%%%%%%%%%%%%%%%%%%%%%%%%%%%%%%%%
      By solving (\ref{eq: sbl ml}), we obtain the solution of \(\boldsymbol{\gamma}\) with most of elements being zero. Note that $\boldsymbol{\gamma}$ controls the variance of \(\mathbf{s}\). When \(\gamma_k=0\), it implies \(s_k=0\), which results in a sparse solution. It has been shown that with appropriately chosen parameters, the SBL algorithm outperforms $\ell_1$ minimization and iteratively reweighted algorithms \cite{wipf2004sparse}.

\end{itemize}

In Table \ref{tb:smv}, we summarize the key features of the sparse recovery algorithms and  briefly comment on their computational complexity.

 \begin{figure} [t]
 \centering
    \subfigure[]
  {\epsfig{figure=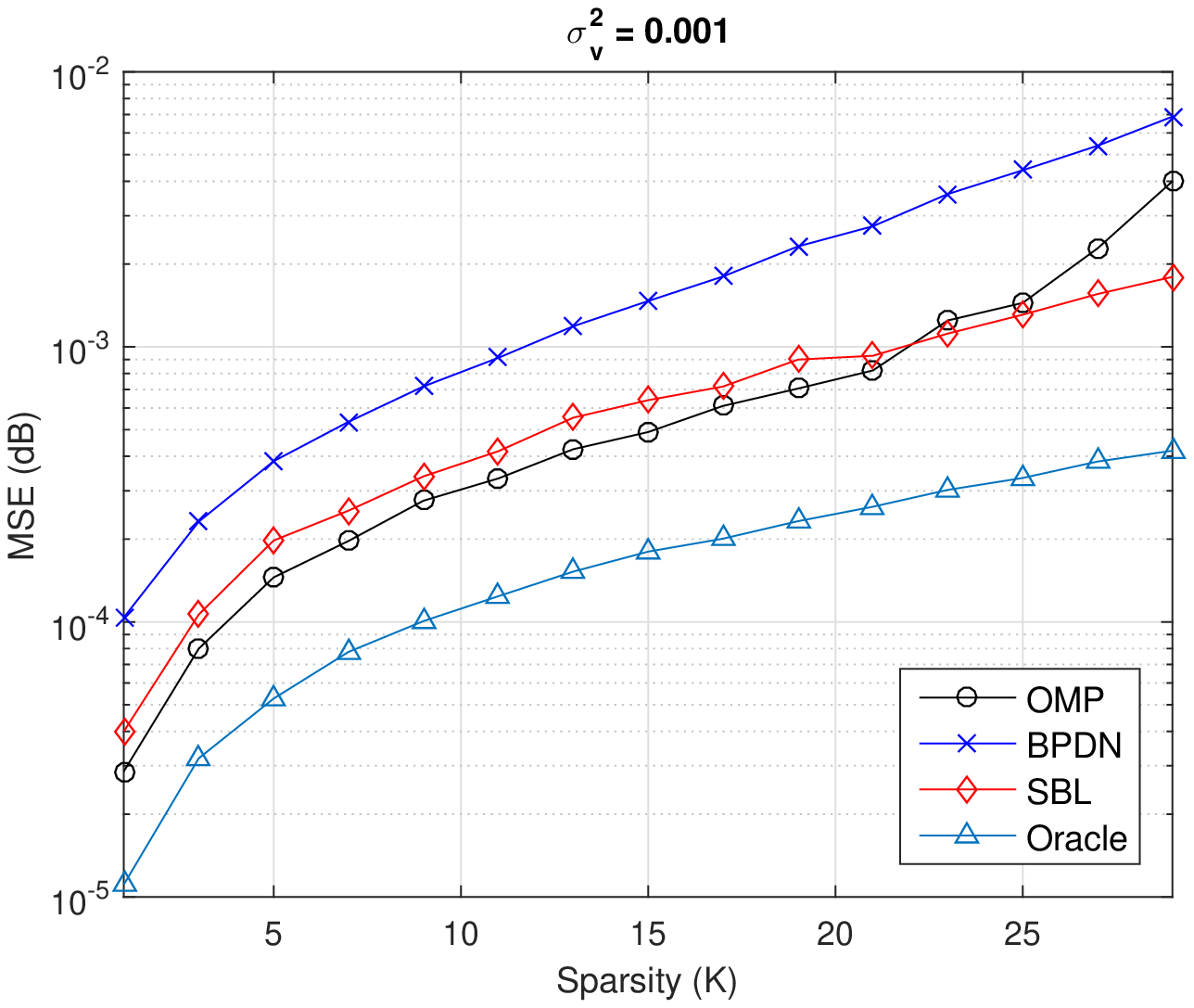, width=80mm,height=65mm}}
   \subfigure[]
  {\epsfig{figure=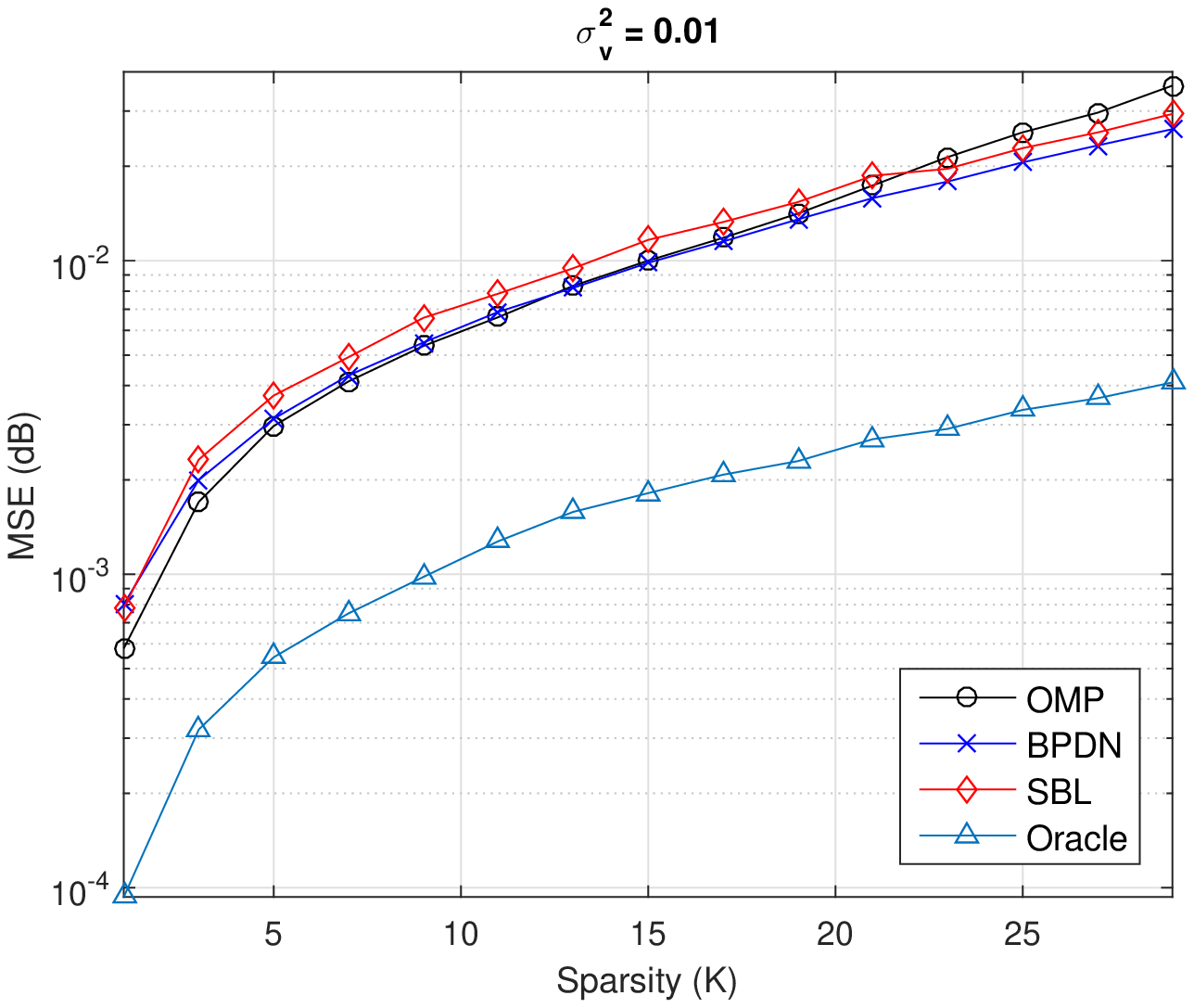, width=80mm,height=65mm}}
   \subfigure[]
  {\epsfig{figure=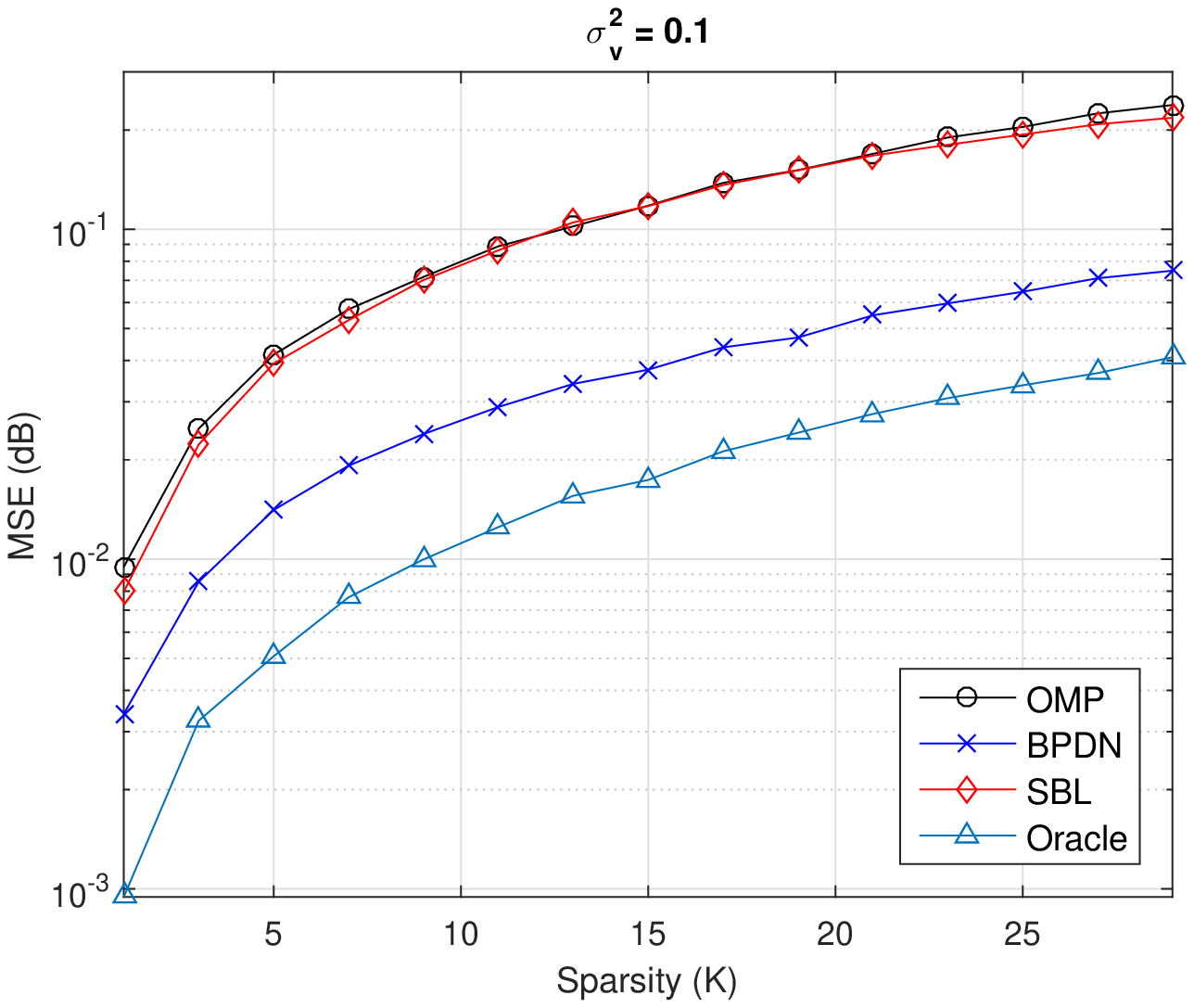, width=80mm,height=65mm}}
  \caption { Performance comparison of several recovery algorithms  (a) in high SNR ($\sigma_v^2$ = 0.01), (b) in mid SNR ($\sigma_v^2$ = 0.01), and (c) in low SNR ($\sigma_v^2$ = 0.1).
   } \label{fig:perf}
\end{figure}

To sum up, to choose the right algorithm is very important but it is not easy to find the algorithm fulfilling the designer's expectation among a variety of algorithms.
The performance and complexity depend on various parameters like problem size $m$ and $n$, sparsity level $k$, quality of sensing matrix, and noise intensity. In order to provide simple idea, we present the performance of three
well-known recovery algorithms BPDN, OMP, and SBL for different setups.  The sensing matrix $\mathbf{H}$ is generated from $N(0,1/n)$ and the nonzero elements of the signal vector $\mathbf{s}$ are from i.i.d. Gaussian distribution $N(0,1)$.
We fix $m=100$ and $n=256$ and vary the sparsity level $k$ and the noise variance $\sigma_v^2$. As a reference, we also plot the Oracle estimation where the signal recovery is performed
using the perfect knowledge of the support. Fig.~\ref{fig:perf} shows the MSE performance as a function of sparsity level $k$ for several recovery algorithms. We consider different levels of noise variance $\sigma_v^2 = 0.001, 0.01$ and $0.1$ in Fig.~\ref{fig:perf} (a), (b) and (c), respectively.  When the noise level is small (i.e., high SNR) and the sparsity level is less than 20, the OMP algorithm outperforms the BPDN and SBL. This is not a surprise because the support found by the OMP is very accurate for this scenario so that the quality of the resulting estimate is also very reliable. Note that the SBL slightly outperforms the OMP  when the sparsity level exceeds 20.  When noise variance increases to the mid-level, the performance of all three algorithms is more or less similar. Whereas, when the noise level is high, the BPDN performs significantly better than the OMP and SBL. It is clear from this observations that the performance of the recovery algorithms depends strongly on the system parameters. Performance is also sensitive to the structure  of
sensing matrix and the distribution of the signal vector. Hence, for the given setup and operating condition, one should check the empirical performance to find the right algorithm.

\subsection{Can We Do Better If Multiple Measurement Vectors Are Available?}
% =============================================
In many wireless communication applications, such as the wireless channel estimation problem and AoA and AoD estimation in mmWave communication systems \cite{mmwave,mmwave_mmv}, multiple snapshots (more than one observation) are available and further the nonzero positions of these vectors are invariant or varying slowly.
The problem to recover the sparse vector from multiple observations, often called multiple measurement vectors (MMV) problem, received much attention recently due to its superior performance compared to the single measurement vector (SMV) problem (see Fig.~\ref{fig:mmv}).
Group of measurements sharing common support are useful in many wireless communication applications since multiple measurements exploit correlation among sparse signal vectors and also filter out noise component and interference.
%thereby enhancing the identification quality of the support.
%
%For example, when the MMV model is considered in the sparse channel estimation, we can naturally exploit the property that
Wireless channel estimation is a good example since the support of the channel impulse response does not change much over time and across different MIMO antenna pairs \cite{vetterli,chan_choi,sptp,dist}.
In addition, temporal correlations between multiple source vectors (e.g., correlated fading of the channel gain) can be exploited in the algorithm design\cite{rao_temporal,rao_ar,prasad,choi_shim}.
%
%Wireless channel estimation is also a good example since the channel responses are correlated in time.
%so that the nonzero values as well as the support are correlated in time.
%
%The MMV model is also useful in estimating AoA and AoD parameters

%%%%%%%%%%%%%%%%%%%%%%%%%%%%%%%%%%%%%%%%%%%%%%%%%%%
\begin{table*}
\begin{center}
\caption{Summary of MMV-based sparse recovery algorithms} \label{tb:mmv} \small
\begin{tabular}{ |l|p{3cm}|p{9cm}| }
%\multicolumn{2}{ |c| }{Sparse recovery algorithm ({\bf remove this row})} \\
\hline
Scenario & References & Remark \\ \hline
%\multirow{1}{*}{Scenario 1}
%   & N/A & Easy to solve \\ \hline
\multirow{6}{*}{Scenario 1}
   & SOMP \cite{tropp2006algorithms} & Extension of OMP for the MMV setup. Computational complexity of the SOMP is lower than other candidate algorithms.   \\ \cline{2-3}
   & Convex relaxation \cite{tropp_part2} & Mixed $\ell_1$ norm is used to replace $\ell_0$ norm. The convex optimization package is used for algorithm.  \\ \cline{2-3}
   & MSBL \cite{wipf2007empirical}  & Extention of SBL for the MMV setup. It offers excellent recovery performance but the computational complexity is a bit higher.  \\ \cline{2-3}
   & MUSIC-augmented CS \cite{kiryung}  &  The subspace criterion of  MUSIC algorithm is used to identify the support.
     \\ \cline{2-3}
   & TSBL \cite{rao_temporal}   & Equivalence between block sparsity model and
    MMV model was used to exploit the correlations between the source vectors.   \\ \cline{2-3}
    & AR-SBL \cite{rao_ar}, Kalman-filtered CS \cite{kalmancs}& The multiple source vectors are modeled by auto-regressive process. The support and amplitude of the source vectors is jointly estimated via iterative algorithm.
   \\ \hline
   \multirow{3}{*}{Scenario 2}
   & KSBL \cite{prasad} & The auto-regressive process is used to model the dynamics of the source vectors. Kalman filter is incorporated to estimate the support and gains sequentially. \\ \cline{2-3}
     & AMP-MMV \cite{schniter} & Graphical model is used to describe the variations of the source vectors. Message passing over a part of graph having dense connections is handled via the AMP method \cite{amp1}.  \\ \cline{2-3}
      & sKTS \cite{choi_shim} & The deterministic binary vector is used to model the sparsity structure of the source vectors. The EM algorithm is used for joint estimation of sparsity pattern and gains. \\ \hline
     \multirow{2}{*}{Scenario 3}
     & Modified-CS \cite{modifiedcs} & The new elements added to the support detected in the previous measurement vector is found via $\ell_1$ optimization. The candidates of poor quality are eliminated via thresholding.   \\ \cline{2-3}
   & DCS-AMP \cite{schniter2} & The dynamic change of the support is modeled by the markov process and efficient message passing algorithm based on AMP is applied.  \\ \hline
   \end{tabular}
\end{center}
\end{table*}
%%%%%%%%%%%%%%%%%%%%%%%%%%%%%%%%%%%%%%%%%%%%%%%%%%%

Depending on how the support information is shared among multiple measurement vectors, MMV scenario can be divided into three distinct cases:
\begin{enumerate}
\item The supports of multiple measurement vectors are the same %($\Omega_1 =\cdots = \Omega_N = \Omega$),
    but the values for nonzero positions are distinct. The system matrix for each measurement vector is identical. %for all measurement vectors. %($\mathbf{H}_{1} = \cdots = \mathbf{H}_{N}$).
\item The supports of multiple measurement vectors are the same %($\Omega_1 =\cdots = \Omega_N = \Omega$)
       but the values for nonzero positions are distinct. The system matrix for all measurement vectors are also distinct.
\item The support of multiple measurement vectors are slightly different.
    %(the cardinality of the set difference $(\Omega_i - \Omega_{i-1})$ is small).
\end{enumerate}
%
%
%%%%%%%%%%%%%%%%%%%%%%%%%%%%%%%%%%%%%%%%%%%%%%%%%%
\begin{figure}[t]
\begin{center}
	\includegraphics[width=110mm,height=95mm]{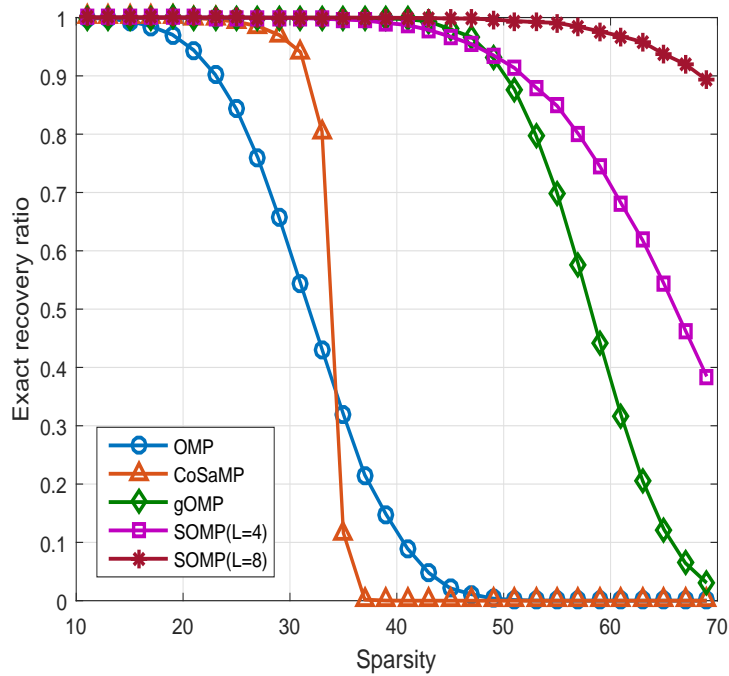}
\caption{Performance comparison between OMP and SOMP ($M=32, N=64, k=8$).}
\label{fig:mmv}
\end{center}
\end{figure}
%%%%%%%%%%%%%%%%%%%%%%%%%%%%%%%%%%%%%%%%%%%%%%%%%%
%
%
\begin{figure} [!]
 \centering
    \subfigure[Scenario 1: supports are the same and the system matrices are also the same.]
  {\epsfig{figure=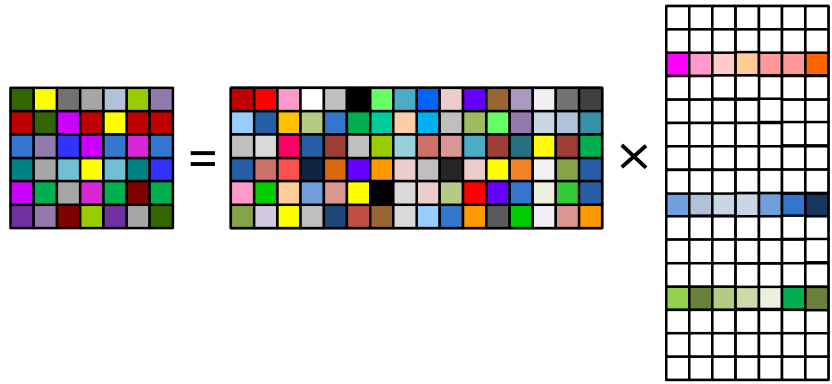, width=110mm}}
   \subfigure[Scenario 2: supports are the same but the system matrices are different.]
  {\epsfig{figure=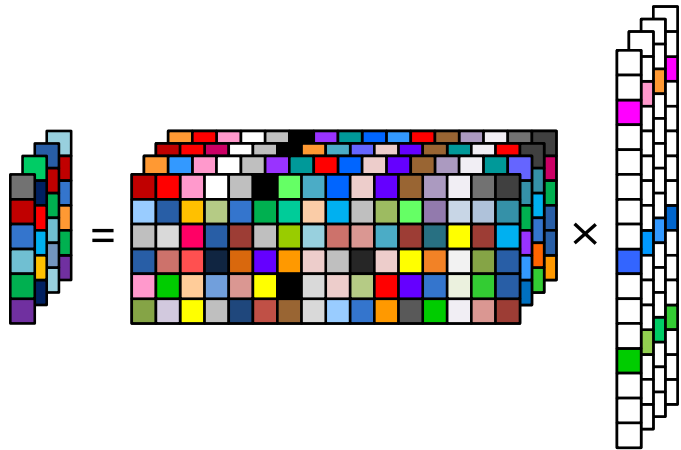, width=90mm}}
   \subfigure[Scenario 3: supports vary slightly in time.]
  {\epsfig{figure=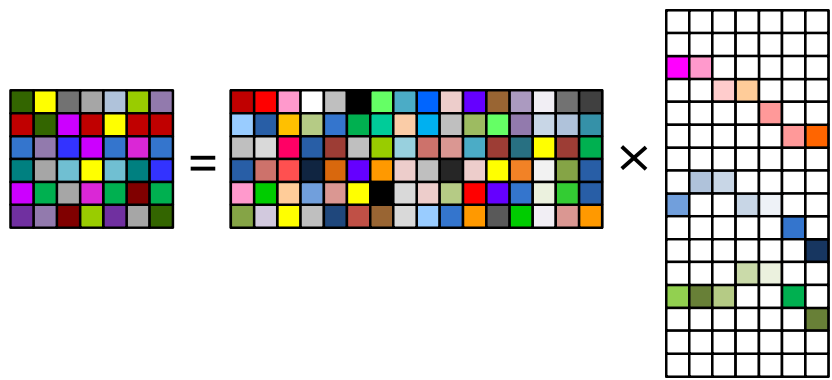, width=110mm}}
  \caption { Illustration of three MMV scenarios.
   } \label{fig:sc}
\end{figure}

%
%In the first scenario, the problem can easily be solved by stacking the measurement vectors in a single vector (i.e., $\mathbf{y} = \matc{\mathbf{y}_1^{T} & \cdots & \mathbf{y}_N^{T}}^{T}$), and then applying a conventional sparse recovery algorithm.
% as if single measurement vector is received.
%
%
The first scenario is the most popular scenario of the MMV problem (see Fig.~\ref{fig:sc}(a)).
In this scenario, we express the measurement vectors as
%\begin{align} \label{eq:mmv}
\begin{eqnarray}
    \mathbf{Y} = \mathbf{H S} + \mathbf{N}
\end{eqnarray}
where
$\mathbf{Y} = [\mathbf{y}_{1} ~ \cdots ~ \mathbf{y}_{N}]$,
$\mathbf{S} = [\mathbf{s}_{1} ~ \cdots ~ \mathbf{s}_{N}]$, and
$\mathbf{N} = [\mathbf{n}_{1} ~ \cdots ~ \mathbf{n}_{N}]$.
The recovery algorithm finds the column indices of $\mathbf{H}$ corresponding to the nonzero row vectors of $\mathbf{S}$ using the measurement matrix $\mathbf{Y}$.
It has been shown from theoretic analysis that the performance of the MMV-based algorithm improves exponentially with the number of measurements \cite{mmv_proof,mmv_proof2}.
In fact, MMV-based sparse recovery algorithms perform much better than the SMV-based recovery algorithms \cite{rahut,rao_mmv,tropp2006algorithms,wipf2007empirical}.
Various recovery algorithms have been proposed for MMV scenario.
In \cite{tropp_part2}, the convex relaxation method based on the mixed norm has been proposed. In \cite{tropp2006algorithms}, the greedy algorithm called simultaneous OMP (SOMP) is proposed.
Statistical sparse estimation techniques for MMV scenario include MSBL \cite{wipf2007empirical}, AR-SBL \cite{rao_ar}, and TSBL \cite{rao_temporal}.
In \cite{kiryung}, an approach to identify the direction of arrival (DoA) in array signal processing using the MMV model has been investigated.
%Using the close connection between the DoA estimation problem and the MMV model, the recovery algorithms are devised such that the subspace criterion of the MUSIC algorithm is augmented with the CS recovery algorithm.
%
Further improvement in the recovery performance can be achieved by exploiting the statistical correlations between the signal amplitudes  \cite{rao_temporal,rao_ar,kalmancs}.

The second scenario is a generalized version of the first scenario in the sense that system matrices are different for all measurement vectors. Extensions of OMP algorithm \cite{duarte2005distributed}, iteratively reweighted algorithm \cite{ding2015joint}, sparse Bayesian learning algorithm \cite{ji2009multitask}, and Kalman-based sparse  recovery algorithm \cite{choi_shim} has been proposed. In \cite{schniter}, it has been shown that the graph-based inference method is effective in this scenario.
In the third scenario, the recovery algorithms need to keep track of temporal variations of the support since the sparsity pattern changes over time. However, since the change is usually very small, the sparsity pattern can be tracked by estimating the difference between two support sets for consecutive measurement vectors \cite{modifiedcs,lau}. The algorithm employing approximate message passing (AMP) is used for this scenario in \cite{schniter2}. In Table \ref{tb:mmv}, we summarize the recovery algorithms based on the MMV model.

In summary, main point in this subsection is that having multiple {\it correlated} measurements is very helpful
% takes full advantage of correlation among multiple measurements
in improving the performance.
Depending on system setup and model, one may achieve sub-linear, linear, or super-linear performance gain proportional to the number of measurements.  Often we spend lots of time on the algorithm selection yet have hard time satisfying required performance. Remembering that good recovery performance is what we want at the end of the day, easy and practical way to achieve the goal is to use multiple measurements (see Fig. \ref{fig:mmv}). In fact, in many static or slowly-varying environments, measurements of the previous sample time are largely correlated to the current ones and one can simply use them for better performance.
To come up with an efficient CS technique for the time-varying scenario is important research problem.

% =============================================
\subsection{Can We Do Better If Integer Constraint Is Given?}
\label{sec:sparsedet}
% =============================================

%%%%%%%%%%%%%%%%%%%%%%%%%%%%%%%%%%%%%%%%%%%%%%%%%%
\begin{figure}[t]
\begin{center}
	\includegraphics[width=110mm,height=95mm]{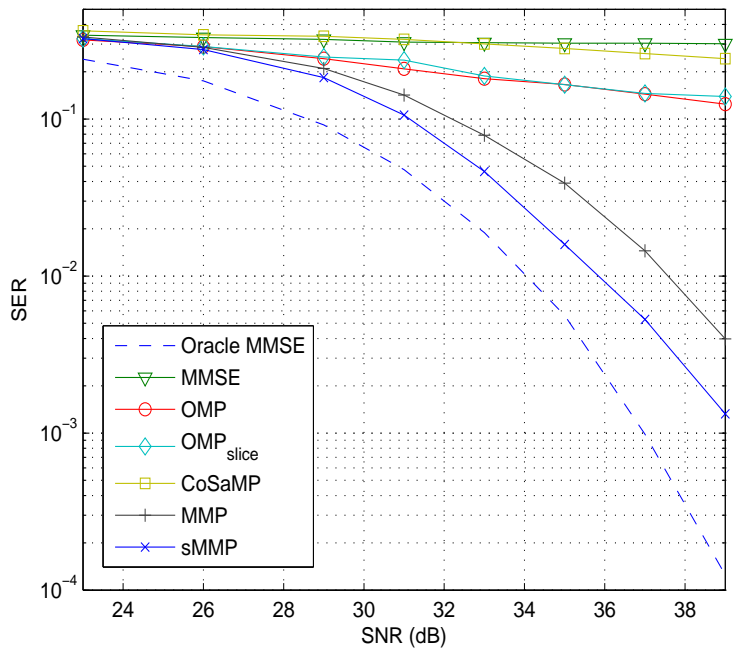}
\caption{SER performance when the nonzero positions of input vector is chosen from 16-quadrature amplitude modulation (QAM)  ($m=12, n=24, k=5$). sMMP refers to the MMP algorithm equipped with slicing operation.}
\label{fig:ser}
\end{center}
\end{figure}
%%%%%%%%%%%%%%%%%%%%%%%%%%%%%%%%%%%%%%%%%%%%%%%%%%

When the nonzero elements of the target vector $\mathbf{s}$ are from the set of finite alphabets, one can exploit this information for the better reconstruction of the sparse vector.
%
%One example for this scenario is the symbol detection in IoT network \cite{zhu}.
In order to incorporate the integer constraint into the sparse input vector, one can either modify the conventional detection algorithm or incorporate an integer constraint into the sparse recovery algorithm.

First, when the detection approach is used, one can simply add the zero into the constellation set $\Theta$. For example, if nonzero elements of $\mathbf{s}$ are chosen from binary phase shift keying (BPSK) (i.e., $s_i \in \Theta = \{ -1, 1 \}$), then the modified constellation set becomes $\Theta' = \{-1, 0, 1 \}$. Sparsity constraint $\| \mathbf{s} \|_0 = k$ can also be used to limit the search space of the detection algorithm. For example, if the maximum likelihood (ML) tree search (e.g., sphere decoding) is performed, one can stop the search if the cardinality of the symbol vector equals the sparsity.

On the other hand, when the sparse recovery algorithm is used, one should incorporate the quantization step to map real (complex) value into the symbol.
As a simple example, if the OMP algorithm is used, $\Theta = \{-1, 1 \}$, and $\hat{\mathbf{s}}_i = [0.7 ~ -0.6]^T$, then the quantized output becomes $Q_{\Theta} (\hat{\mathbf{s}}_i) = [1 ~ -1]^T$.
Note, however, that just using the quantized output might not be effective, in particular for the sequential greedy algorithms due to the error propagation.
For example, if an index is chosen incorrectly in one iteration, then the estimate will also be incorrect and thus the quantized output will bring additional quantization error, deteriorating the subsequent detection process.
In this case, parallel tree search strategy can be a good option to alleviate the error propagation phenomenon. For example, a tree search algorithm performs the parallel search to find multiple promising candidates (e.g., MMP \cite{kwon}). Among the multiple candidates, the best one minimizing the residual magnitude is chosen in the last minute ($\hat{\mathbf{s}} = \arg \min_{\| \mathbf{s} \|_0 = k , \mathbf{s} \in \Theta^k} \| \mathbf{y} - \mathbf{H s} \|$).
The main benefit of tree search method, from the perspective of incorporating the integer slicer, is that it deteriorates the quality of incorrect candidate yet enhances the quality of correct one. This is because the quality of incorrect candidates gets worse due to the additional quantization noise caused by the slicing while no such phenomenon happens to be the correct one (note that the quantization error is zero for the correct symbol).
As a result, as shown in Fig. \ref{fig:ser}, the recovery algorithms accounting for the integer constraint of the symbol outperform those without considering this property.

Moral of the story in this subsection is that if additional hint exists, then one can actively use it for better performance. Integer constraint discussed above would be a useful hint in the detection of the sparse vector. However, these hints might be a double-edged sword; if it is not used properly, it will do harm. Thus, one should use with caution in order not to worsen the performance. Since the parallel detection strategy we mentioned has cost issue, one can come up with solution to trade-off the cost and performance. Soft quantization can be one option. Iterative detection using a prior information from the decoder would also be a viable option to be investigated.

% =============================================
\subsection{Should We Know Sparsity a Priori?}
\label{sec:shouldwe}
% =============================================
%
Some algorithm requires the sparsity of an input signal while others do not need this. For example, sparsity information is unnecessary for the $\ell_1$-norm minimization approaches but many greedy algorithms need this since the sparsity is used as stoping criteria of the algorithm. When needed, one should estimate the sparsity using various heuristics. Before we discuss on this, we need to consider what will happen if the sparsity information is incorrect.
In a nutshell, setting the number of iterations not being equivalent to the sparsity leads to either early or late termination of the greedy algorithm. In the former case (i.e., {\it underfitting} scenario), the desired signal will not be fully recovered  while some of the noise vector is treated as a desired signal for the latter case (i.e., {\it overfitting} scenario). Both cases are undesirable, but performance loss is typically more severe for underfitting due to the loss of the signal. Thus, it might be safe to use slightly higher sparsity, especially when the noise effect is not that severe. For example, if the sparsity estimate is $\hat{k}$, one can take $1.2\hat{k}$ as an iteration number of OMP.

As a sparsity estimation strategy, the residual-based stopping criterion and cross validation \cite{CV} are well-known.
The residual based stopping criterion is widely used to identify the sparsity level (or iteration number) of the greedy algorithm.
Basically, this scheme terminates the algorithm when the residual power is smaller than the pre-specified threshold $\epsilon$ (i.e., $\| \mathbf{r}^i \|_2 < \epsilon$). The iteration number at the termination point is set to the sparsity level.
However, since the residual magnitude decreases monotonically and the rate of decay depends on the system parameters,
it might not be easy to identify the optimal terminating point.
%In the residual based stopping criterion, the iteration number where the residual decrease becomes slow is chosen as a sparsity.
%

Cross validation (CV) is another technique to identify the model order (sparsity level $k$ in this case) \cite{CV}.
In this scheme, the measurement vector $\mathbf{y}$ are divided into two parts: a training vector $\mathbf{y}^{(t)}$ and a validation vector $\mathbf{y}^{(v)}$. In the first step, we generate a sequence of possible estimates $\hat{\mathbf{s}}_1 , \cdots , \hat{\mathbf{s}}_n$ using a training vector $\mathbf{y}^{(t)}$, where $\hat{\mathbf{s}}_i$ denotes the estimate of  $\mathbf{s}$ obtained under the assumption that the sparsity equals $i$.
In the second step, the sparsity is predicted using the validation vector $\mathbf{y}^{(v)}$.
Specifically, for each estimate $\hat{\mathbf{s}}_i$, the validation error $\epsilon_i = \| \mathbf{y}^{(v)} - \mathbf{H}^{(v)} \hat{\mathbf{s}}_i \|_2$ is computed. %for each estimate $\hat{\mathbf{s}}_i$.
Initially, when the count $i$ increases, the quality of the estimate improves since more signal elements are added and thus the validation error $\epsilon_i$ decreases. However, when the count $i$ exceeds the sparsity, that is, when we choose more columns than needed, we observe no more decrease in $\epsilon_i$ and just noise will be added to the validation error.
Since the validation error has a convex shape, the number generating the minimum validation error is returned as the sparsity estimate ($\hat{k} = \arg \min_i \epsilon_i$).

When one uses the CS algorithm in which the sparsity is used as an input parameter, one needs to be aware of the sparsity information. More importantly, in order to decide whether the CS technique is useful or not, one should know if the target signal vector is sparse. There are a variety of ways to check the sparsity. In order to pursue an accurate evaluation of the sparsity, it is desired to use a bit expensive option like cross validation. In many stationary situations, fortunately, one can do it once and use it for a while.

\section{Conclusion and Future Direction}	
\label{sec:conclusion}
%%%%%%%%%%%%%%%%%%%%%%%%%%%%%%%%%%%%%%%%%%%%%%%%%%%%%%%%%%%%%%%%%%%%%%%%%%%%
In this article, we have provided an overview of the CS technique for wireless communication systems.
%Sparsity exists in transmit signal vector, channel response, impulse noise, lis everywhere in wireless commuication
%
We discussed basics of CS techniques, three subproblems of CS related to wireless communications, and various wireless applications which CS techniques can be applied to. We also discussed several main issues that one should be aware of and subtle points that one should pay attention to.
%
%This field is growing fast and there are still much work
There are a broad class of wireless applications to which the CS technique would be beneficial and much work remains to be pursued. We list here some of future research directions.
%
% ------------------------------------------------
\begin{itemize}
\item It would be interesting to design a flexible CS technique and sparse recovery algorithm that can adapt to various and diverse wireless environments and input conditions.
    %In wireless communications, there are a variety of underlying structure that can be exploited in deriving CS algorithms. Various channel characteristics, system models, and network connectivity can be accounted in the formulation of CS problem.
    It is very difficult and cumbersome task for wireless practitioners to do calibration and tuning on the various channel conditions and parameters (e.g., measurement size and input sparsity). Thus, it would be very useful to come up with solution that adapts to diverse wireless environments.
\item Most of the system matrices discussed in the CS theory are Gaussian and random matrices. In many real situations, however, we cannot use the ``purely" random matrices and need to use the deterministic matrices. Unfortunately, there is no well-known guideline or design principle for the deterministic sensing matrix. Thus, it would be good to have practical design guideline for the system using the deterministic matrices. Related to this, it would be useful to come up with a simple and effective dictionary learning technique suited for wireless communication systems.
\item Most of performance metrics in wireless communication systems are statistical (e.g., BER, BLER) and it would be useful to have more realistic analysis tool based on statistical approach. As mentioned, analytic tools to date (e.g., RIP or mutual coherence) are too stringent for wireless communication purpose and it would be good to have more flexible tool that bridges the theory and practice.
\item It might be worth investigating if the machine learning technique would be helpful in classifying whether the CS technique can be applied or not in a given scenario. Note that it is very difficult to judge if the CS technique is effective in the the complicated wireless communication scenarios. Since we already have large collection of data at the base-station, state-of-the-art machine learning approach like deep neural network (DNN) might be useful to extract a simple answer to this problem.
\item What if the system matrix is sparse, not the input vector. In this work, we have primarily discussed the scenario where the desired input vector is sparse. But there are some scenarios where the input-output connection is sparse, not the input vector (e.g., massive multiple access scheme). Solution of this problem will be useful in various machine-type applications. Related to this, applying the matrix completion technique to wireless communication problem would also be interesting direction to be pursued.
\item As the dimension of wireless communication systems increases,  design time and computational complexity  increase significantly. We often want an algorithm whose complexity scales linearly with the problem size.
    %We did not have thorough discussion on the implementation issue in this work.
    Without doubt, for the successful development of CS-based wireless communication systems, low-complexity and fast implementation is of great importance. Development of implementation-friendly algorithm and architecture would expedite the commercialization of CS-based wireless systems.

\end{itemize}
% ------------------------------------------------

As a final remark, we hope that this article will serve as a useful guide for wireless communication researchers, in particular for those who want to use the CS technique to grasp the gist of this interesting paradigm. Since our treatment in this paper is rather casual and non-analytical, one should dig into details with further study. However, if the readers learn in mind that essential knowledge in a pile of information is always {\it sparse}, the journey will not be that tough.

%%%%%%%%%%%%%%%%%%%%%%%%%%%%%%%%%%%%%%%%%%%%%%%%%%%%%%%%%%%%%%%%%%%%%%%%%%%%
\section*{Acknowledgment}
%%%%%%%%%%%%%%%%%%%%%%%%%%%%%%%%%%%%%%%%%%%%%%%%%%%%%%%%%%%%%%%%%%%%%%%%%%%%
The authors would like to thank J. Kim and H. Seo for their help in simulations and  the anonymous reviewers for their invaluable comments and suggestions, which helped  to improve the quality of the paper.

%%%%%%%%%%%%%%%%%%%%%%%%%%%%%%%%%%%%%%%%%%%%%%%%%%%%%%%%%%%%%%%%%%%%%%%%%%%%
%\appendices
%%%%%%%%%%%%%%%%%%%%%%%%%%%%%%%%%%%%%%%%%%%%%%%%%%%%%%%%%%%%%%%%%%%%%%%%%%%%

\vspace{1cm}
%\clearpage

%%%%%%%%%%%%%%%%%%%%%%%%%%%%%%%%%%%%%%%%%%%%%%%%%%%%%
\section{Proof of Theorem 1}
%%%%%%%%%%%%%%%%%%%%%%%%%%%%%%%%%%%%%%%%%%%%%%%%%%%%%
%\begin{proof}
We first consider {\it only if} case. We assume that there are more than one vector, say $\mathbf{s}_1$ and $\mathbf{s}_2$, satisfying $\mathbf{y} = \mathbf{Hs}$ and both have at most $k$ nonzero elements. Then, by letting $\mathbf{u} = \mathbf{s}_1 - \mathbf{s}_2$, we have $\mathbf{Hu} = \mathbf{0}$ where $u$ has at most $2k$ nonzero elements.
Since $spark(\mathbf{H}) > 2k$ from the hypothesis, any $2k$ columns in $\mathbf{H}$ are linearly independent, implying that $\mathbf{u} = \mathbf{0}$ and hence $\mathbf{s}_1 = \mathbf{s}_2$.
We next consider the {\it if} case. We assume that, for given $\mathbf{y}$, there exists at most one $k$-sparse signal $\mathbf{s}$ satisfying $\mathbf{y} = \mathbf{H} \mathbf{s}$ and $spark( \mathbf{H} ) \leq 2k$.
Under this hypothesis, there exist a set of $2k$ columns that are linearly dependent, implying that there exists $2k$-sparse vector $\mathbf{u}$ in $N( \mathbf{H} )$ (i.e., $\mathbf{H u} = \mathbf{0}$). Since $\mathbf{u}$ is $2k$-sparse, we can express it into the difference of two $k$-sparse vectors $s_1$ and $s_2$ ($\mathbf{u} = \mathbf{s}_1 - \mathbf{s}_2$).
Since $\mathbf{H u} = \mathbf{0}$, $\mathbf{H}(\mathbf{s}_1 - \mathbf{s}_2) = \mathbf{0}$ and hence $\mathbf{H} \mathbf{s}_1 = \mathbf{H} \mathbf{s}_2$, which contradicts the hypothesis that there is at most one $k$-sparse solution satisfying $\mathbf{y} = \mathbf{H} \mathbf{s}$. Thus, we should have $spark ( \mathbf{H} ) > 2k$.


\begin{thebibliography}{1}

\bibitem{donoho} D. Donoho, ``Compressed sensing," {\it IEEE Trans. Inf. Theory}, vol. 52, no. 4, pp. 1289-1306, April 2006.

    \bibitem{stable} E.\ Candes, J.\ Romberg, and T.\ Tao, ``Stable signal recovery from incomplete and inaccurate measurements," {\it Comm. Pure Appl. Math.}, vol. 59, no. 8, pp.\ 1207-1223, Aug. 2006.

\bibitem{shannon_94Rauhut} A. V. Oppenheim and R. W. Schafer, {\it Discrete-time Signal Processing,} Prentice Hall, 2010.


\bibitem{eldar_book} Y.\ C.\ Eldar and G.\ Kutyniok, {\it Compressed Sensing : Theory and Applications}, Cambridge Univ. Press, 2012.

\bibitem{cs_magazine} E.\ J.\ Candes and M.\ B.\ Wakin, ``An introduction to compressive sampling," {\it IEEE Signal Processing Mag.}, vol. 25, pp. 21-30, March 2008.

\bibitem{zhu_han} Z. Han, H. Li, and W. Yin, {\it Compressive Sensing for Wireless Networks}, Cambridge Univ. Press, 2013.

\bibitem{baraniuk_magazine} R. G. Baraniuk, ``Compressive sensing," {\it IEEE Signal Processing Mag.,} vol. 24, pp. 118-121, July 2007.

\bibitem{rauhut} S. Foucart and H. Rauhut, {\it A Mathematical Introduction to Compressive Sensing}, Applied and Numerical Harmonic Analysis, Birkhauser, 2013.

\bibitem{fornasier} M. Fornasier and H. Rauhut. {\it Compressive sensing}, In Handbook of Mathematical Methods in Imaging, pp 187-228. Springer, 2011.

\bibitem{hayashi}  K. Hayashi, M. Nagahara, and T. Tanaka, ``A user's guide to compressed sensing for communications systems," {\it IEICE Trans. Commun.}, vol. E96-B, no. 3, pp. 685-712, March 2013.

\bibitem{bpdn} S.\ S.\ Chen, D.\ L.\ Donoho, and M.\ A.\ Saunders, ``Atomic decomposition by basis pursuit," {\it SIAM J. Scientif. Comput.}, vol. 20, no. 1, pp.\ 33-61, 1998.

        \bibitem{uncertainty} E. J. Candes, J. Romberg, and T. Tao, ``Robust uncertainty principles: exact signal reconstruction from highly incomplete frequency information," {\it IEEE Trans. Inf. Theory}, vol. 52, pp. 489-509, Feb. 2006.

\bibitem{boyd} S. J. Kim, K. Koh, M. Lustig, S. Boyd, and D. Gorinevsky, ``An interior-point method for large-scale $\ell_1$-regularized least squares," {\it IEEE Journal on Sel. Topics in Signal Process.} vol. 1, pp. 606-617, Dec. 2007.

\bibitem{cvx} http://cvxr.com/cvx/

\bibitem{l1magic} http://users.ece.gatech.edu/justin/l1magic/

%\bibitem{structure} M. F. Duarte and Y. C. Eldar, ``Structured compressed sensing: from theory to applications," {\it IEEE Trans. Signal Process.}, vol. 59, no. 9, pp. 4053-4085, Sept. 2011.

\bibitem{omp} J.\ A.\ Tropp and A.\ C.\ Gilbert, ``Signal recovery from random measurements via orthogonal matching pursuit," {\it IEEE Trans. Inf. Theory}, vol. 53, no. 12, pp. 4655-4666, Dec. 2007.

\bibitem{jian16} J. Wang and B. Shim, ``Exact recovery of sparse signals using orthogonal matching pursuit: How many iterations do we need?," {\it IEEE Trans. Signal Process.}, vol 64, no. 16, pp. 4194-4202, Aug. 2016.

\bibitem{Strohmer} T. Strohmer and R.\ Heath, ``Grassmannian frames with applications to coding and communication measurements," {\it Appl. Comput. Harmon. Anal.}, vol. 53, vol. 14, pp. 257-275, 2004.

\bibitem{DonohoElad} D. Donoho and M. Elad, ``Optimally sparse representation in general (non-orthogonal) dictionaries via L1 minimization," Proc. {\it Natl. Acad. Sci.}, vol. 100, pp. 2197-2202, 2003.

\bibitem{gomp} J. Wang, S. Kwon, and B. Shim, ``Generalized orthogonal matching pursuit," {\it IEEE Trans. Signal Process.}, vol. 60, no. 12, pp. 6202-6216, Dec. 2012.

\bibitem{sp} W. Dai and O. Milenkovic, ``Subspace pursuit for compressive sensing signal reconstruction," {\it IEEE Trans. Inf. Theory}, vol. 55, pp. 2230-2249, May 2009.

\bibitem{tzhang} T. Zhang, ``Sparse recovery with orthogonal matching pursuit under RIP," IEEE Trans. Inf. Theory, vol. 57, no. 9, pp. 6215.6221, Sept. 2011.

\bibitem{berger_mag} C. R. Berger and W. Zhaohui, H. Jianzhong, and Z. Shengli, ``Application of compressive
sensing to sparse channel estimation," {\it IEEE Commun. Magazine}, vol. 48, no. 11, pp. 164-174, Nov. 2010.

\bibitem{paredes2007ultra} J.L. Paredes, G.R. Arce and Z. Wang, ``Ultra-wideband compressed sensing: channel estimation," {\it IEEE Jounal of Selected Topics in Signal Process.}, vol. 1, no. 3, pp. 383-395, Oct. 2007.

\bibitem{michelusi2012uwb1} N. Michelusi, U. Mitra, A.F. Molisch and M. Zorzi, ``UWB sparse/diffuse channels, part I: Channel models and Bayesian estimators," {\it IEEE Trans. Signal Process.}, vol. 60, no. 10, pp. 5307-5319, Oct. 2012.

\bibitem{michelusi2012uwb2} N. Michelusi, U. Mitra, A.F. Molisch and M. Zorzi, ``UWB sparse/diffuse channels, part II: Estimator analysis and practical channels," {\it IEEE Trans. Signal Process.}, vol. 60, no. 10, pp. 5320-5333, Oct. 2012.

\bibitem{cs_uwa} W.\ Li and J.\ C.\ Preisig, ``Estimation of rapidly time-varying sparse channels," {\it IEEE J. Oceanic Eng.}, vol. 32, pp. 927-939, Oct. 2007.

\bibitem{berger} C.\ R.\ Berger, S.\ Zhou, J.\ C.\ Preisig and P.\ Willett, ``Sparse channel estimation for multicarrier underwater acoustic communication: from subspace methods to compressed sensing," {\it IEEE Trans. Signal Process.}, vol. 58, pp. 1708-1721, March 2010.

\bibitem{bajwa2010compressed} W. U. Bajwa, J. Haupt, A. M. Sayeed, and R. Nowak, ``Compressed channel sensing: A new approach to estimating sparse multipath channels," {\it Proceedings of the IEEE}, vol. 98, no. 6, pp. 1058-1076, June 2010.

\bibitem{LTE_standard} 3GPP TS 36.101. ``User Equipment (UE) Radio Transmission and Reception." {\it 3rd Generation Partnership Project; Technical Specification Group Radio Access Network (E-UTRA)}. URL: http://www.3gpp.org.

\bibitem{dist} X. Rao, and V. K. N. Lau, ``Distributed compressive CSIT estimation and feedback for FDD multi-user massive MIMO systems," {\it  IEEE Trans, Signal Process.}, vol. 62, no. 12, pp. 3261-3271, Jun. 2014.

\bibitem{prior} X. Rao and V. K. N. Lau, ``Compressive sensing with prior support quality information and application to massive MIMO channel estimation with temporal correlation," {\it IEEE Trans. Signal Process.}, vol. 63, no. 18, pp. 4914-4924, Sept. 2015.

\bibitem{rappaport} T. S. Rappaport, S. Shu, R. Mayzus, Z. Hang, Y. Azar, K. Wang, G. N. Wang, M. Samimi, and F. Gutierrez, ``Millimeter wave mobile communications for 5G cellular: It will work!,", {\it IEEE Access}, vol. 1, pp. 335-349, 2013.

\bibitem{el2014spatially} O. E. Ayach, S. Rajagopal, S. Abu-Surra, and P. Zhouyue, and R. W. Heath, ``Spatially sparse precoding in millimeter wave MIMO systems," {\it IEEE Trans. Wireless Commun.}, vol. 13, no. 3, pp. 1499-1513, March 2014.

\bibitem{mmwave} A. Alkhateeb, O. E. Ayach, G. Leus, and R. W. Heath, ``Channel estimation and hybrid precoding for millimeter wave cellular systems," {\it IEEE J. Sel. Topics Signal Process.}, vol. 8, pp. 831-846, Oct. 2014.

\bibitem{sayeed2002deconstructing} A. M. Sayeed, ``Deconstructing multiantenna fading channels," {\it IEEE Trans. Signal Process.}, vol. 50, no. 10, pp. 2563-2579, Oct. 2002.

\bibitem{Tse_book} D. Tse and P. Viswanath, {\it Fundamentals of Wireless Communication}, Cambridge Univ. Press, 2005.

\bibitem{Van_book} H.L. Van Trees, {\it Detection, estimation, and modulation theory, optimum array processing} John Wiley \& Sons, 2004.

\bibitem{akdeniz2014millimeter} M. R. Akdeniz, Y. Liu, M. K. Samimi, S. Sun, S. Rangan, T. S. Rappaport, and E. Erkip, ``Millimeter wave channel modeling and cellular capacity evaluation," {\it IEEE Journal on Selected Areas in Communications}, vol. 32, no. 6, pp. 1164-1179, June 2014.

\bibitem{samimi20163} M.K. Samimi and T.S. Rappaport, ``3-d millimeter-wave statistical channel model for 5G wireless system design," {\it IEEE Trans. Microwave Theory and Techniques}, vol. 64, no. 7, pp. 2207-2225, July 2016.

\bibitem{meng} J. Meng, Y. Yin, Y. Li, N. T. Nguyen, and Z. Han, ``Compressive sensing based high-resolution channel estimation for OFDM system," {\it IEEE Journal of Selected Topics in Signal Process.}, vol. 7, no. 1, pp. 15-25, Feb. 2012.

\bibitem{ding} W. Ding, C. Pan, L. Dai, and J. Song,  ``Compressive sensing based channel estimation for OFDM systems under long delay channels," {\it IEEE Trans. Broadcasting}, vol. 60, no. 2, pp. 313-321, June 2014.

\bibitem{hlawatsch} G. Taubock and F. Hlawatsch, ``A compressed sensing technique for OFDM channel estimation in mobile environments: exploiting channel sparsity for reducing pilots," {\it IEEE Proc. ICASSP 2008}, April 2008, pp. 2885-2888.

\bibitem{chan_choi} J. W. Choi, B. Shim, and S. Chang, ``Downlink Pilot Reduction for Massive MIMO Systems via Compressed Sensing," {\it IEEE Commun. Letter}, vol. 19, no. 11, pp. 1889-1892, Nov. 2015.

\bibitem{cotter2002sparse} S. F. Cotter,  and B. D. Rao. ``Sparse channel estimation via matching pursuit with application to equalization," {\it IEEE Trans. Commun.}, vol.  50 no. 3, pp.  374-377, March 2002.

    \bibitem{prasad} R. Prasad, C. R. Murphy and B. D. Rao, ``Joint approximately sparse channel estimation and data detection in OFDM systems using sparse Bayesian learning," {\it IEEE Trans. Signal Process.}, vol. 62, no. 14, pp. 3591-3603, July 2014.

\bibitem{gao} Z. Gao, L. Dai, W. Dai, B. Shim, and Z. Wang, ``Structured compressive sensing-based spatio-temporal joint channel estimation for FDD massive MIMO," {\it IEEE Trans. Commun.}, vol. 64, no. 2, pp. 601-617, Feb. 2016.

    \bibitem{lau} X. Rao and V. K. N. Lau, ``Distributed compressive CSIT estimation and feedback for FDD multi-user massive MIMO systems," {\it IEEE Trans. Signal Process.},  vol. 62, no. 12, pp. 3261-3271,  June 2014.

\bibitem{gao2} Z. Gao, L. Dai, Z. Wang, and S. Chen, ``Spatially common sparsity based adaptive channel estimation and feedback for FDD massive MIMO," {\it IEEE Trans. Signal Process.}, vol. 63, no. 23, pp. 6169-6183, Dec. 2015.

\bibitem{naffouri} M. Masood, L. H. Afify, and T. Y. Al-Naffouri, ``Efficient coordinated recovery of sparse channels in massive MIMO," {\it IEEE Trans. Signal Process.}, vol. 63, no. 1, pp. 104-118, Jan. 2015.

\bibitem{Tareq_Al_Naffouri11} T. Al-Naffouri, A. Quadeer, F. Al-Shaalan, and H. Hmida, ``Impulsive noise estimation and cancellation in DSL using compressive sampling," {\it Proc. IEEE International Symp. Circuits and Systems (ISCAS)}, May 2011, pp. 2133-2136.

\bibitem{ana} A. B. Ramirez, R. E. Carrillo, G. Arce, K. E. Barner, and B. Sadler, ``An overview of robust compressive sensing of sparse signals in impulsive noise," {\it Proc. European Signal Process. Conf. (EUSIPCO) 2015}, pp. 2859-2863.

\bibitem{Moon} T.\ K.\ Moon and W.\ Stirling, {\it Mathematical Methods And Algorithms For Signal Processing}, Prentice Hall, 1999.

\bibitem{cr_mag} B. Wang and K. J. R. Liu,  ``Advances in cognitive radio networks: a survey," {\it IEEE Journal of Selected Topics on Signal Process.}, vol. 5, no. 1, pp. 5-23, Nov. 2010.
	
\bibitem{eldar} D. Cohen and Y. C. Eldar, ``Sub-Nyquist sampling for power spectrum sensing in cognitive radios: a unified approach," {\it IEEE Trans. Signal  Process.}, vol. 62, no. 15, pp. 3897-3910, August 2014.

\bibitem{geert} S. Romero and G. Leus, ``Wideband spectrum sensing from compressed measurements using spectral prior information," {\it IEEE Trans. Signal Process.}, vol. 61, no. 24, pp. 6232-6246, Dec. 2013.

\bibitem{iot2} L. Atzori, A. Lera, and G. Morabito, ``The Internet of Things: A survey," {\it Computer Networks}, vol. 54, no. 11, pp. 2787-2805, Oct. 2010.



\bibitem{dan} D. Liu, Q. Zhou, Z. Zhang, and B. Liu, ``Cluster-based energy-efficient transmission using a new hybrid compressed
    sensing in WSN," {\it IEEE Proc. INFOCOM 2016}, pp. 372-376.	
	
\bibitem{shancang} S. Li, L. D. Xu, and X. Wang, ``Compressed sensing signal and data acquisition in wireless sensor networks and internet of things," {\it IEEE Trans. Industrial Informatics}, vol. 9, no. 4, pp. 2177-2186, Nov. 2013.

 \bibitem{melodia} S. Pudlewski, A. Prasanna, and T. Melodia, ``Compressed-sensing-enabled
video streaming for wireless multimedia sensor networks," {\it IEEE Trans. Mobile Computing}, vol. 11, no. 6, pp. 1060-1072, Jun. 2012.

\bibitem{feng} C. Feng, W. S. A. Au, S. Valalee, and Z. Tan, ``Compressive sensing-based
positioning using RSS of WLAN access points," {\it IEEE Proc. INFOCOM 2010}, pp. 1-9.

\bibitem{bowu} B. Zhang, X. Cheng, N. Zhang, Y. Cui, Y. Li, Q. Liang, ``Sparse target counting and localization in sensor networks based on compressed sensing," {\it IEEE Proc. INFOCOM 2011}, pp. 2255-2263.

    \bibitem{rss_cs} C. Feng, W. S. . Au, S. Valaee, and Z. Tan, ``Received-signal-strength-based indoor positioning using compressive sensing," {\it IEEE Trans. Mobile Computing}, vol. 11, no. 12, pp. 1983-1993, Dec. 2012.

        \bibitem{lcs} J. Wang, D. Fang, X. Chen, Z. Yang, T. Xing,  and L. Cai, ``LCS: compressive sensing based device-free localization for multiple targets in sensor networks," {\it IEEE Proc. INFOCOM 2013}, pp. 145-149.

    \bibitem{rss} H. Liu, H. Darabi, P. Banerjee, and J. Liu, ``Survey of wireless indoor positioning techniques and systems," {\it IEEE Trans. Systems, Man, and Cybernetics, Part C}, vol. 37, no. 6, pp. 1067-1080, Nov. 2007.

        \bibitem{fingerprinting} M. Bshara, U. Orguner, F. Gustafsson, and L. V. Biesen, ``Fingerprinting localization in wireless networks based on received-signal-strength measurements: a case study on WiMAX networks," {\it IEEE Trans. Vehicular Technology}, vol. 59, no. 1, pp. 283-294,  Jan. 2010.



\bibitem{congwang} C. Wang, B. Zhang, K. Ren, J. M. Roveda, C. W. Chen, and Z. Xu, ``A privacy-aware cloud-assisted healthcare monitoring system via compressive sensing," {\it IEEE Proc. INFOCOM 2014}, pp. 2130-2138.

      \bibitem{huang} H. Hong, S. Misra, W. Tang, H. Barani, and H. Al-Azzawi, ``Applications of compressed sensing in communications networks,"  arXiv preprint arXiv:1305.3002, 2013.

\bibitem{iotlte} 3GPP Technical Report (TR) 45.820, ``Cellular system support for ultralow complexity and low throughput Internet of Things (CIoT)."

    \bibitem{metis} METIS D6.6, ``Final report on the METIS system concept and technology road map," 2015.

\bibitem{scma} X. Wang, ``Sparse coding multiple (SCMA)," {\it IEEE }, vol. 5, pp. 5-23, Nov. 2015.

\bibitem{correal} N. Patwari, J. N. Ash, S. Kyperountas, A. O. Hero III, R. L. Moses, and N. S. Correal, ``Locating the nodes: cooperative localization in wireless sensor network" {\it IEEE Signal Proc. Magazine}, pp. 54-69, July 2005.

\bibitem{mmwave_overview} R. W. Heath, N. Gonzalez-Prelcic, S. Rangan, W. Roh, and A. M. Sayeed, ``An overview of signal processing techniques for milimeter wave MIMO
    systems," {\it IEEE Journal of Selected Topics Signal Process.}, vol. 10, no. 3, pp. 436-453, Feb. 2016.
	
\bibitem{globe14} J.\ W. Choi and B.\ Shim, ``New approach for massive MIMO detection using sparse error recovery," Proc. of {\it IEEE GLOBECOM conf.} 2014.

\bibitem{candes2011compressed} E. J. Candes, Y. C. Eldar, D. Needell, and P. Randall, ``Compressed sensing with coherent and redundant dictionaries," {\it Applied and Comp. Harmonic Anal.}, vol. 31, no. 1, pp. 59-73, July 2011.

\bibitem{kreutz2003dictionary}  K. Kreutz-Delgado, J. F. Murray, B. D. Rao, K. Engan, T. Lee, and T. J. Sejnowski, ``Dictionary learning algorithms for sparse representation,"
{\it Neural Computation}, vol. 15, no. 2, pp. 349-396, Feb. 2003.
	
\bibitem{aharon2006img} M. Aharon, M. Elad, and A. Bruckstein, ``K-SVD: An algorithm for designing overcomplete dictionaries for sparse representation," {\it IEEE Trans. Signal Process.}, vol. 54, no. 11, pp. 4311-4322, Nov. 2006.

\bibitem{ding2015channel} Y. Ding and B. D. Rao, ``Channel estimation using joint dictionary learning in fdd massive mimo systems," in {\it 2015 IEEE Global Conference on Signal and Information Processing (GlobalSIP)}. IEEE, 2015, pp. 185-189,

\bibitem{3gpp.25.996} 3GPP TS 25.996. ``Universal Mobil Telecommunications System (UMTS); Spatial channel model for Multiple Input Multiple Output (MIMO) simulations," version 12.0.0 Release 12, Sep 2014. URL: http://www.3gpp.org.

\bibitem{elad2007optimized} M.Elad, ``Optimized projections for compressed sensing," {\it IEEE Trans. Signal Process.,} vol. 55, no. 12, pp. 5695-5702, 2007.

\bibitem{duarte2009learning} J.M.Duarte-Carvajalino and G.Sapiro, ``Learning to sense sparse signals: Simultaneous sensing matirx and sparsifying dictionary optimization," {\it IEEE Trans. Image Process.,} vol .18, no. 7, pp. 1395-1408, 2009.

\bibitem{yu2011measurement} Y.Yu, A.P.Petropulu, and H.V.Poor, ``Measurement matrix design for compressive sensing-based mimo radar," {\it IEEE Trans. Signal Process.,} vol. 59, no. 11, pp. 5338-5352, 2011.

\bibitem{malloy2014near} M.L.Malloy and R.D.Nowak, ``Near-optimal adaptive compressed sensing," {\it IEEE Trans.   Inf. Theory, ,} vol. 60, no. 7, pp. 4001-4012, 2014.

\bibitem{haupt2012sequentially} J.Haupt, R. Baraniuk, R.Castro, and R.Nowak, ``Sequentially designed compressed sensing," in {\it Statistical Signal Processing Workshop (SSP)}, 2012, pp.401-404.

\bibitem{qi} C. Qi and W. Wu, ``A study of deterministic pilot allocation for sparse channel estimation in OFDM systems," {\it IEEE Commun. Letter}, vol. 16, no. 5, pp. 742-744, May 2012.
	
\bibitem{lasso}
R. Tibshirani, ``Regression shrinkage and selection via the lasso," {\it Journal of the Royal Statistical Society, Series B}, vol. 58, no. 1, pp. 267-288, 1996.

\bibitem{reweighted_candes} E. Candes, M. Wakin, and S. Boyd, ``Enhancing sparsity by reweighted $\ell_1$ minimization," {\it J. Fourier Anal. Appl.}, vol. 14, pp. 877-905, 2008

\bibitem{cosamp} D. Needell and J. A. Tropp, ``CoSaMP: iterative signal recovery from incomplete and inaccurate samples," {\it Appl. Comput. Harmon. Anal.}, vol. 26, pp. 301-321, 2009.

\bibitem{kwon} S. Kwon, J. Wang, and B. Shim, ``Multipath matching pursuit,'' {\it IEEE Trans. Inf. Theory}, vol. 60, no. 5, pp. 2986-3001, May 2014.

\bibitem{iht} T. Blumensath and M. E. Davies, ``Iterative hard thresholding for compressed sensing," {\it Appl. Comput. Harmon. Anal.}, vol. 27, no. 3, pp. 265-274, 2009.

\bibitem{maleki} A. Maleki and D. L. Donoho, ``Optimally tuned iterative reconstruction algorithms for compressed sensing," {\it IEEE Journal of Selected Topics in Signal Process.}, vol. 4, no. 2, pp, 330-341, April 2010

\bibitem{amp1} D. L. Donoho, A. Maleki, and A. Montanari, ``Message passing algorithms for compressed sensing," {\it Proc. of the National Academy of Sciences}, vol. 106, pp. 18914-18919, Nov. 2009.

\bibitem{amp2} M. Bayati and A. Montanari, ``The dynamics of message passing on dense graphs, with applications to compressed sensing," {\it IEEE Trans. Inf. Theory}, vol. 57, pp. 764-785, Feb. 2011.
	
\bibitem{gorodnitsky1997sparse} I. F. Gorodnitsky and B. D. Rao, ``Sparse signal reconstruction from limited data using FOCUSS: A re-weighted minimum norm algorithm," {\it IEEE Trans. Signal Process.}, vol. 45, no. 3, pp. 600-616, March 1997.

\bibitem{candes2008enhancing} E. J. Candes, M. B. Wakin, and S. P. Boyd, ``Enhancing sparsity by reweighted $\ell_1$ minimization," {\it Fourier anal. and appl.} vol. 14, no. 5-6, pp. 877-905, Dec. 2008.

\bibitem{chartrand2008iteratively} R. Chartrand and W. Yin, ``Iteratively reweighted algorithms for compressive sensing," {\it IEEE Proc. ICASSP 2008}, pp. 3869-3872, March 2008.

\bibitem{jeffrey} M. Figueiredo, ``Adaptive sparseness using Jeffreys prior," {\it Advances in Neural Information Processing Systems (NIPS)}, 2002.

\bibitem{bcs} S. Ji, Y. Xue, and L. Carin, ``Bayesian compressed sensing," {\it IEEE Trans. Signal Process.}, vol. 56, pp. 2346-2356, June 2008.

\bibitem{tipping2001sparse} M. E. Tipping, ``Sparse Bayesian learning and the relevance vector machine," {\it The journal of machine learning research}, vol. 1, pp. 211-244, 2001.

\bibitem{wipf2004sparse} D. P. Wipf and B. D. Rao, ``Sparse Bayesian learning for basis selection," {\it IEEE Trans. Signal Process.}, vol. 52, no. 8, pp. 2153-2164, August 2004.

\bibitem{mmwave_mmv} Y. Y. Lee, C. Wang, and Y. H. Huang, ``A Hybrid RF/Baseband precoding processor based on parallel-index-selection matrix-inversion-bypass simultaneous orthogonal matching pursuit for millimeter wave MIMO systems," {\it IEEE Trans. Signal Process.}, vol. 63, no. 2, pp. 305-317, Jan. 2015.

\bibitem{vetterli}   Y. Barbotin, A. Hormati, S. Rangan, and M. Vetterli, ``Estimation of sparse MIMO channels with common support," {\it IEEE Trans. Commun.}, vol. 60, no. 12, pp. 3705-3716, Dec. 2012.

\bibitem{sptp} Z. Gao, L. Dai, W. Dai, B. Shim, and Z. Wang, ``Structured compressive sensing-based spatio-temporal joint channel estimation for FDD massive MIMO," {\it IEEE Trans. Commun.}, vol. 64, no. 2, pp. 601-617, Feb. 2016.

\bibitem{rao_temporal} Z. Zhang and B. D. Rao, ``Sparse signal recovery with temporally correlated source vectors using sparse Bayesian learning," {\it IEEE Journal of Selected Topics in Signal Processing}, vol. 5, pp. 912-926, Sept. 2011.

\bibitem{rao_ar} Z. Zhang and B. D. Rao, ``Sparse signal recovery in the presence of correlated multiple measurement vectors," {\it Proc. ICASSP 2010}, Dallas, TX USA, 2010, pp. 3986-3989.



\bibitem{choi_shim} J. W. Choi and B. Shim, ``Statistical recovery of simultaneously sparse time-varying signals from multiple measurement vectors," {\it IEEE Trans. Signal Process.}, vol. 63, no. 22, pp. 6136-6148, Nov. 2015.

\bibitem{tropp2006algorithms} J. A. Tropp, A. C. Gilbert, and M. J. Strauss, ``Algorithms for simultaneous sparse approximation. Part I: greedy pursuit," {\it Signal Process.}, vol. 86, no. 3, pp. 572-588, March 2006.

\bibitem{tropp_part2}   J. A. Tropp, A. C. Gilbert, M. J. Strauss, ``Algorithms for simultaneous sparse approximation. Part II: convex relaxation," {\it Signal Process.}, vol. 86, pp. 589-602, 2006.

\bibitem{rao_mmv} S. F. Cotter, B. D. Rao, K. Engan, K. Kreutz-Delgado, ``Sparse solutions to linear inverse problems with multiple measurement vectors," {\it IEEE Trans. Signal Process.}, vol. 53, pp. 2477-2488, July 2005.

%\bibitem{kim} J. Kim, O. Lee, and J. Ye, ``Compressive MUSIC: revisiting the link between compressive sensing and array signal processing," {\it IEEE Trans. Inf. Theory}, vol. 58, no. 1, pp. 278-301, 2012.

\bibitem{kiryung} K. Lee, Y. Bresler, and M. Junge, ``Subspace methods for joint sparse recovery," {\it IEEE Trans. Inf. Theory}, vol. 58, no. 6, pp. 3613-3641, 2012.

\bibitem{kalmancs} N. Vaswani, ``Kalman filtered compressed sensing," {\it Proc. ICIP } San Diego, CA USA, 2008, pp. 893-896.

\bibitem{schniter} J. Ziniel and P. Schniter, ``Efficient high-dimensional inference in the multiple measurement vector problem," {\it IEEE Trans. Signal Process.}, vol. 61, pp. 340-354, Jan. 2013.

\bibitem{modifiedcs} N. Vaswani and W. Lu, ``Modified-CS: modifying compressive sensing for problems with partially known support," {\it IEEE Trans. Signal Process.}, vol. 58, no. 9, pp. 4595-4607, Sept. 2010.

\bibitem{schniter2} J. Ziniel and P. Schniter, ``Dynamic compressive sensing of time-varying signals via approximate message passing," {\it IEEE Trans. Signal Process.}, vol. 61, no. 21, pp. 5270-5284, Nov. 2015.

\bibitem{mmv_proof} Y. C. Eldar and H. Rauhut, ``Average case analysis of multichannel sparse recovery using convex relaxation," {\it IEEE Trans. Inf. Theory}, vol. 56, no. 1, pp. 505-519, Jan. 2010.

\bibitem{mmv_proof2} R. Gribonval, H. Rauhut, K. Schnass, and P. Vandergheynst, ``Atoms of all channels, unite! Average case analysis of multi-channel sparse recovery using greedy algorithms," {\it J. Fourier Anal. Appl.}, vol. 14, no. 5, pp. 655-687, 2008.

\bibitem{rahut} Y. C. Eldar and H. Rauhut, ``Average case analysis of multichannel sparse recovery using convex relaxation," {\it IEEE Trans. Inf. Theory}, vol. 6, no. 1, pp. 505-519, Jan. 2010.

\bibitem{wipf2007empirical} D. P. Wipf and B. D. Rao, ``An empirical Bayesian strategy for solving the simultaneous sparse approximation problem," {\it IEEE Trans. Signal Process.} vol. 55, no. 7, pp. 3704-3716, July 2007.

\bibitem{duarte2005distributed} M. F. Duarte, S. Sarvotham, and D. Baron, M. B. Wakin, and R. G. Baraniuk, ``Distributed compressed sensing of jointly sparse signals," {\it Proc. Asilomar Conf. Signals, Sys. Comput.}, pp. 1537-1541, May 2005.

\bibitem{ding2015joint} Y. Ding and B. D. Rao, ``Joint dictionary learning and recovery algorithms in a jointly sparse framework," in {\it 2015 49th Asilomar Conference on Signals, Systerms and Computers.} IEEE, 2015, pp. 1482-1486.

\bibitem{ji2009multitask} S. Ji, D. Dunson, and L. Carin, ``Multitask compressive sensing," {\it IEEE Trans. Signal Process.}, vol. 57, no. 1, pp. 92-106, Jan. 2009.

\bibitem{CV} R. Ward, ``Compressed sensing with cross validation,'' {\it IEEE Trans. Inf. Theory}, vol. 55, no. 12, pp.\ 5773-5782, Dec. 2009.



%\bibitem{3gpp211} 3GPP TS 36.211 V12.0.0 (2013-12): ``Evolved Universal Terrestrial Radio Access (E-UTRA); Physical Channels and Modulation (Release 12)"
%\bibitem{wan} F. Fazel, M. Fazel, and M. Stojanovic, ``Random access compressed sensing for energy-efficient underwater sensor networks," {\it IEEE Selected Areas of Commun.}, vol. 29, no. 9, pp. 1660-1670, Sept. 2011.

%\bibitem{rauhut2008compressed} H. Rauhut, K. Schnass, and P. Vandergheynst, ``Compressed sensing and redundant dictionaries," {\it IEEE Trans. Inf. Theory}, vol. 54, no. 5, pp. 2210-2219, May 2008.
%\bibitem{zhu} H. Zhu and G. B. Giannakis, ``Exploiting sparse user activity in multiuser detection," {\it IEEE Trans. Commun.}, vol. 59, no. 2, pp. 454-465, Feb. 2011.
%\bibitem{msbl} D. P. Wipf and B. D. Rao, ``An empirical Bayesian strategy for solving the simultaneous sparse approximation problem," {\it IEEE Trans. Signal Proc.}, vol. 55, no. 7, pp. 3704-3716, July 2007.
%\bibitem{amp_donoho} D. L. Donoho, A. Maleki, and A. Montanari, ``Message passing algorithms for compressed sensing," {\it in Proc. Nat. Acad. Sci.,}, vol. 106, pp. 18914-18919, Nov. 2009.
%\bibitem{malioutov2005sparse} D. Malioutov, M. {\c{C}}etin, and A.S. Willsky, "A sparse signal reconstruction perspective for source localization with sensor arrays" {\it IEEE Trans. Signal Process.}, vol. 53, no. 8, pp. 3010-3022, 2005.
%\bibitem{kim2012compressive} J.M. Kim, O.K. Lee and J.C. Ye, "Compressive MUSIC: revisiting the link between compressive sensing and array signal processing", {\it IEEE Trans.  Information Theory}, vol. 58, no. 1, pp. 278-301, 2012.
%\bibitem{lee2012subspace} K. Lee, Y. Bresler and M. Junge, "Subspace methods for joint sparse recovery", {\it IEEE Trans. Information Theory}, vol. 58, no. 6, pp. 3613-3641, 2012.
%\bibitem{tang2013compressed} G. Tang, B.N. Bhaskar, P. Shah and B. Recht, "Compressed sensing off the grid", {\it IEEE Trans. Information Theory}, vol. 59, no. 11, pp. 7465-7490, 2013.
%\bibitem{yang2013off} Z. Yang, L. Xie and C. Zhang, "Off-grid direction of arrival estimation using sparse Bayesian inference", {\it IEEE Trans. on Signal Process.}, vol. 61, no. 1, pp. 38-43, 2013.
%\bibitem{tan2014joint} Z. Tan, P. Yang and A. Nehorai, "Joint sparse recovery method for compressed sensing with structured dictionary mismatches", {\it IEEE Trans. Signal Process.}, vol. 62, no. 19, pp. 4997-5008, 2014.
%\bibitem{zhu2011sparsity} H. Zhu, G. Leus and G.B. Giannakis, "Sparsity-cognizant total least-squares for perturbed compressive sampling", {\it IEEE Trans.  Signal Process.}, vol. 59, no. 5, pp. 2002-2016, 2011.
%\bibitem{admm} S. Boyd, N. Parikh, E. Chu, B. Peleato, and J. Eckstein, ``Distributed optimization and statistical learning via the alternating direction method of multipliers," {\it Foundations and Trends in Machine Learning}, vol. 3, no. 1, pp. 1-122, 2010.
%\bibitem{bajwa2010compressed} W.U. Bajwa, J. Haupt, A.M. Sayeed, and R. Nowak. "Compressed channel sensing: A new approach to estimating sparse multipath channels." {\it Proceedings of the IEEE}, vol 98, no.6, pp.1058-1076, 2010.


\end{thebibliography}
\end{document}